\documentclass[sigconf]{acmart}

\usepackage[final]{optional}

\usepackage{balance}

\def\BibTeX{{\rm B\kern-.05em{\sc i\kern-.025em b}\kern-.08emT\kern-.1667em\lower.7ex\hbox{E}\kern-.125emX}}
    
\opt{sub,full}{
    \usepackage[colorlinks=true,citecolor=blue,linkcolor=blue]{hyperref}
    \usepackage{cite}
    \usepackage{authblk}
    \usepackage{wrapfig}
    \usepackage{fullpage}
    \usepackage{titlesec}
}

\usepackage[utf8]{inputenc}
\usepackage{amsmath,amsthm,amssymb,amsfonts}
\usepackage{physics}
\usepackage{enumerate}
\usepackage{mathtools}

\usepackage{graphicx}
\usepackage{bbm} \newcommand{\1}{\mathbbm{1}}
\usepackage{minibox}

\usepackage{graphicx}
\usepackage{subcaption}

\usepackage[textsize=scriptsize]{todonotes} \newcommand{\todoi}[1]{\todo[inline]{#1}}
\newtheorem{theorem}{Theorem}
\newtheorem{definition}[theorem]{Definition}
\newtheorem{lemma}[theorem]{Lemma}
\newtheorem{corollary}[theorem]{Corollary}
\newtheorem{observation}[theorem]{Observation}
\numberwithin{theorem}{section}
\renewcommand\vec{\mathbf}
\DeclarePairedDelimiter\ceil{\lceil}{\rceil}
\DeclarePairedDelimiter\floor{\lfloor}{\rfloor}
\newcommand{\N}{\mathbb{N}}
\newcommand{\Z}{\mathbb{Z}}
\newcommand{\R}{\mathbb{R}}
\newcommand{\Q}{\mathbb{Q}}
\newcommand{\restrictf}{f_{[\vec{x}(i) \to j]}}

\newcommand{\barx}{\overline{\vec{x}}}
\newcommand{\bara}{\overline{\vec{a}}}
\newcommand{\baru}{\overline{\vec{u}}}
\newcommand{\rec}{\mathrm{recc}}

\newcommand{\restateableLemma}[4]{
    \begin{lemma} \label{#1}
        #3
    \end{lemma}
    \expandafter\newcommand\csname restate#2\endcsname{
        \noindent{\bf Lemma~\ref{#1}.}
        \emph{#3}
    }
    \expandafter\newcommand\csname proof#2\endcsname{
        \begin{proof}
            #4
        \end{proof}
    }
}

\newcommand{\restateableTheorem}[4]{
    \begin{theorem} \label{#1}
        #3
    \end{theorem}
    \expandafter\newcommand\csname restate#2\endcsname{
        \noindent{\bf Theorem~\ref{#1}.}
        \emph{#3}
    }
    \expandafter\newcommand\csname proof#2\endcsname{
        \begin{proof}
            #4
        \end{proof}
    }
}

\newcommand{\restateableObservation}[4]{
    \begin{observation} \label{#1}
        #3
    \end{observation}
    \expandafter\newcommand\csname restate#2\endcsname{
        \noindent{\bf Observation~\ref{#1}.}
        \emph{#3}
    }
    \expandafter\newcommand\csname proof#2\endcsname{
        \begin{proof}
            #4
        \end{proof}
    }
}

\begin{document}

\copyrightyear{2019} 
\acmYear{2019} 
\setcopyright{acmcopyright}
\acmConference[PODC '19]{2019 ACM Symposium on Principles of Distributed Computing}{July 29-August 2, 2019}{Toronto, ON, Canada}
\acmBooktitle{2019 ACM Symposium on Principles of Distributed Computing (PODC '19), July 29-August 2, 2019, Toronto, ON, Canada}
\acmPrice{15.00}
\acmDOI{10.1145/3293611.3331615}
\acmISBN{978-1-4503-6217-7/19/07}

\title{Composable Computation in Discrete Chemical Reaction Networks}

\author{Eric E. Severson}
\email{eseverson@ucdavis.edu}
\orcid{1234-5678-9012}
\affiliation{%
  \institution{University of California, Davis}
  \streetaddress{One Shields Ave}
  \city{Davis}
  \state{California}
  \postcode{95616-5200}
}

\author{David Haley}
\email{drhaley@ucdavis.edu}
\affiliation{%
  \institution{University of California, Davis}
  \streetaddress{One Shields Ave}
  \city{Davis}
  \state{California}
  \postcode{95616-5200}
}

\author{David Doty}
\email{doty@ucdavis.edu}
\affiliation{%
  \institution{University of California, Davis}
  \streetaddress{One Shields Ave}
  \city{Davis}
  \state{California}
  \postcode{95616-5200}
}

%

%
\begin{abstract}
We study the composability of discrete chemical reaction networks (CRNs) that \emph{stably compute}
(i.e., with probability 0 of error)
integer-valued functions $f:\N^d\to\N$.
We consider \emph{output-oblivious} CRNs in which the output species is never a reactant (input) to any reaction.
The class of output-oblivious CRNs is fundamental,
appearing in earlier studies of CRN computation,
because it is precisely the class of CRNs that can be composed by simply renaming the output of the upstream CRN to match the input of the downstream CRN.

Our main theorem precisely characterizes the functions $f$ stably computable by output-oblivious CRNs with an initial leader.
The key necessary condition is that for sufficiently large inputs, $f$ is the minimum of a finite number of nondecreasing \emph{quilt-affine} functions.
(An affine function is linear with a constant offset; a \emph{quilt-affine} function is linear with a periodic offset).
\end{abstract}


%

%

%
\maketitle


\section{Introduction}
\label{sec:intro}



A foundational model of chemistry commonly used in natural sciences is that of chemical reaction networks (CRNs): 
finite sets of chemical reactions such as $A + B \to A + C$.
The model is described as a continuous time, discrete state, Markov process~\cite{Gillespie77}.
A configuration of the system is a vector of non-negative integers specifying the molecular counts of the species (e.g., $A$, $B$, $C$), a reaction can occur only when all its reactants are present, and transitions between configurations correspond to reactions 
(e.g., when the above reaction occurs the count of $B$ is decreased by $1$ and the count of $C$ increased by $1$).
CRNs are widely used to describe natural biochemical systems such as the intricate cellular regulatory networks responsible for the information processing within cells.
Looking beyond the \emph{scientific} goal of understanding natural CRNs, 
to the \emph{engineering} goal of constructing programmable, autonomous smart molecules, artificial CRNs have been implemented using
the physical primitive of nucleic-acid strand displacement cascades~\cite{SolSeeWin10, cardelli2011strand, chen2013programmable, srinivas2017enzyme}.


Population protocols,
a widely-studied model of distributed computing with very limited agents,
are a restricted subset of CRNs 
(those with two reactants and two products in each reaction) 
that nevertheless capture many of the interesting features of CRNs.
The key feature is the inability of agents (molecules) to control their schedule of communication (collisions).
The decision problems solvable by population protocols have been studied extensively:
they can simulate Turing machines with high probability in polylogarithmic time
(with~\cite{angluin2006fast} or without~\cite{kosowski2018brief} an initial leader),
whereas requiring probability 0 of error limits the computable predicates to being semilinear~\cite{angluin2007computational}.


\subsection{Function computation}

Computation of functions $f:\N^d \to \N$ was discussed briefly in the first population protocols paper \cite[Section 3.4]{angluin2006passivelymobile},
which focused more on Boolean predicate computation,
and it was defined formally first for CRNs~\cite{CheDotSolDetFuncNaCo, DotHajLDCRNNaCo}
and later for population protocols~\cite{belleville2017hardness}.
The class of functions stably computable in either model is the same: 
the semilinear functions~\cite{angluin2007computational, CheDotSolDetFuncNaCo}.
We use the CRN model because it is more natural for describing functions,
but our results also apply to the population protocol model.

To represent an input $\vec{x} \in \N^d$,
we start in a configuration with counts $\vec{x}(i)$ of species $X_i$ for each $i \in \{1,\ldots,d\}$,
and count 1 of a ``leader'' species $L$.\footnote{
    The  leader is discussed in Section~\ref{sec:role-of-leader}.
    A CRN may ignore its leader, as in Fig.~\ref{fig:CRN-computing-examples}.
}
A function $f:\N^d\to\N$ is \emph{stably computable} by a CRN if a correct 
and \emph{stable} configuration $\vec{O}$
(i.e., on input $\vec{x}$ the count of $Y$ is $f(\vec{x})$ in all configurations reachable from $\vec{O}$)
remains reachable no matter what reactions occur.\footnote{
    We use this definition throughout the paper,
    but we mention here that it is equivalent to two other natural definitions.
    The first definition is that any fair sequence of reactions \emph{will} take the CRN to such a correct stable configuration,
    where \emph{fair} means that any configuration that is infinitely often reachable is eventually reached.
    The second definition is that a correct stable configuration is actually reached with probability 1.
}
\begin{figure}
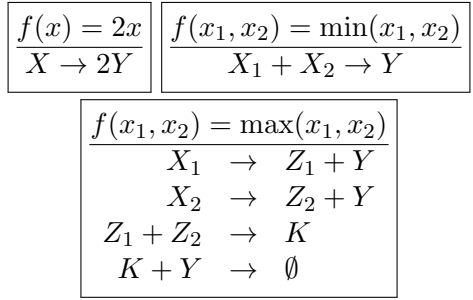
{}
    \opt{full}{\centering}
    \begin{tabular}{cc}
        \hbox{
            \minibox[c,frame]{
                $f(x) = 2x$
                \\ \hline
                $X \to 2Y$
            }
            \minibox[c,frame]{
                $f(x_1,x_2) = \min(x_1,x_2)$
                \\ \hline
                    $X_1+X_2 \to Y$
            }
        }
    \end{tabular}
        
    \vskip0.05cm
        
    \begin{tabular}{cc}
        \minibox[c,frame]{
            $f(x_1,x_2) = \max(x_1,x_2)$
            \\ \hline
            $\begin{array}{rcl}
            X_1 & \to & Z_1 + Y
            \\
            X_2 & \to & Z_2 + Y
            \\
            Z_1+Z_2 & \to & K
            \\
            K+Y & \to & \emptyset
            \end{array}$
        }
    \end{tabular}

    \caption{\footnotesize Functions stably computed by CRNs. Note max is computed as $x_1+x_2-\min(x_1,x_2)$.}
    \label{fig:CRN-computing-examples}
\end{figure}

See Fig.~\ref{fig:CRN-computing-examples} for examples.
It is known that a function $f:\N^d\to\N$ is stably computable by a CRN if and only if it is \emph{semilinear}:
intuitively, it is a piecewise affine function.
(See Definition~\ref{semilinear function def}.)

\subsection{Composability}
Note a key difference between the CRNs for $\min$ and $\max$ in Fig.~\ref{fig:CRN-computing-examples}:
the former only produces the output species $Y$, 
whereas the latter also contains reactions that consume $Y$.
In one possible sequence of reactions for the $\max$ CRN, the inputs can be exhausted through the first two reactions before ever executing the last two reactions.  In doing so,
the count of $Y$ overshoots its correct value of $\max(x_1,x_2)$
before 
the excess 
is consumed by the reaction $K+Y \to \emptyset$.

For this reason that the $\min$ CRN is more easily composed with a downstream CRN.
For example, the function $2 \cdot \min(x_1, x_2)$ is stably computed by the reactions $X_1 + X_2 \to W$ 
(computing $w = x_1+x_2$)
and $W \to 2Y$
(computing $y = 2w$), 
renaming the output of the $\min$ CRN to match the input of the multiply-by-2 CRN.
However, this approach does not work to compute $2 \cdot \max(x_1,x_2)$;
changing $Y$ to $W$ in the four-reaction $\max$ CRN and adding the reaction $W \to 2Y$ can erroneously result in up to $2(x_1 + x_2)$ copies of $Y$ being produced.
Intuitively,
the multiply-by-2 reaction $W \to 2Y$ competes with the upstream reaction $K+W \to \emptyset$ from the $\max$ CRN.

This motivates us to study the class of functions $f:\N^d\to\N$ stably computable by \emph{output-oblivious} CRNs:
those in which the output species $Y$ appears only as a product,
never as a reactant.
We call such a function \emph{obliviously-computable}.
Any obliviously-computable function must be nondecreasing, 
otherwise reactions could incorrectly overproduce output 
(see Observation \ref{MC is nondecreasing}).

Obliviously-computable functions must also be semilinear, so it is reasonable to conjecture that a function is obliviously-computable if and only if it is semilinear and nondecreasing.
In fact, this is true for 1D functions $f:\N\to\N$
(see Section~\ref{sec:one-dim-case}).
However, in higher dimensions, the function $\max:\N^2\to\N$ is semilinear and nondecreasing, 
yet not obliviously-computable;
its consumption of output turns out to be unavoidable. 
%
Assuming there is no leader,
this is simple to prove:
Since $\max(1,0)=1$, starting with one $X_1$, a $Y$ can be produced.
Similarly, a $Y$ can be produced starting with one $X_2$.
Then with one $X_1$ and one $X_2$, 
these reactions can happen in parallel and produce two $Y$'s, 
too many since $\max(1,1) = 1$. 
It is more involved to prove that even with a leader,
it remains impossible to obliviously compute max; 
see Section~\ref{sec:max}.\footnote{
    This result was obtained independently by Chugg, Condon, and Hashemi~\cite{chugg2018output}. 
}

\subsection{The role of the leader}
\label{sec:role-of-leader}
Our model includes an initial leader, which is essential for our general constructions (see Sections~\ref{sec:one-dim-case} and \ref{Construction}). The class of stably computable functions is identical whether an initial leader is allowed or not~\cite{DotHajLDCRNNaCo},
    as is the class of stably computable predicates~\cite{angluin2007computational}.

    Interestingly,
    the class of obliviously-computable functions we study
    is provably larger when an initial leader is allowed. For example, consider the function $f(x)=\min(1,x)$ (see Fig.~\ref{fig:CRN-reactions-with-leader}). $f$ is stably computable with or without a leader, but only the construction with a leader is output-oblivious. Without using a leader, $f$ is not obliviously-computable 
    (see Observation~\ref{MC0 is superadditive}).
    
\begin{figure}
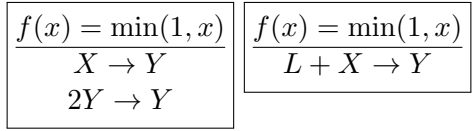
{}
    \opt{full}{\centering}
    \begin{tabular}{cc}
            \minibox[c,t,frame]{
                $f(x) = \min(1,x)$
                \\ \hline
                $X \to Y$
                \\
                $2Y \to Y$
            }
            \minibox[c,t,frame]{
                $f(x) = \min(1,x)$
                \\ \hline
                    $L + X \to Y$
            }
    \end{tabular}
    \caption{\footnotesize $\min(1,x)$ is stably computed by a \emph{leaderless} non-output-oblivious CRN (left), and an output-oblivious CRN with a single leader $L$ (right).}
    \label{fig:CRN-reactions-with-leader}
\end{figure}

Including the leader gives additional power to the model. This gives more power to our CRN constructions,
but makes our impossibility results stronger.
Fully classifying the obliviously-computable functions 
in a leaderless model remains an open question.

\newcommand{\quiltAffineExampleOneDCaption}{ 
    A 1D (single-input) quilt-affine function $\floor{\frac{3x}{2}} = \frac{3}{2}x + B(\overline{x} \mod 2)$,
    where $B(\overline{0}) = 0$ and $B(\overline{1}) = -\frac{1}{2}$.
}
\newcommand{\quiltAffineExampleTwoDCaption}{ 
    A 2D quilt-affine function $g(\vec{x}) = (1,2)\cdot\vec{x} + B(\barx\mod 3),$
    where $B(\barx) = 0$ except when 
    $\barx \in \{\overline{(1,2)}, \overline{(2,2)}, \overline{(2,1)}\}$.
}
\newcommand{\twoRepresentativeCaption}{
    A 2D function satisfying Theorem~\ref{thm:main result}.
}
\newcommand{\scalingLimitCaption}{
    The scaling limit gives a 2D real-valued obliviously-computable function from \cite{ContinuousMC}.
}

\opt{final}{
    \begin{figure*}[t!]
    \centering
    \begin{subfigure}[t]{0.47\textwidth}
        \centering
        \includegraphics[width=0.4\textwidth]{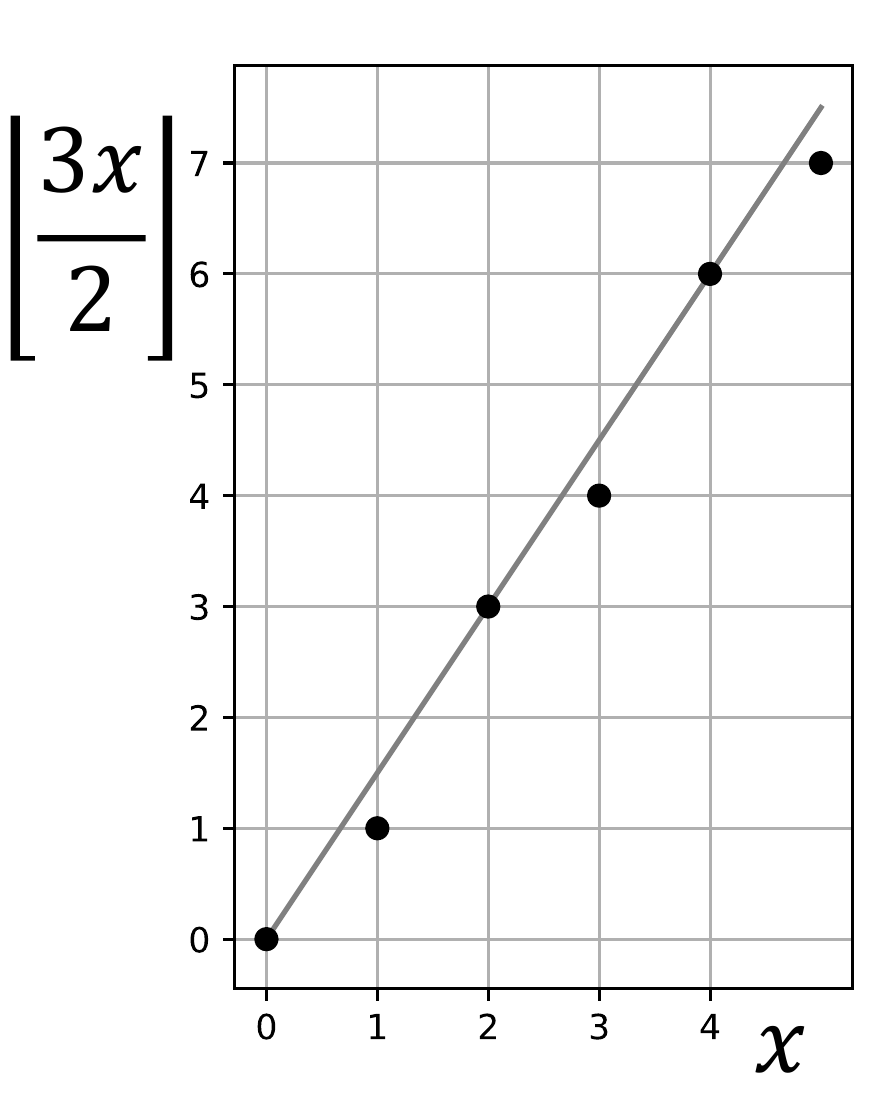}
        \caption{ \footnotesize
            \quiltAffineExampleOneDCaption
        }
        \label{fig:1D-quilt-affine-example}
    \end{subfigure}%
    \qquad
    \begin{subfigure}[t]{0.47\textwidth}
        \centering
        \includegraphics[width=0.8\textwidth]{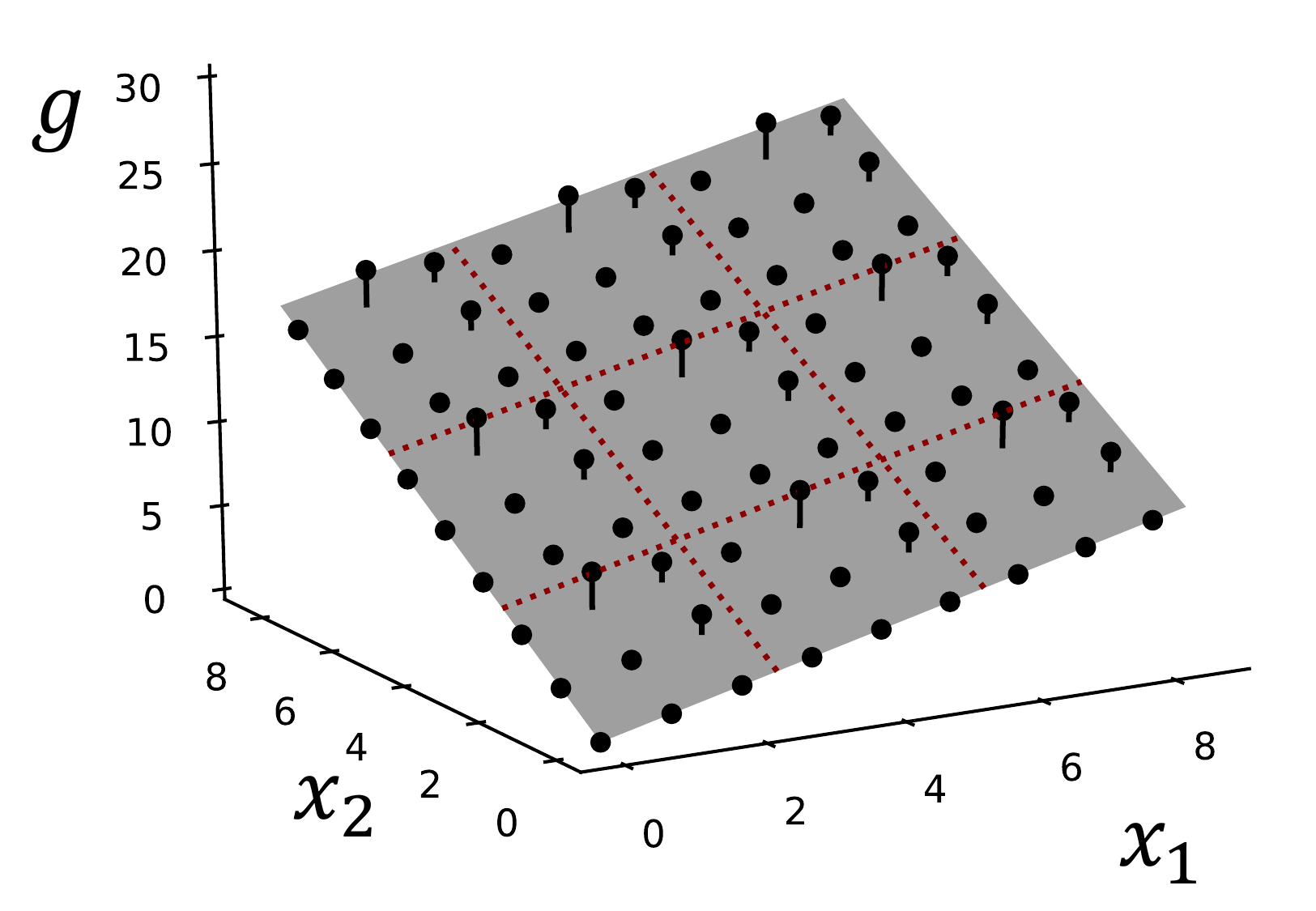}
        \caption{ \footnotesize
            \quiltAffineExampleTwoDCaption
        }
        \label{fig:2D-quilt-affine-example}
    \end{subfigure}%
    \caption{ \footnotesize
        Examples of 1D and 2D quilt-affine functions.
    }
    \end{figure*}

    \begin{figure*}[t!]
    \centering
    \begin{subfigure}[t]{0.47\textwidth}
        \centering
        \includegraphics[width=0.8\textwidth]{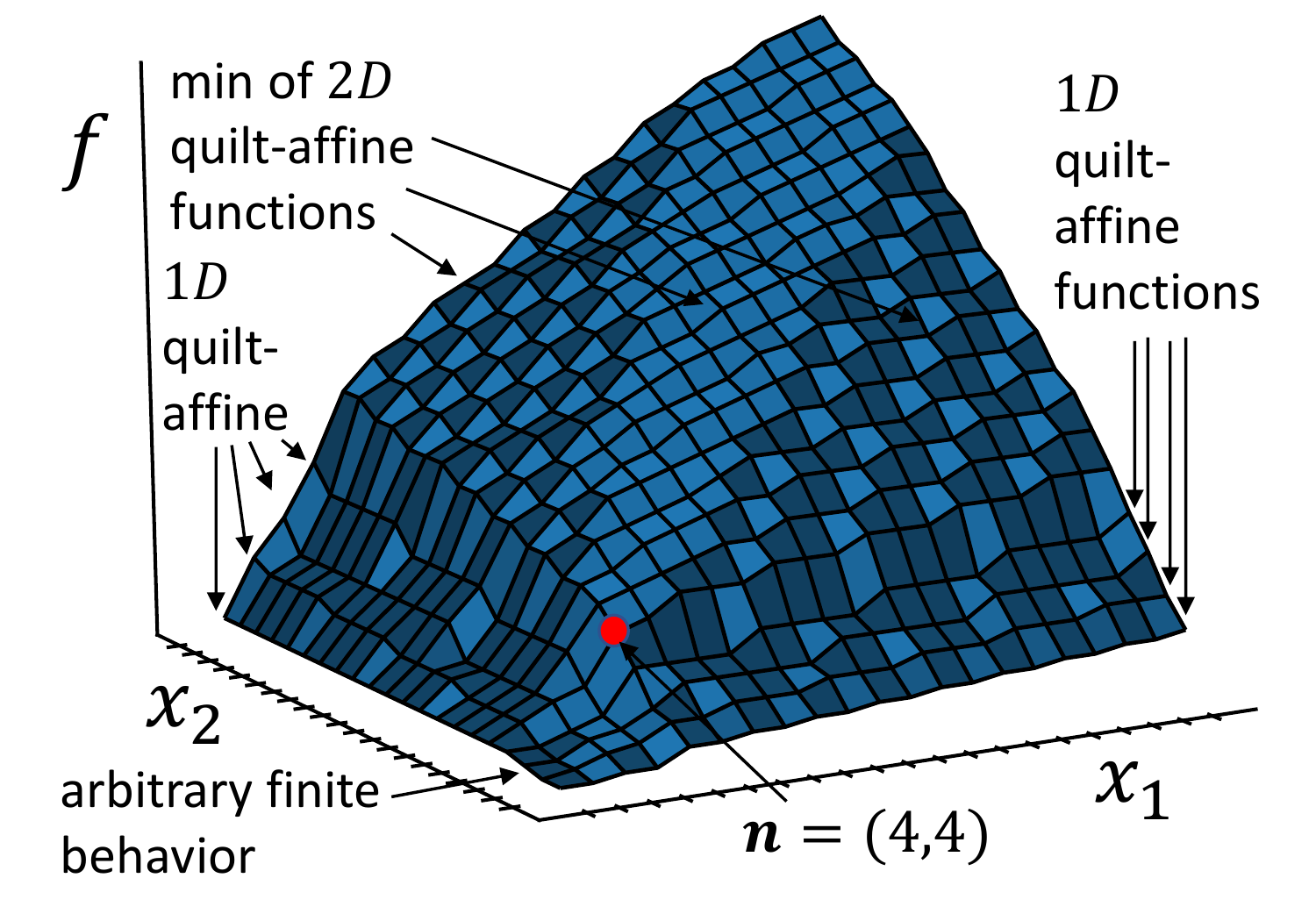}
        \caption{ \footnotesize
            \twoRepresentativeCaption
        }
        \label{fig:2D-obliviously-computable-example}
    \end{subfigure}
    \qquad
    \begin{subfigure}[t]{0.47\textwidth}
        \centering
        \includegraphics[width=0.8\textwidth]{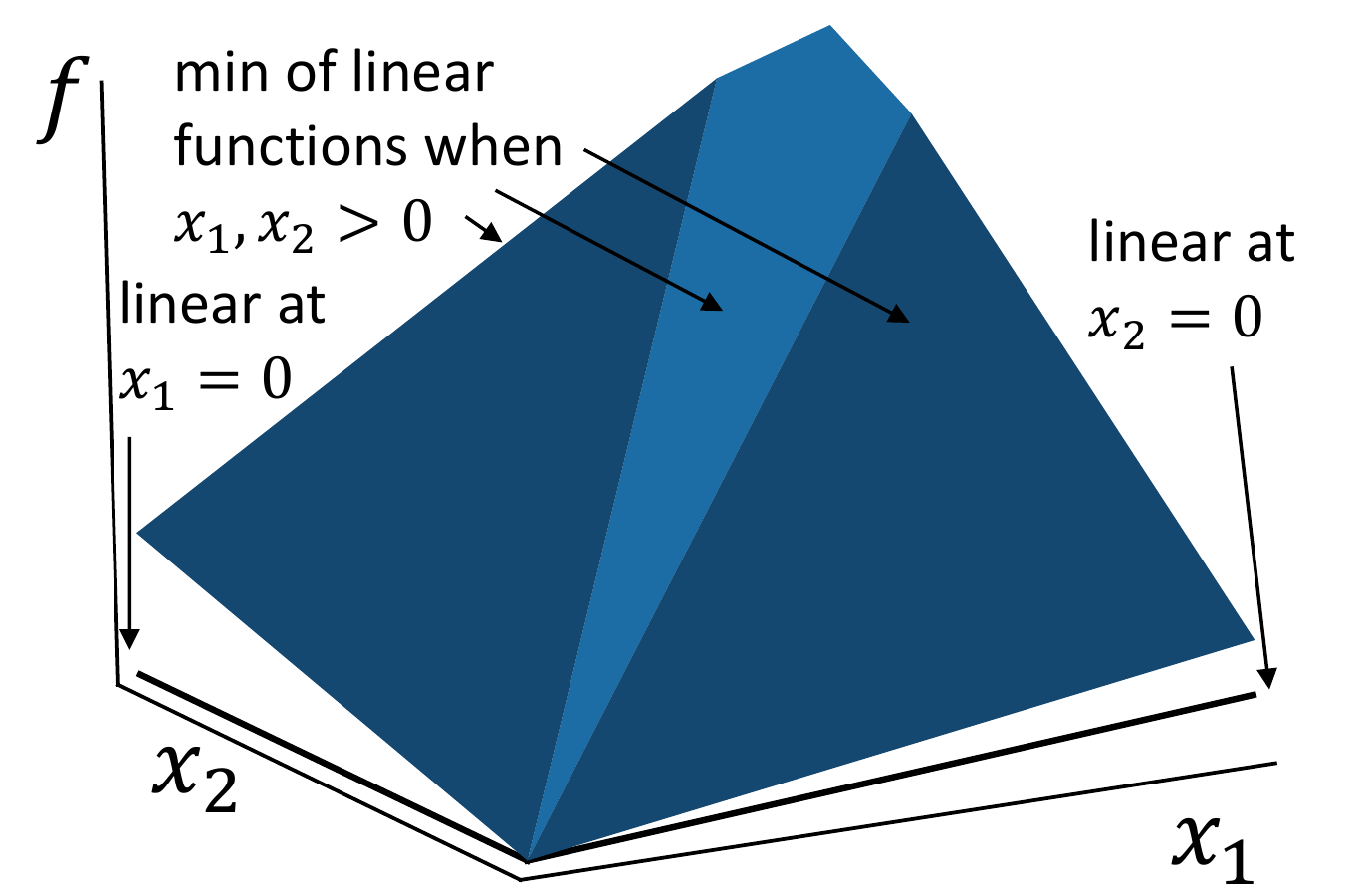}
        \caption{ \footnotesize
            \scalingLimitCaption
        }
        \label{fig:2D-scaling_limit}
    \end{subfigure}
    \caption{ \footnotesize
        Example discrete and real-valued obliviously-computable functions.
    }
    \end{figure*}
}


\opt{full,sub}{
    \begin{figure}[t!]
    \centering
    \begin{subfigure}[t]{0.47\textwidth}
        \centering
        \includegraphics[width=0.4\textwidth]{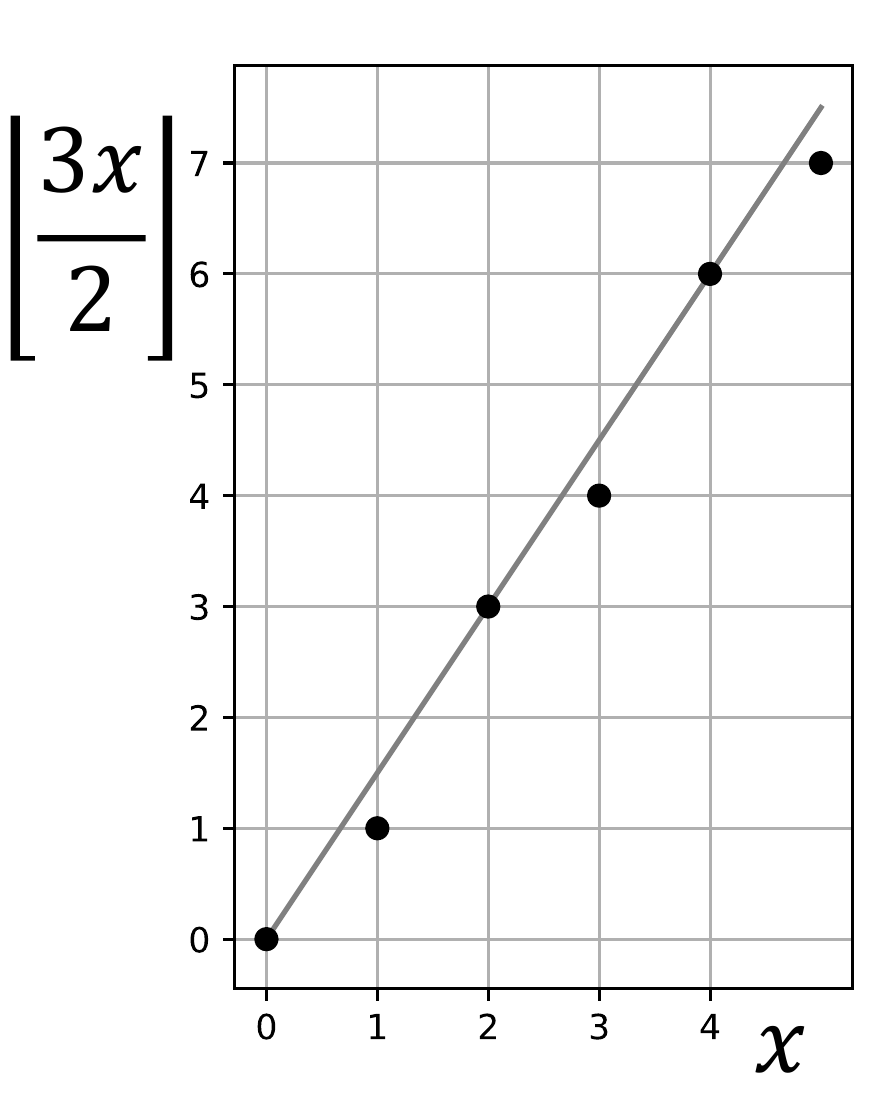}
        \caption{ \footnotesize
            \quiltAffineExampleOneDCaption
        }
        \label{fig:1D-quilt-affine-example}
    \end{subfigure}%
    \qquad
    \begin{subfigure}[t]{0.47\textwidth}
        \centering
        \includegraphics[width=0.8\textwidth]{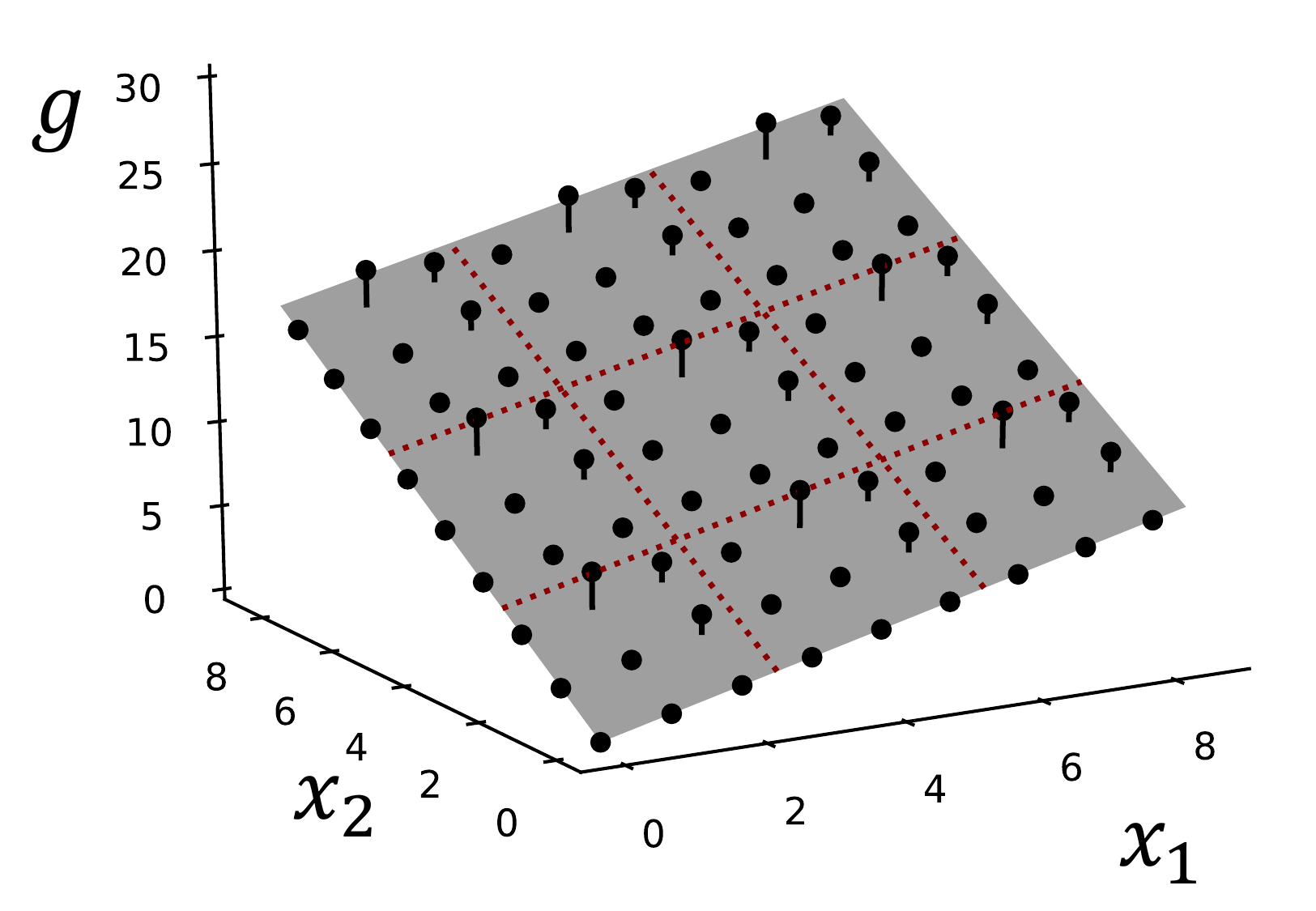}
        \caption{ \footnotesize
            \quiltAffineExampleTwoDCaption
        }
        \label{fig:2D-quilt-affine-example}
    \end{subfigure}%
    \caption{ \footnotesize
        Examples of 1D and 2D quilt-affine functions.
    }
    \end{figure}

    \begin{figure}[t!]
    \centering
    \begin{subfigure}[t]{0.47\textwidth}
        \centering
        \includegraphics[width=0.8\textwidth]{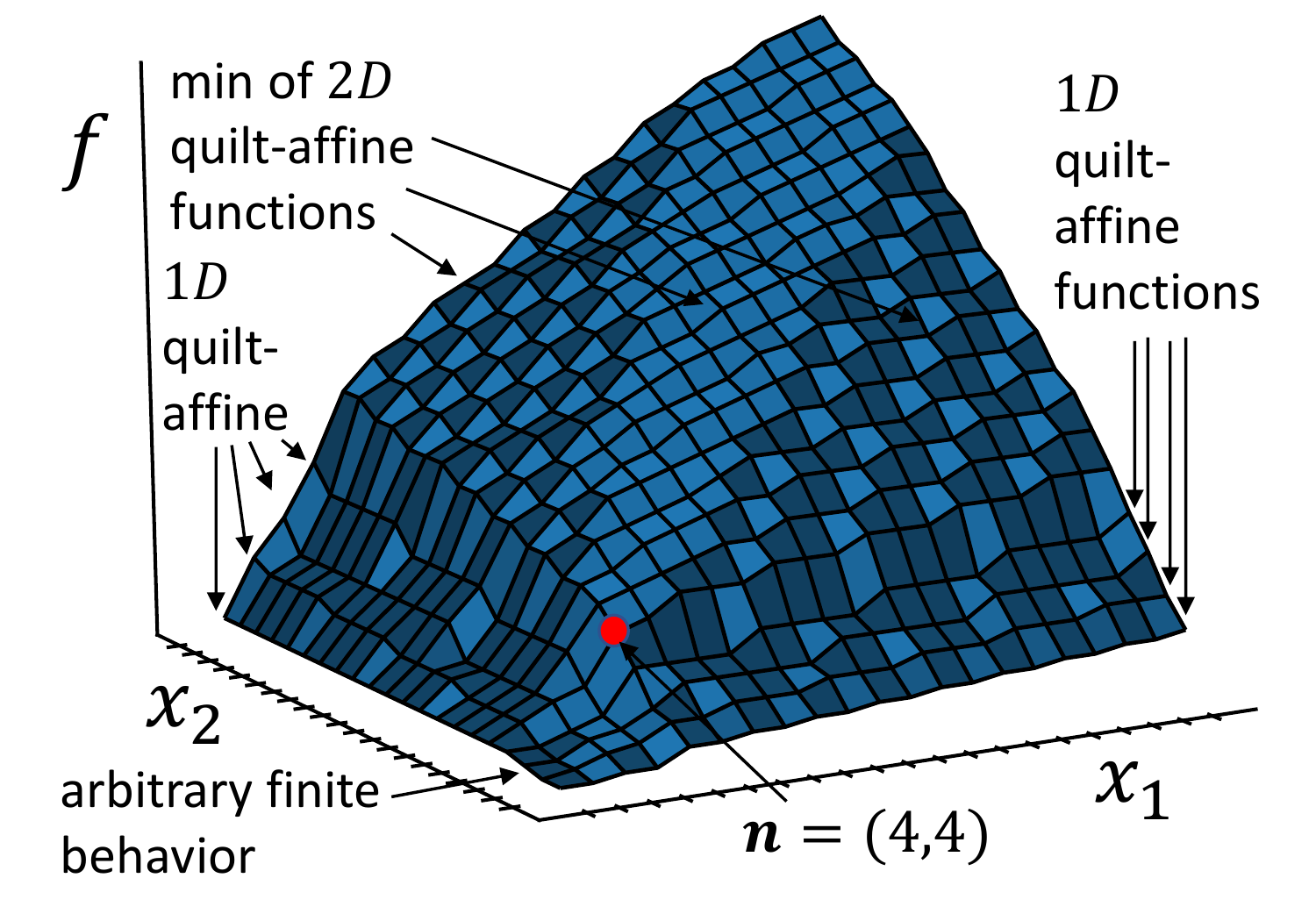}
        \caption{ \footnotesize
            \twoRepresentativeCaption
        }
        \label{fig:2D-obliviously-computable-example}
    \end{subfigure}
    \qquad
    \begin{subfigure}[t]{0.47\textwidth}
        \centering
        \includegraphics[width=0.8\textwidth]{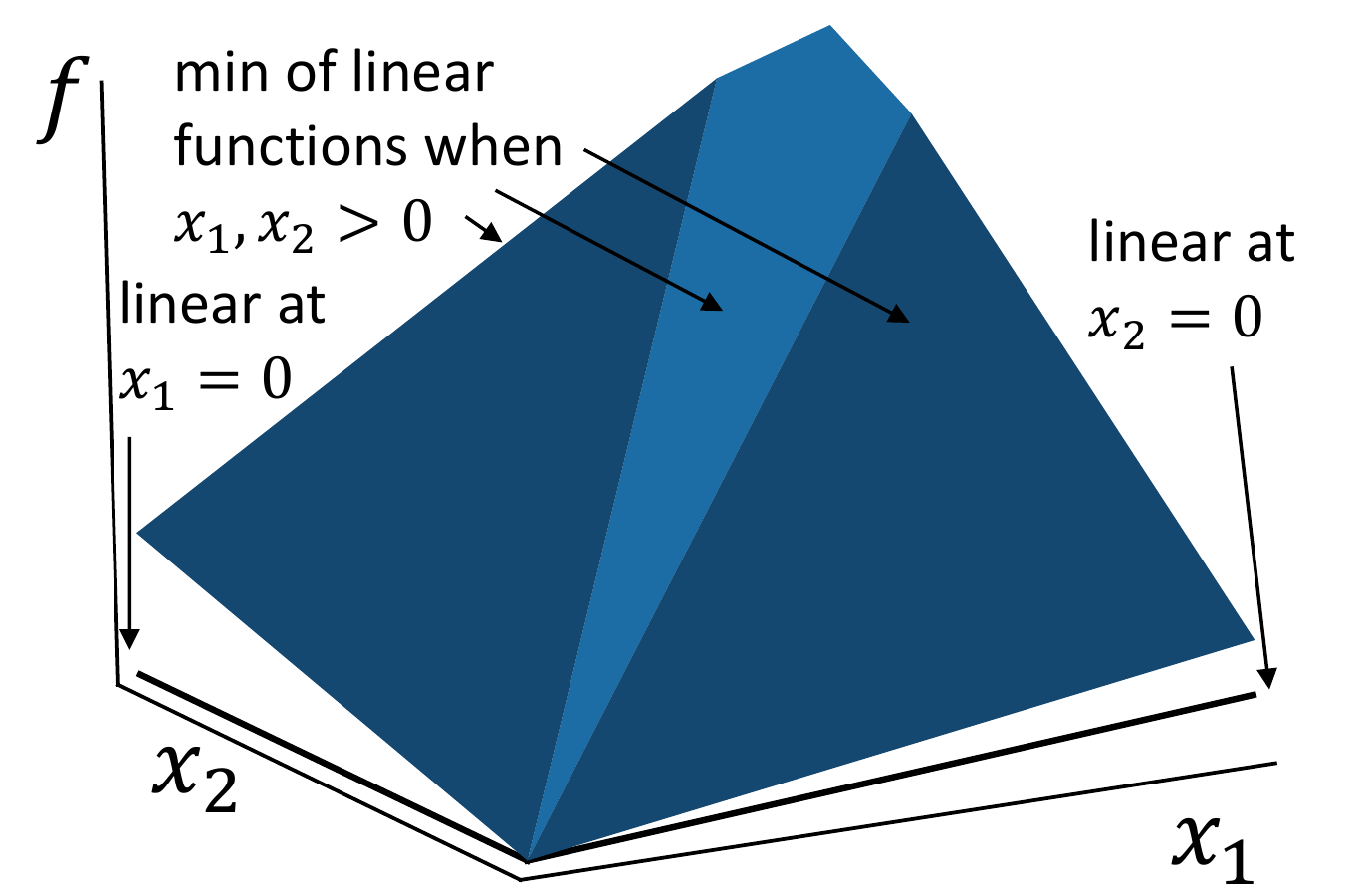}
        \caption{ \footnotesize
            \scalingLimitCaption
        }
        \label{fig:2D-scaling_limit}
    \end{subfigure}
    \caption{ \footnotesize
        Example discrete and real-valued obliviously-computable functions.
    }
    \end{figure}
}

\subsection{Contribution}
Our main result, Theorem~\ref{thm:main result}, 
provides a complete characterization of the class of obliviously-computable functions. 
It builds off a key definition: a $\emph{quilt-affine}$ function is a nondecreasing function that is the sum of a rational linear function and periodic function (formalized as Definition \ref{quilt-affine def}). 
For example, functions such as $\floor{\frac{3x}{2}}$ are quilt-affine (see Fig.~\ref{fig:1D-quilt-affine-example}).
Such floored division functions are natural to the discrete CRN model 
($\floor{\frac{3x}{2}}$ is stably computed by reactions $X\to 3Z$, $2Z\to Y$). Fig.~\ref{fig:2D-quilt-affine-example} shows a higher-dimensional quilt-affine function, 
with a ``bumpy quilt'' structure that motivates the name. 
Quilt-affine functions are also characterized by nonnegative periodic finite differences, a structure key to showing they are obliviously-computable (see Lemma \ref{CRC for quilt-affine function}).

Theorem~\ref{thm:main result} states that a function $f:\N^d\to\N$ is obliviously-computable if and only if
\begin{enumerate}[i)]
\begin{samepage}
    \item\label{defn:intro:1}
    [nondecreasing]
    $f$ is nondecreasing,
    
    \item\label{defn:intro:2}
    [eventually-min]
    for sufficiently large inputs, $f$ is the minimum of a finite number of quilt-affine functions,
    and
    
    \item\label{defn:intro:3}
    [recursive]
    every restriction $f_r:\N^{d-1}\to\N$ obtained by fixing some inputs to a constant value\footnote{
        Note that Theorem~\ref{thm:main result} defines fixed-input restrictions slightly differently;
        see Section~\ref{sec:full-dim-case} for an explanation.
    }
    is obliviously-computable 
    (i.e., eventually the minimum of quilt-affine functions).
\end{samepage}
\end{enumerate}


Condition~\eqref{defn:intro:2} characterizes $f$ when all inputs are sufficiently large 
(greater than some $\vec{n} \in \N^d$),
whereas condition~\eqref{defn:intro:3} characterizes $f$ when some inputs are fixed to smaller values.
See Fig.~\ref{fig:2D-obliviously-computable-example} for a representative example of an obliviously-computable $f:\N^2\to\N$. This pictured function has arbitrary nondecreasing values in the ``finite region'' where $\vec{x}<(4,4)$, has eventual 1D quilt-affine behavior along the lines $x_1=0,1,2,3$ and $x_2=0,1,2,3$, and is the minimum of 3 different quilt-affine functions in the ``eventual region'' where $\vec{x}\geq(4,4)$. 
This behavior generalizes naturally to higher dimensions.

The most technically sophisticated part of our result is the proof that the eventually-min condition~\eqref{defn:intro:2} of Theorem~\ref{thm:main result} is necessary;
the main ideas of this proof are outlined in Section~\ref{sec:forward-direction-proof-outline}.

\subsection{Related work}
Chalk, Kornerup, Reeves, and Soloveichik~\cite{ContinuousMC}
showed an analogous result in the \emph{continuous} model of CRNs in which species amounts are given by nonnegative real concentrations.
A consequence of their characterization is that any obliviously-computable real-valued function
is a minimum of linear functions when all inputs are positive. 
%
In Theorem~\ref{thm-continuous-comparison},
we demonstrate that the limit of ``scalings'' of a function $f:\N^d \to \N$ satisfying our main Theorem~\ref{thm:main result} is in fact a function $\hat{f}:\R_{\geq 0}^d \to \R_{\geq 0}$ satisfying the main theorem of~\cite{ContinuousMC} (see Fig.~\ref{fig:2D-scaling_limit}).
The discrete details lost in the scaling limit constitute precisely the unique challenges of proving Theorem~\ref{thm:main result}
that are not handled by~\cite{ContinuousMC}.
In particular, our function class can contain arbitrary finite behavior and repeated finite irregularities.


Returning to the \emph{discrete} (a.k.a., \emph{stochastic})
CRN model we study,
Chugg, Condon, and Hashemi~\cite{chugg2018output}
independently
investigated the special case of \emph{two-input} functions $f:\N^2 \to \N$ computable by output-oblivious CRNs,
obtaining a characterization equivalent to ours when restricted to 2D.
Their characterization is phrased much differently,
with specially constructed ``fissure functions'' to describe the function behavior across what we describe as \emph{under-determined regions} (intuitively, thin ``1D'' regions bounded by parallel lines, where $f$ cannot be described by a unique quilt-affine function, see Section~\ref{Forward Direction}).
%
The ideas required to prove the 2D case are sophisticated and far from simple,
yet unfortunately, these ideas do not extend straightforwardly to higher dimensions.
The planarity of the 2D input space constrains the regions induced by separating hyperplanes
(i.e., lines)
in a strong way.
Furthermore,
the fact that there is only one nontrivial integer dimension smaller than 2 implies that the under-determined regions are simpler to reason about than in the case where they can have arbitrary dimension between 1 and $d$.
Finally, even restricted to 2D,
a notable aspect of our characterization is expressing $f$ a minimum of quilt-affine functions, which are simple intrinsic building blocks that generalize immediately.

\subsection{Other ways of composing computation}

In Section~\ref{sec:composition-iff-oblivious} we show that a for a CRN $\mathcal{C}$ be composable with downstream CRN $\mathcal{D}$ by ``concatenation'' 
(renaming $\mathcal{C}$'s output species to match $\mathcal{D}$'s input species and ensuring all other species names are disjoint between $\mathcal{C}$ and $\mathcal{D}$),
it is (in a sense) necessary and sufficient for $\mathcal{C}$ to be output-oblivious.
There are other ways to compose computations, however.

A common technique (e.g.~\cite{gasieniec2018fast}) is for $\mathcal{C}$ to detect when its output has changed and send a \emph{restart} signal to $\mathcal{D}$.
However, it is not obvious how to do this with function computation as defined in this paper, where $\mathcal{D}$ changes $\mathcal{C}$'s output by consuming it.

Another technique (e.g.~\cite{angluin2006fast}) is to set a \emph{termination signal},
which is a sub-CRN that, with high probability, creates a copy of a signal species $T$, but not before $\mathcal{C}$ has converged.
$T$ then ``activates'' the reactions of $\mathcal{D}$, so that $\mathcal{D}$ will not consume the output of $\mathcal{C}$ until it is safe to do so.
However, this has some positive failure probability.
In fact, if we require $T$ to be guaranteed with probability 1 to be produced only after the CRN has converged, 
only constant functions can be stably computed.
Worse yet, in the leaderless case, it is provably impossible to achieve this guarantee even with positive probability~\cite{doty2019efficient}.

    
    
\section{Preliminaries}
\label{sec:prelim}

\subsection{Notation}
$\N$ denotes the set of nonnegative integers. 
For a set $\mathcal{S}$ (of species), 
we write $\N^\mathcal{S}$ to denote the set of vectors indexed by the elements of $\mathcal{S}$
(equivalently, functions $f:\mathcal{S}\to\N$). 
Vectors appear in boldface, and we reserve uppercase $\vec{A}\in\N^{\mathcal{S}}$ for such vectors indexed by species, 
and lowercase $\vec{a}\in\N^d,\Z^d,\Q^d,\R^d$ for vectors indexed by integers.
$\vec{A}(S)$ or $\vec{a}(i)$ denotes the element indexed by $S\in\mathcal{S}$ or $i\in\{1,\ldots,d\}$. We write $\vec{a}\leq\vec{b}$ to denote pointwise vector inequality $\vec{a}(i)\leq\vec{b}(i)$ for all $i$.

For $p \in \N_+$,
$\Z/p\Z$ denotes the additive group of integers modulo $p$, whose elements are congruence classes.
Generalizing to higher dimensions, $\Z^d/p\Z^d$ denotes the additive group of $\Z^d$ modulo $p$, whose elements are congruence classes. 
For $\vec{x}\in\N^d$ where $d\geq 1$, we write $\barx\mod p$ to denote the congruence class $\{\vec{x}+p\vec{z}:\vec{z}\in\Z^d\}\in\Z^d/p\Z^d$,
also denoted $\barx$ when $p$ is clear from context.

\opt{full}{
    $\R^d_{\geq{0}}$ denotes the nonnegative orthant in $\R^d$. We consider \emph{regions} $R=\{\vec{x}\in\R^d_{\geq{0}}:T\vec{x}-\vec{h}\geq 0\}$ which are convex polyhedra given by a set of inequalities. $R\cap\N^d$ denotes all integer points in $R$, and for $\vec{x}\in\N^d$, $R\cap(\barx\mod p)$ denotes the integer points in $R$ in the same congruence class as $\vec{x}$.
}

\subsection{Chemical reaction networks}
We use the established definitions of stable function computation by (discrete) chemical reaction networks \cite{StablyComputableAreSemilinear,chugg2018output}:

A \emph{chemical reaction network} (CRN) $\mathcal{C}=(\mathcal{S},\mathcal{R})$ is defined by a finite set $\mathcal{S}$ of species and a finite set $\mathcal{R}$ of reactions, where a reaction $(\vec{R},\vec{P})\in\N^\mathcal{S}\times\N^\mathcal{S}$ describes the counts of consumed \emph{reactant} species and produced \emph{product} species.\footnote{
    We do not limit ourselves to \emph{bimolecular} (two input) reactions, 
    but the higher-order reactions we use can easily be converted to have this form. 
    For example, $3X\to Y$ is equivalent to two reactions $2X\leftrightarrow X_2$ and $X+X_2\to Y$.
}
For example, given $\mathcal{S}=\{A,B,C\}$, the reaction $((1,0,2),(0,2,1))$ would represent $A+2C\to 2B+C$.

A configuration $\vec{C}\in\N^\mathcal{S}$ specifies the integer counts of all species. 
Reaction $(\vec{R},\vec{P})$ is \emph{applicable} to $\vec{C}$ if $\vec{R}\leq\vec{C}$, and yields $\vec{C'}=\vec{C}-\vec{R}+\vec{P}$, 
so we write $\vec{C}\to\vec{C'}$. A configuration $\vec{D}$ is $\emph{reachable}$ from $\vec{C}$ if there exists a finite sequence of configurations such that $\vec{C}\to\vec{C}_1\to\ldots\to\vec{C}_n\to\vec{D}$; we write $\vec{C}\to^*\vec{D}$ to denote that $\vec{D}$ is reachable from $\vec{C}$. Note this reachability relation is additive: if $\vec{A}\to^*\vec{B}$, then $\vec{A}+\vec{C}\to^*\vec{B}+\vec{C}$. This property is key in future proofs to show the reachability of configurations which overproduce output.

To compute a function\footnote{We consider codomain $\N$ without loss of generality, since $f:\N^d\to\N^l$ is stably computable if and only if each output component is stably computable by parallel CRNs.} $f:\N^d\to\N$, the CRN $\mathcal{C}$ will include an ordered subset $\{X_1,\ldots,X_d\}\subset\mathcal{S}$ of \emph{input species}, an \emph{output species} $Y$, and a \emph{leader species} $L\in\mathcal{S}$. (Note that we consider removing the leader in Section~\ref{sec:leaderless}).

The computation of $f(\vec{x})$ will start from an \emph{initial configuration} $\vec{I_x}$ encoding the input with $\vec{I_x}(X_i)=\vec{x}(i)$ for all $i=1,\ldots,d$, along with a single leader $\vec{I_x}(L)=1$, and count $0$ of all other species. A \emph{stable} configuration $\vec{C}$ has unchanged output $\vec{C}(Y)=\vec{D}(Y)$ for any configuration $\vec{D}$ reachable from $\vec{C}$. The CRN $\mathcal{C}$ \emph{stably computes} $f:\N^d\to\N$ if for each initial configuration $\vec{I_x}$ encoding any $\vec{x}\in\N^d$, and configuration $\vec{C}$ reachable from $\vec{I_x}$, there is a stable configuration $\vec{O}$ reachable from $\vec{C}$ with correct output $\vec{O}(Y)=f(\vec{x})$.

\subsection{Composition via output-oblivious CRNs}
\label{sec:composition-iff-oblivious}


This section formally defines our notions of 
``composable computation with CRNs via concatenation of reactions'' 
and 
``output-oblivious'' CRNs that don't consume their output,
showing these notions to be essentially equivalent.

A CRN is \emph{output-oblivious} if the output species $Y$ is never a reactant\footnote{
    A more general definition in \cite{chugg2018output} of \emph{output-monotonic} CRNs just requires no reaction to reduce the count of output species. This can be directly seen to classify the same set of functions, see \opt{full,sub}{Observation~\ref{obs:output-monotonic-is-equivalent}.} \opt{final}{\cite{severson2019composable}.}
}:
for any reaction $(\vec{R},\vec{P})$, $\vec{R}(Y)=0$. 
A function $f:\N^d\to\N$ is \emph{obliviously-computable} if $f$ is stably computed by an output-oblivious CRN.

\opt{sub}{(
    See Appendix \ref{sec:appendix-composability}.)
}

We begin with an easy observation:

\begin{observation}
\label{MC is nondecreasing}
    An obliviously-computable function $f:\N^d\to\N$ must be nondecreasing.
\end{observation}

\begin{proof}
    Assume a CRN $\mathcal{C}$ (with output species $Y$) stably computes $f$, but $f(\vec{a})>f(\vec{b})$ for $\vec{a}\leq\vec{b}$. To stably compute $f(\vec{a})$, input configuration $\vec{I_a}\to^*\vec{O}$ for some configuration with $\vec{O}(Y)=f(\vec{a})$. However, since $\vec{a}\leq\vec{b}$, that same sequence of reactions can be applied from the input configuration $\vec{I_b}\geq\vec{I_a}$. This overproduces $Y$ since $f(\vec{a})>f(\vec{b})$. Thus to stably compute $f(\vec{b})$, some reaction must consume $Y$ as a reactant, so $\mathcal{C}$ cannot be output-oblivious.
\end{proof}


A CRN being output-oblivious
was shown in \cite{ContinuousMC} (for \emph{continuous} CRNs) to be equivalent to being ``composable via concatenation'',
meaning renaming the output species of one CRN to match the input of another.
This equivalence still holds in our discrete CRN model.
\opt{full,final}{
    This is formalized as Observation~\ref{lem-output-oblivious-implies-composable} and Lemma~\ref{lem-composable-implies-output-oblivious}.
}

\opt{final,full}{

\opt{sub}{Recall that
output-obliviousness
was shown in \cite{ContinuousMC} (for continuous CRNs) to be equivalent to ``composability via concatenation''.
We now 
show this result for the discrete CRN model.}
For CRNs $\mathcal{C}_f$ stably computing $f:\N^d\to\N$ and $\mathcal{C}_g$ stably computing $g:\N\to\N$, define the $\emph{concatenated}$ CRN $\mathcal{C}_{g\circ f}$ 
by combining species and reactions, with $\mathcal{C}_f$'s output species as $\mathcal{C}_g$'s input species and no other common species,
plus a reaction $L \to L^f + L^g$ creating a copy of the leader from each of $\mathcal{C}_f$ and $\mathcal{C}_g$.

We first observe that this composition works correctly if the \emph{upstream} CRN $\mathcal{C}_f$ is output-oblivious. Intuitively, the reactions from $\mathcal{C}_g$ can only affect the reactions from $\mathcal{C}_f$ via the common species $W$, but this output species of $\mathcal{C}_f$ is never used as a reactant to stably compute $f(\vec{x})$.


\begin{observation}
\label{lem-output-oblivious-implies-composable}
If $\mathcal{C}_f$ stably computes $f:\N^d\to\N$, $\mathcal{C}_g$ stably computes $g:\N\to\N$, and $\mathcal{C}_f$ is output-oblivious, then the concatenated CRN $\mathcal{C}_{g\circ f}$ stably computes the composition $g\circ f:\N^d\to\N$.
\end{observation}



Note that the downstream CRN $\mathcal{C}_g$ need not be output-oblivious, but if two output-oblivious CRNs are composed, then the composition $\mathcal{C}_{g\circ f}$ remains output-oblivious. More generally, $g$ can take any number of inputs from output-oblivious CRNs, which act as modules for arbitrary feedforward composition.

The converse shows that a composable CRN is essentially output-oblivious. If $\mathcal{C}_f$ can be correctly composed with any downstream $\mathcal{C}_g$, then $\mathcal{C}_f$ must function correctly even if downstream reactions from $\mathcal{C}_g$ starve it of the common species $W$. Thus $\mathcal{C}_f$ will still stably compute $f$ if we remove all reactions with output $W$ as a reactant, making it output-oblivious.
\opt{final}{The proof appears in~\cite{severson2019composable}.}

\begin{lemma}
\label{lem-composable-implies-output-oblivious}
Let $\mathcal{C}_f$ stably compute $f:\N^d\to\N$ such that for any $\mathcal{C}_g$ stably computing $g:\N\to\N$, the concatenated CRN $\mathcal{C}_{g\circ f}$ stably computes the composition $g\circ f:\N^d\to\N$. Then $\mathcal{C}_f$ still stably computes $f$ if we remove all reactions using the output species as a reactant, making it output-oblivious.
\end{lemma}



\opt{full}{

\begin{proof}
    Let $C_f$ (with output species $W$) stably compute $f:\N^d\to\N$.
    Let $g(x)=x$ be the identity function, stably computed by $W\to Y$. Assume the concatenated CRN $\mathcal{C}_{g\circ f}$ stably computes $g\circ f = f$. Let $\mathcal{C}'_f$ be the output-oblivious CRN with all reactions using $W$ as a reactant removed from $\mathcal{C}_f$. We now show that $\mathcal{C}'_f$ also stably computes $f$.
    
    For any $\vec{x}\in\N^d$, let $\vec{I_x}$ be the initial configuration encoding $\vec{x}$ in $\mathcal{C}'_f$, and $\vec{C}$ be a configuration reachable from $\vec{I_x}$. Now we can naturally view $\vec{I_x}$ and $\vec{C}$ also as configurations in the concatenated CRN $\mathcal{C}_{g\circ f}$. We first consider configuration $\vec{D}$ reachable from $\vec{C}$ by applying reaction $W\to Y$ $\vec{C}(W)$ until $\vec{D}(W)=0$. Now because $\mathcal{C}_{g\circ f}$ stably computes $f$, there exists a sequence of reactions $\alpha$ from $\vec{D}$ to a stable configuration $\vec{O}$ with $\vec{O}(Y)=f(\vec{x})$. $\alpha$ could contain reactions that use $W$ as a reactant, but because $\vec{D}(W)=0$, any such reactions must occur after additional $W$ has been produced.
    
    Before any such reactions using $W$ as a reactant, we can insert more reactions $W\to Y$, reaching an intermediate configuration $\vec{D}_2$ again with $\vec{D}_2(W)=0$. There must exist some new sequence of reactions $\alpha_2$ from $\vec{D}_2$ to a stable configuration $\vec{O}_2$ with $\vec{O}_2(Y)=f(\vec{x})$. Repeating this process, we will eventually find a sequence of reactions $\beta$ from $\vec{C}$ to a correct stable configuration $\vec{O}_n$. Notice that we only have to ``splice in'' new reactions when $W$ is produced, and this can only happen at most $f(\vec{x})$ times, so this process will terminate.
    
    Thus we have demonstrated a sequence of reactions $\beta$ in $\mathcal{C}_{g\circ f}$ from $\vec{C}$ reaching a stable correct output configuration $\vec{O}_n$ without using any reactions using $W$ as a reactant. $\vec{O}_n(Y)=f(\vec{x})$, so $\beta$ contains precisely $f(\vec{x})$ copies of the reaction $W\to Y$. Ignoring these reactions then gives a sequence of reactions $\beta'$ in $\mathcal{C}'_f$ from $\vec{C}$ to a correct configuration $\vec{O}'$ with $\vec{O}'(W)=f(\vec{x})$. Notice that $\vec{O}_n$ being stable in $\mathcal{C}_{g\circ f}$ implies $\vec{O}'$ is stable in $\mathcal{C}'_f$ since no additional $W$ can be produced. This shows $\mathcal{C}'_f$ stably computes $f$ as desired.
\end{proof}

}


\opt{full}{

We finally observe that the more general definition of \emph{output-monotonic} CRNs (which cannot decrease the count of the output species) stably compute precisely the same set of functions as output-oblivious CRNs:

\begin{observation}
\label{obs:output-monotonic-is-equivalent}
$f:\N^d\to\N$ is stably computable by an output-oblivious CRN $\iff$ $f$ is stably computable by an output-monotonic CRN.
\end{observation}

\begin{proof}
$\implies:$ Any output-oblivious CRN must be output-monotonic.

$\impliedby:$ If $f$ was stably computed by an output-monotonic CRN $\mathcal{C}$ which is not already output-oblivious, then there must be reactions of the form $Y+\ldots\to Y+\ldots$ with output species $Y$ acting as a catalyst. $\mathcal{C}$ can be made output-oblivious by replacing all such occurrences of $Y$ as a catalyst by a new catalyst species $Z$ that is always produced alongside $Y$. 
Since $\mathcal{C}$ was output-monotonic,
if a $Y$ is ever produced, it cannot be consumed.
Thus any reactions with $Y$ as a catalyst are ``turned on'' the moment the first $Y$ is produced and never turn off again.
So it does not change the reachable configurations to irreversibly produce a $Z$ alongside $Y$ and use $Z$ as the catalyst in place of $Y$.
This output-oblivious CRN thus also stably computes $f$.
\end{proof}

}

 }

\subsection{Semilinear functions}
The functions stably computable by a CRN were shown in \cite{CheDotSolDetFuncNaCo}, building from work in \cite{StablyComputableAreSemilinear}, to be precisely the \emph{semilinear} functions, which are defined based on semilinear sets\footnote{Semilinear sets have other common equivalent definitions \cite{angluin2006passivelymobile}; the above definition is convenient for our proof.} 

\begin{definition}
\label{semilinear set def}
A subset $S\subseteq\N^d$ is \emph{semilinear} if $S$ is a finite Boolean combination (union, intersection, complement) of \emph{threshold} sets of the form $\{\vec{x}\in\N^d:\vec{a}\cdot\vec{x}\geq b\}$ for $\vec{a}\in\Z^d,b\in\Z$ and \emph{mod} sets of the form $\{\vec{x}\in\N^d:\vec{a}\cdot\vec{x}\equiv b\mod c\}$ for $\vec{a}\in\Z^d,b\in\Z,c\in\N_+$.
\end{definition}

A \emph{semilinear function} can be concisely defined as having a semilinear graph, but a more useful equivalent definition comes from Lemma 4.3 of \cite{CheDotSolDetFuncNaCo}\footnote{Lemma 4.3 in \cite{CheDotSolDetFuncNaCo} has domains that are non-disjoint linear sets. We assume the domains are disjoint for convenience, making the domains semilinear sets.}:

\begin{definition}[\cite{CheDotSolDetFuncNaCo}]
\label{semilinear function def}
A function $f:\N^d\to\N$ is \emph{semilinear} if $f$ is the finite union of affine partial functions, whose domains are disjoint semilinear subsets of $\N^d$.
\end{definition}
%
%
%
All functions discussed have been semilinear. For example, the function 
\[
\min(x_1,x_2)=
\begin{cases}
x_1, & \text{if } x_1\leq x_2\\
x_2, & \text{if } x_1 > x_2
\end{cases}
\]
is semilinear with affine partial functions on disjoint domains which are defined by a single threshold and thus semilinear.

Similarly, the function 
\[
\bigg\lfloor\frac{3x}{2}\bigg\rfloor=
\begin{cases}
\frac{3}{2}x, & \text{if $x$ is even}\\
\frac{3}{2}x-\frac{1}{2}, & \text{if $x$ is odd}
\end{cases}
\]
is semilinear with affine partial functions on disjoint domains which are defined by parity (a single mod predicate) and thus semilinear. All quilt-affine $f:\N^d\to\N$ are semilinear by the same argument.

\begin{lemma}[\cite{CheDotSolDetFuncNaCo}]
\label{Semilinear Are Stably Computable Lemma}
A function $f:\N^d\to\N$ is stably computable $\iff$ $f$ is semilinear.
\end{lemma}

\section{Warm-up: One-dimensional case}
\label{sec:one-dim-case}

For functions with one-dimensional input, the 
necessary conditions of being nondecreasing and semilinear 
are also sufficient.

\opt{sub,final}{
\begin{figure}
    \includegraphics[width=2.8in]{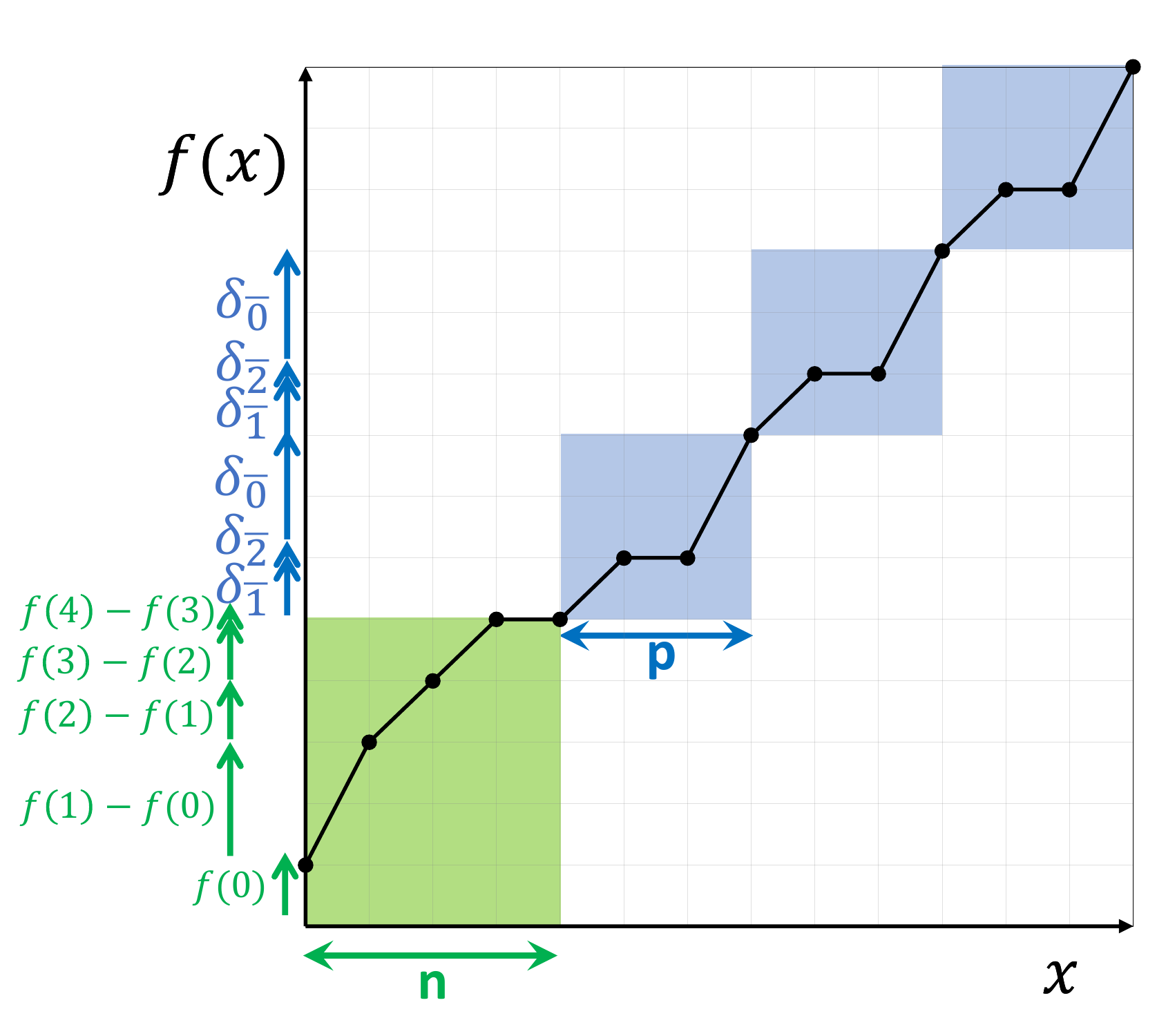}
    \caption{\footnotesize Every semilinear nondecreasing $f:\N\to\N$ is eventually quilt-affine, with periodic finite differences $\delta_{\overline{x}}$.}
    \label{fig:1d-eventually-quilt-affine}
\end{figure}
}


\restateableTheorem{1d MC theorem}{ThmOneDimMC}{
    $f:\N\to\N$ is obliviously-computable $\iff$ $f$ is semilinear and nondecreasing.
}{
    $\implies:$ Lemma \ref{Semilinear Are Stably Computable Lemma} and Observation \ref{MC is nondecreasing}.

    $\impliedby:$ If $f:\N\to\N$ is semilinear and nondecreasing, it will eventually be \emph{quilt-affine}
    (generalized to higher-dimensional functions as Definition~\ref{quilt-affine def}) and thus have periodic finite differences: for some $n\in\N$, period $p\in\N_+$, and finite differences $\delta_{\overline{0}},\ldots,\delta_{\overline{p-1}}\in\N$, then for all $x\geq n$,  $f(x+1)-f(x)=\delta_{(\overline{x}\mod p)}$ (see Fig.~\ref{fig:1d-eventually-quilt-affine}).
    
    Because $f$ is semilinear, by Definitions \ref{semilinear set def} and \ref{semilinear function def}, it can be represented as a disjoint union of affine partial functions, whose domains are semilinear sets, and thus represented as finite Boolean combinations of threshold $\{x\in\N:x\geq a\}$ and mod $\{x\in\N:x\equiv b\mod c\}$ sets. Now take $n\in\N$ greater than all such $a$ and $p=\mathrm{lcm}(c)$ for all such $c$. Then for all $x\geq n$, $f$ periodically cycles between affine partial functions. Because $f$ is nondecreasing, these periodically-repeated affine partial functions must all have the same slope. This implies $f$ is eventually quilt-affine, with periodic finite differences for all $x\geq n$ as claimed.
    
    The CRN $\mathcal{C}$ to stably compute $f$ uses input species $X$, output species $Y$, leader $L$, and species $L_0,\ldots,L_{n-1},P_{\overline{0}},\ldots,P_{\overline{p-1}}$ corresponding to auxiliary ``states'' of the leader,
    i.e., exactly one of $L, L_0,\ldots,L_{n-1},P_{\overline{0}},\ldots,P_{\overline{p-1}}$ is present at any time. Intuitively the leader tracks how many input $X$ it has seen, where the count past $n$ wraps around mod $p$, and outputs the correct finite differences. The reactions of $\mathcal{C}$ are as follows
    \begin{align*}
        L&\to f(0)Y+L_0
        \\
        L_i+X&\to[f(i+1)-f(i)]Y+L_{i+1}\qquad\text{for all }i=0,\ldots,n-2
        \\
        L_{n-1}+X&\to[f(n)-f(n-1)]Y+P_{\overline{n}}
        \\
        P_{\overline{a}}+X&\to\delta_{\overline{a}} Y+P_{\overline{a+1}}
        \qquad\text{for all }\overline{a}=\overline{0},\ldots,\overline{p-1}. 
        \qedhere
    \end{align*}
}

\opt{sub}{A proof is given in the Appendix (\ref{sec:appendix-1d-case}).}

\opt{full}{
\begin{figure}
    \centering
    \includegraphics[width=\opt{final}{2.8in}\opt{full}{2.8in}]{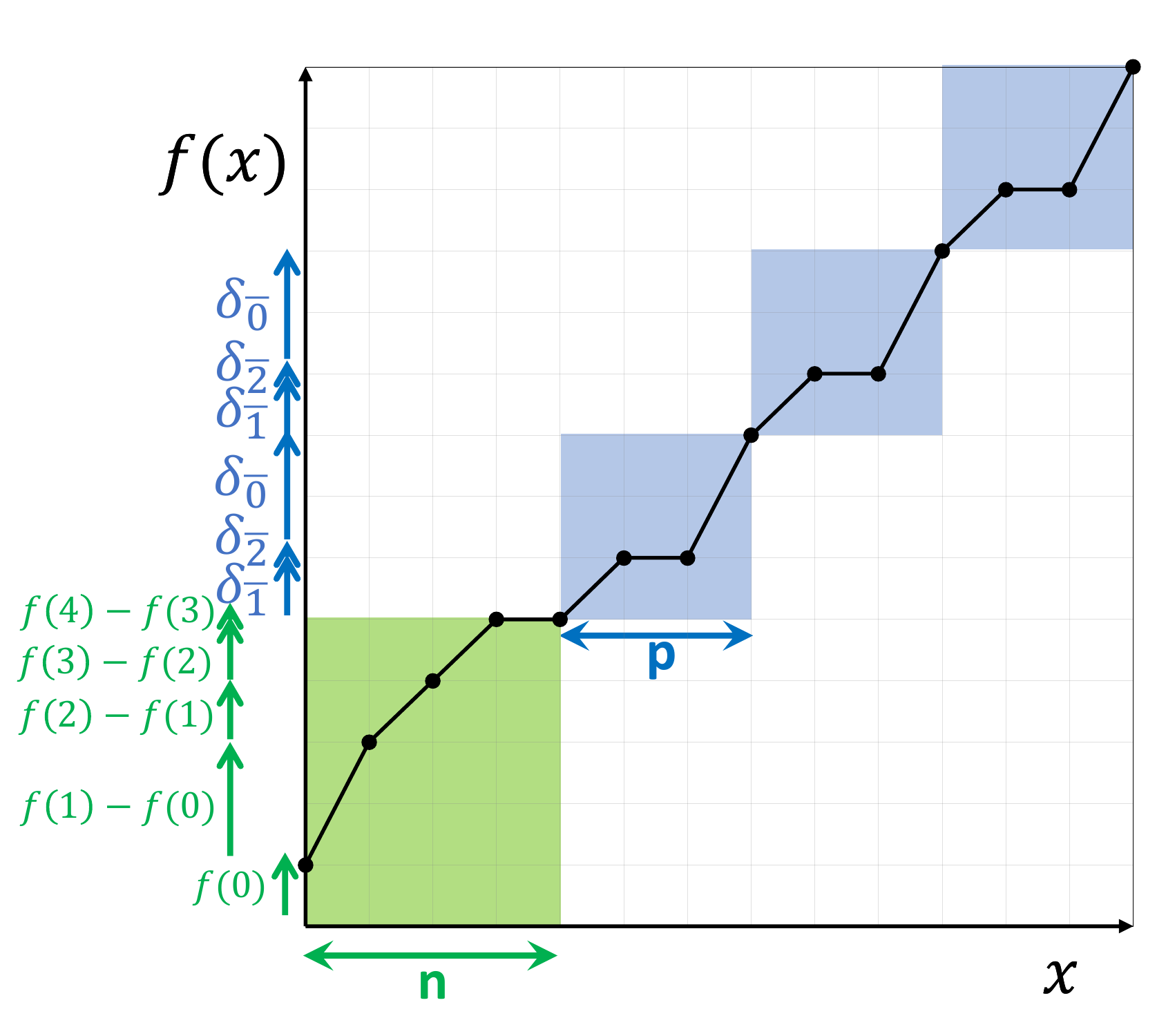}
    \caption{\footnotesize Every semilinear nondecreasing $f:\N\to\N$ is eventually quilt-affine, with periodic finite differences $\delta_{\overline{x}}$.}
    \label{fig:1d-eventually-quilt-affine}
\end{figure}
}

Intuitively,
the proof works as follows. We show semilinear, nondecreasing $f:\N\to\N$ must have the eventually quilt-affine structure in Fig.~\ref{fig:1d-eventually-quilt-affine}. From this structure, we define a CRN that uses auxiliary leader states to track the value of $x$ (or $\overline{x}\mod p$ once $x\geq n$), while outputting the correct finite differences from adding each input.

\opt{full,final}{\proofThmOneDimMC}

In the 1D case,
we can also characterize the functions obliviously-computable \emph{without} a leader: they are semilinear and \emph{superadditive}: meaning  $f(x)+f(y)\leq f(x+y)$ for all $x,y$.
(Theorem~\ref{1d MC0 theorem})

\section{Impossibility result}
\label{sec:max}

The characterization of obliviously-computable functions as precisely semilinear and nondecreasing from Theorem $\ref{1d MC theorem}$ 
is insufficient
in higher dimensions. As an example, consider the function 
$\max:\N^2\to\N$,
which is both semilinear and nondecreasing. We prove $\max$ is not obliviously-computable via a more general lemma:



\restateableLemma{dickson contradiction lemma}{LemDicksonContradiction}{
    Let $f:\N^d\to\N$. If there exists an increasing sequence $(\vec{a}_1,\vec{a}_2,\ldots)\in\N^d$ such that for all $i<j$ there exists some $\vec{\Delta}_{ij}\in\N^d$ with
    \opt{full}{\[ f(\vec{a}_i+\vec{\Delta}_{ij})-f(\vec{a}_i)>f(\vec{a}_j+\vec{\Delta}_{ij})-f(\vec{a}_j), \]} 
    \opt{sub,final}{$ f(\vec{a}_i+\vec{\Delta}_{ij})-f(\vec{a}_i)>f(\vec{a}_j+\vec{\Delta}_{ij})-f(\vec{a}_j), $}
    then $f$ is not obliviously-computable.
}{
    Assume toward contradiction an output-oblivious CRN $\mathcal{C}$ stably computes $f$. To stably compute each $f(\vec{a}_i)$, the initial configuration $\vec{I}_{\vec{a}_i}\to^*\vec{O}_i$ for some configuration with $\vec{O}_i(Y)=f(\vec{a}_i)$, giving a sequence of configurations $(\vec{O}_i)_{i=1}^{\infty}$. By Dickson's Lemma \cite{DicksonsPaper}, any sequence of nonnegative integer vectors has a nondecreasing subsequence, so there must be $\vec{O}_i\leq\vec{O}_j$ for some $i<j$. By assumption there exists $\vec{\Delta}_{ij}\in\N^d$ such that 
    $$f(\vec{a}_i+\vec{\Delta}_{ij})-f(\vec{a}_i)>f(\vec{a}_j+\vec{\Delta}_{ij})-f(\vec{a}_j)$$
    
    Now consider the initial configuration $\vec{I}_{\vec{a}_i+\vec{\Delta}_{ij}}\geq \vec{I}_{\vec{a}_i}$, so define the difference $\vec{D}=\vec{I}_{\vec{a}_i+\vec{\Delta}_{ij}}-\vec{I}_{\vec{a}_i}\in\N^\mathcal{S}$. Then the same sequence of reactions $\vec{I}_{\vec{a}_i}\to^*\vec{O}_i$ is applicable to $\vec{I}_{\vec{a}_i+\vec{\Delta}_{ij}}$ reaching configuration $\vec{C}_i=\vec{O}_i+\vec{D}$, with $\vec{C}_i(Y)=\vec{O}_i(Y)=f(\vec{a}_i)$. Then to stably compute $f(\vec{a}_i+\vec{\Delta}_{ij})$ there must exist a further sequence of reactions $\alpha$ from $\vec{C}_i$ that produce an additional $f(\vec{a}_i+\vec{\Delta}_{ij})-f(\vec{a}_i)$ copies of output $Y$.
    
    By the same argument, from initial configuration $\vec{I}_{\vec{a}_j+\vec{\Delta}_{ij}}$ the configuration $\vec{C}_j=\vec{O}_j+\vec{D}$ is reachable, with $\vec{C}_j(Y)=\vec{O}_j(Y)=f(\vec{a}_j)$. Then since $\vec{O}_i\leq\vec{O}_j$, we have $\vec{C}_i\leq\vec{C}_j$, so the same sequence of reactions $\alpha$ is applicable to $\vec{C}_j$, reaching some configuration $\vec{C}'_j$ with an additional $f(\vec{a}_i+\vec{\Delta}_{ij})-f(\vec{a}_i)$ copies of output $Y$, so
    $$\vec{C}'_j(Y)=f(\vec{a}_j)+f(\vec{a}_i+\vec{\Delta}_{ij})-f(\vec{a}_i)>f(\vec{a}_j+\vec{\Delta}_{ij})$$
    Then $\vec{I}_{\vec{a}_j+\vec{\Delta}_{ij}}\to^*\vec{C}'_j$ overproduces $Y$, so the output-oblivious CRN $\mathcal{C}$ cannot stably compute $f(\vec{a}_j+\vec{\Delta}_{ij})$.
}

\begin{figure}
    \opt{full}{\centering}
    \includegraphics[width=\opt{final}{2.25in}{\opt{full}{3in}}]{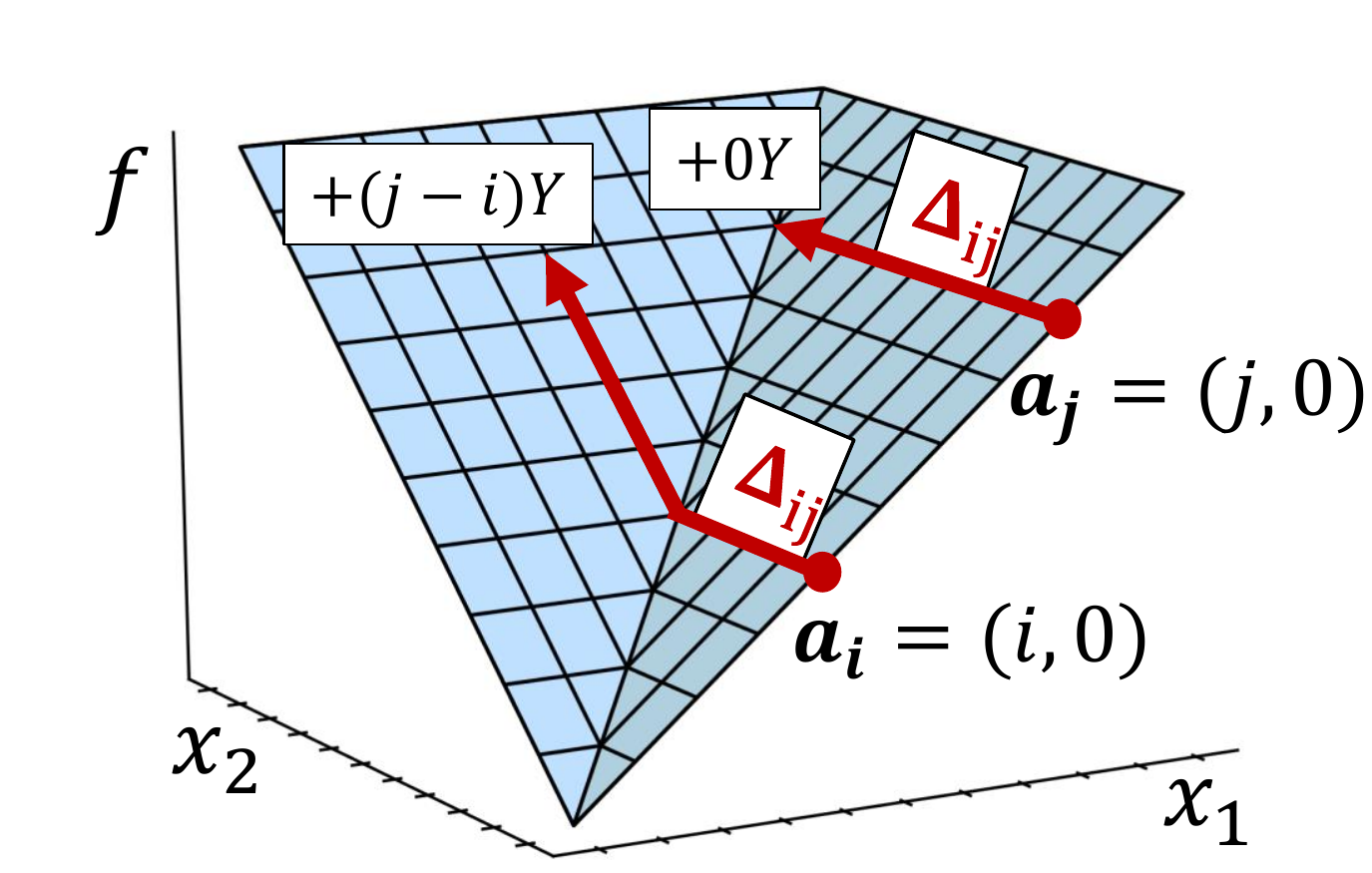}
    \caption{\footnotesize Lemma \ref{dickson contradiction lemma} applied to $f=\max(x_1,x_2)$.}
    \label{fig:max-example}
\end{figure}

\opt{sub}{A proof is given in the Appendix (\ref{sec:appendix-contradiction-lemma}). We now use it to show $\max$ is not obliviously-computable.}
\opt{full,final}{Before proving Lemma~\ref{dickson contradiction lemma}, we use it to show $\max$ is not obliviously-computable.}

For $f=\max(x_1, x_2)$, 
we let $\vec{a}_i=(i,0)$ and 
$\vec{\Delta}_{ij}=(0,j)$, so for $i<j,$
\opt{full,sub}{\[ \max(i,j)-\max(i,0)=j-i>\max(j,j)-\max(j,0)=0 \]}
\opt{final}{$ \max(i,j)-\max(i,0)=j-i>\max(j,j)-\max(j,0)=0 $}
as desired (see Fig.~\ref{fig:max-example}).
Adding $\vec{\Delta}_{ij}$ input after computing $f(\vec{a}_i)$ should produce $j-i$ additional output $Y$. However, adding $\vec{\Delta}_{ij}$ input after computing $f(\vec{a}_j)$ should not. Lemma \ref{dickson contradiction lemma} uses this to show there exists a reaction sequence that overproduces $Y$, thus max is not obliviously-computable.
\opt{full,final}{We now prove Lemma~\ref{dickson contradiction lemma}.}

\opt{full,final}{\proofLemDicksonContradiction}

Lemma~\ref{dickson contradiction lemma} is our main technical tool used to show that a particular semilinear, nondecreasing function is not obliviously-computable,
the key challenge in the impossibility direction of Theorem~\ref{thm:main result}. 

\section{Main result: Full-dimensional case}
\label{sec:full-dim-case}


To formally state our main result, Theorem \ref{thm:main result}, we must first define \emph{quilt-affine} functions as the sum of a linear and periodic function (see Fig.~\ref{fig:2D-quilt-affine-example}):
\begin{definition}
\label{quilt-affine def}

A nondecreasing function $g:\N^d\to\Z$ is quilt-affine (with period $p$) if there exists $\vec{\nabla}_g\in\Q^d_{\geq0}$ and $B:\Z^d/p\Z^d\to\Q$ such that
\opt{full}{ \[ g(\vec{x})=\vec{\nabla}_g\cdot\vec{x}+B(\barx\mod p). \] }
\opt{sub,final}{ $ g(\vec{x})=\vec{\nabla}_g\cdot\vec{x}+B(\barx\mod p). $ }
\end{definition}


We call $\vec{\nabla}_g$ the \emph{gradient} of $g$, and the periodic function $B$ the \emph{periodic offset}. Without loss of generality we have the same period $p$ along all inputs, since $p$ could be the least common multiple of the periods along each input component. Note that $\nabla_g\cdot\vec{x}$ and $B$ can each be rational, but the sum $g(x)\in\Z$ will be integer-valued. We allow $g$ to have negative output for technical reasons\footnote{
The quilt-affine functions that describe $f$ for large inputs may be negative on inputs close to the origin.},
but in the case that $g$ is quilt-affine with nonnegative output (i.e. $g:\N^d\to\N$), there is a simple output-oblivious CRN construction to stably compute $g$. 
The intuitive idea is to use a single leader that reacts with every input species sequentially, tracks the periodic value $\barx\mod p$, and outputs the correct changes in $g$ (Lemma \ref{CRC for quilt-affine function}).

Our main result has a recursive condition where we fix the input of a function $f:\N^d\to\N$. 
For each $i=1,\ldots,d$ and $j\in\N$,
define the \emph{fixed-input restriction}\footnote{
    We define $\restrictf$ to have domain $\N^d$ because it is notationally convenient to have the same domain as $f$, but $\restrictf$ only has relevant input in $d-1$ of its input components, 
    making condition \eqref{defn-obv-computable-3} recursive.
} 
$\restrictf: \N^d \to \N$ of $f$ for all $\vec{x} \in \N^d$ by
$
\restrictf(\vec{x})=f(\vec{x}(1),\ldots,\vec{x}(i-1),j,\vec{x}(i+1),\ldots,\vec{x}(d)).
$


We can now formally state our main result:

\begin{theorem}
\label{thm:main result}
$f:\N^d\to\N$ is obliviously-computable $\iff$
\begin{enumerate}[i)]
    \item \label{defn-obv-computable-1}
    [nondecreasing]
    $f$ is nondecreasing,
    
    \item \label{defn-obv-computable-min}
    [eventually-min]
there exist quilt-affine $g_1,\ldots,g_m:\N^d\to\Z$ and $\vec{n}\in\N^d$ such that for all $\vec{x}\geq\vec{n}$, $f(\vec{x})=\min_k(g_k(\vec{x}))$, and
    
    \item \label{defn-obv-computable-3}
    [recursive]
    all fixed-input restrictions $\restrictf$ are obliviously-computable.
\end{enumerate}
\end{theorem}


We first prove that these conditions imply $f$ is obliviously-computable via a general CRN construction in Section~\ref{Construction}.

The nondecreasing condition~\eqref{defn-obv-computable-1} is necessary by Observation~\ref{MC is nondecreasing}. 
It is immediate to see the recursive condition~\eqref{defn-obv-computable-3} is also necessary:
\begin{observation}
If $f:\N^d\to\N$ is obliviously-computable, then any fixed-input restriction $\restrictf:\N^d\to\N$ is obliviously-computable.
\end{observation}
\begin{proof} 
Let the output-oblivious CRN $\mathcal{C}$ stably compute $f$. We define the output-oblivious CRN $\mathcal{C}'$ to ``hardcode'' the input $\vec{x}(i)=j$ by modifying the reactions of $\mathcal{C}$. Replace all instances of the leader $L$ and input species $X_i$ by $L'$ and $X_i'$ respectively, then add the initial reaction $L\to jX_i'+L'$. It is straightforward to verify that $\mathcal{C}'$ stably computes $\restrictf$.
\end{proof}

Then the remaining work (and biggest effort of this paper) is to show the necessity of the eventually-min condition~\eqref{defn-obv-computable-min}: that every obliviously-computable function can be represented as eventually a minimum of a finite number of quilt-affine functions, which is shown as Theorem~\ref{MC implies min}. 
Its proof relies on $f$ being semilinear, nondecreasing, and not having any ``contradiction sequences'' to apply Lemma~\ref{dickson contradiction lemma}. 
Thus the proof of Theorem~\ref{MC implies min} also yields the following alternative characterization to Theorem \ref{thm:main result}:

\begin{theorem}
\label{thm-negative-characterization}
$f:\N^d\to\N$ is obliviously-computable $\iff$ $f$ is semilinear, nondecreasing, and has no sequence $(\vec{a}_1,\vec{a}_2,\ldots)$ meeting the conditions of Lemma \ref{dickson contradiction lemma}.
\end{theorem}


This gives a ``negative characterization'' identifying behavior obliviously-computable functions must \emph{avoid},
whereas Theorem~\ref{thm:main result}, is a ``positive characterization'' describing the allowable behavior of such functions.
We include Theorem~\ref{thm-negative-characterization}, though it is less descriptive of the function,
because it may be useful in other contexts.

\section{Construction}
\label{Construction}

\newcommand{\lemAndProofQuiltAffineConstruction}{
    \begin{lemma}
    \label{CRC for quilt-affine function}
    Every quilt-affine function $g:\N^d\to\N$ is obliviously-computable.
    \end{lemma}
    
    \begin{proof}
    Let $g:\N^d\to\N$ be quilt-affine with period $p$ (recalling Definition~\ref{quilt-affine def}). Notice that $g$ has periodic finite differences. For each congruence class $\bara\in\Z^d/p\Z^d$ and input component $i=1,\ldots,d$, where $\vec{e}_i$ is the $i$th standard basis vector, define
    \[
    \delta_{\bara}^i=\vec{\nabla}_g\cdot\vec{e}_i+B(\overline{\vec{a}+\vec{e}_i}\mod p)-B(\bara\mod p)\in\N.
    \]
    Observe that for all $\vec{x}\in\bara$, $g(\vec{x}+\vec{e}_i)-g(\vec{x})=\delta_{\bara}^i$. We now use these periodic finite differences to construct an output-oblivious CRN $\mathcal{C}$ to stably compute $g$.
    
   The CRN $\mathcal{C}$ has input species $X_1,\ldots,X_d$, output species $Y$, leader species $L$ and $p^d$ additional species $L_{\bara}$ for each $\bara\in\Z^d/p\Z^d$ coresponding to auxilliary ``states'' of the leader. 
    The initial reaction $L\to g(\vec{0})Y+L_{\overline{\vec{0}}}$ is accompanied by $dp^d$ reactions of the form
    \[
        L_{\bara}+X_i
        \to 
        \delta_{\bara}^iY+L_{\overline{\vec{a}+\vec{e}_i}}
    \]
    for each $i=1,\ldots,d$ and $\bara\in\Z^d/P\Z^d$. This CRN first creates $g(\vec{0})$ output, then sequentially outputs all finite differences, and is easily verified to stably compute $g$.
    \end{proof}
}
First we show that any quilt-affine function with nonnegative range is stably computed by an output-oblivious CRN: 
\opt{sub}{(See Lemma \ref{CRC for quilt-affine function}.)}

\opt{full,final}{\lemAndProofQuiltAffineConstruction}

We now prove (in Lemma~\ref{lem-general-construction}) one direction of Theorem~\ref{thm:main result}: that conditions \eqref{defn-obv-computable-1}, \eqref{defn-obv-computable-min}, and \eqref{defn-obv-computable-3} imply an output-oblivious CRN can stably compute $f$. Intuitively, by the eventually-min condition~\eqref{defn-obv-computable-min} we compute $f(\vec{x})$ for $\vec{x}\geq\vec{n}$ by composing min and quilt-affine functions. If $\vec{x}\ngeq\vec{n}$, then $\vec{x}(i)=j$ for some input $i$ and $j<\vec{n}(i)$. By the recursive condition~\eqref{defn-obv-computable-3} we compute $\restrictf(\vec{x})=f(\vec{x})$
\footnote{As a result, this construction is recursive, with an additional input being fixed at each level of the recursion, so the base case is simply a constant function.}. The key remaining insight is a trick (similar to a proof in \cite{ContinuousMC}) to compose these pieces using minimum and indicator functions.


The proof of Lemma~\ref{lem-general-construction} then expresses $f$ as such a minimum of finitely many pieces. We justify that $f$ is obliviously-computable by showing that each piece is obliviously-computable, since by Observation \ref{lem-output-oblivious-implies-composable} obliviously-computable functions are closed under composition.
\opt{sub}{A proof of Lemma~\ref{lem-general-construction} is given in the Appendix (\ref{sec:appendix-construction}).}

\restateableLemma{lem-general-construction}{LemGeneralConstruction}{
    If $f:\N^d\to\N$ satisfies the conditions of Theorem~\ref{thm:main result}, 
    $f$ is obliviously-computable.
}{
    Assume $f:\N^d\to\N$ satisfies 
the conditions of Theorem~\ref{thm:main result}.
Then by eventually-min condition~\eqref{defn-obv-computable-min}, there exist  quilt-affine $g_1,\ldots,g_m:\N^d\to\Z$ and $\vec{n}\in\N^d$ (without loss of generality assume $\vec{n}=(n,\ldots,n)$) such that $f(\vec{x})=\min_k(g_k(\vec{x}))$ for all $\vec{x}\geq\vec{n}$.

\newcommand{\xn}{\vec{x}\lor\vec{n}}

Let $\xn=(\max(\vec{x}(1),n),\ldots,\max(\vec{x}(d),n))$ denote the componentwise max of $\vec{x}$ and $\vec{n}$. Let $\1_{\vec{x}(i) > j}:\N^d \to \{0,1\}$ denote the indicator function that is 1 $\iff$ its input $\vec{x}$ obeys $\vec{x}(i)>j$. Recall $\restrictf$ is the fixed-input restriction setting input $\vec{x}(i)=j$.
%
We claim that $f$ can be expressed as
\begin{equation} \label{eq-f-express-as-min}
    f(\vec{x})
    =
    \min \bigg[
        f(\xn),
        \underbrace{\restrictf(\vec{x}) + \1_{\{\vec{x}(i) > j\}}(\vec{x}) \cdot f(\xn)}_{\substack{i=1,\ldots,d \\ j=0,\ldots,n-1}}
    \bigg].
\end{equation}
We first show $f\geq\min[\ldots]$ since for all $\vec{x}\in\N^d$, $f(\vec{x})$ is achieved by some term. If $\vec{x}\geq\vec{n}$, then $f(\vec{x})=f(\xn)$. If $\vec{x}\ngeq\vec{n}$, there must be $\vec{x}(i)=j$ for some $i=1,\ldots,d$ and $j=0,\ldots,n-1$, so $f(\vec{x})=\restrictf(\vec{x})=\restrictf(\vec{x}) + \1_{\{\vec{x}(i) > j\}}(\vec{x}) \cdot f(\xn)$ since the indicator is $0$.

We next show $f\leq\min[\ldots]$ since $f(\vec{x})\leq$ each term for all $\vec{x}\in\N^d$. $f(\vec{x})\leq f(\xn)$ since $\vec{x}\leq(\xn)$ and $f$ is nondecreasing. When $\1_{\{\vec{x}(i) > j\}}(\vec{x})=1$, we then have $f(\vec{x})\leq \restrictf(\vec{x}) + \1_{\{\vec{x}(i) > j\}}(\vec{x}) \cdot f(\xn)$. If $\1_{\{\vec{x}(i) > j\}}(\vec{x})=0$, then $\vec{x}(i)\leq j$ so $f(\vec{x})\leq\restrictf(\vec{x})$ since $f$ is nondecreasing. Thus equation \ref{eq-f-express-as-min} holds as claimed.

It remains to show that $f$ is obliviously-computable. From Observation \ref{lem-output-oblivious-implies-composable}, output-oblivious CRNs are closed under composition, and equation \ref{eq-f-express-as-min} gives a method to express $f$ as a composition of functions. Thus it suffices to show that each piece is obliviously-computable. Specifically, we show the functions $\min:\N^k\to\N$ (for any $k$), $f(\xn):\N^d\to\N$, $\restrictf(\vec{x}):\N^d\to\N$, and $c(a,b,\vec{x})=a+\1_{\{\vec{x}(i) > j\}}(\vec{x}) \cdot b:\N^{d+2}\to\N$ are each obliviously-computable. Implicit in the composed CRN to stably compute $f$ as the composition  from equation \ref{eq-f-express-as-min} is the ``fan out'' operation where reactions of the form $X_i\to X^1_i,\ldots,X^m_i$ create multiple copies of species $X_i$ to be used as independent inputs to multiple ``modules'' in this composition.
%
%
\begin{description}
    \item[$\min:\N^k\to\N$ is obliviously-computable:]~\\
    Consider the CRN with single reaction $X_1,\ldots,X_k\to Y$, the natural generalization of two-input $\min$ from Fig.~\ref{fig:CRN-computing-examples}.
    
    \item[$f(\xn):\N^d\to\N$ is obliviously-computable:]~\\
    By condition~\eqref{defn-obv-computable-min}, $f(\xn)=\min_k(g_k(\xn))$ since $\xn\geq\vec{n}$, so it suffices to show for each quilt-affine $g_k:\N^d\to\Z$ that $g_k(\xn)$ is obliviously-computable.
    
    By condition~\eqref{defn-obv-computable-min}, $g_k(\vec{x}+\vec{n})\geq f(\vec{x}+\vec{n}) \geq 0$ since $\vec{x}+\vec{n}\geq\vec{n}$. Then $g_k(\vec{x}+\vec{n}):\N^d\to\N$ is still quilt-affine since that property is preserved by translation, but now has guaranteed nonnegative output. Thus by Lemma \ref{CRC for quilt-affine function}, $g_k(\vec{x}+\vec{n}):\N^d\to\N$ is obliviously-computable.
    
    Letting $(\vec{x}-\vec{n})_+=(\max(\vec{x}(1)-n,0),\ldots,\max(\vec{x}(d)-n,0))$, we then show the function $(\vec{x}-\vec{n})_+:\N^d\to\N^d$ is obliviously-computable via the CRN with reactions $(n+1)X_i\to nX_i+Y_i$ for each component $i=1,\ldots,d$. 
    
    Finally, because $\xn=(\vec{x}-\vec{n})_++\vec{n}$, we have shown $g_k(\xn)=g_k((\vec{x}-\vec{n})_++\vec{n})$ is obliviously-computable as the composition of obliviously-computable $g_k(\vec{x}+\vec{n})$ and $(\vec{x}-\vec{n})_+$.
    
    \item[$\restrictf(\vec{x}):\N^d\to\N$ is obliviously-computable:]~\\
    This is precisely the assumed recursive condition~\eqref{defn-obv-computable-3}.
    
\opt{final}{    \item[{\parbox[t]{\linewidth}{$c(a,b,\vec{x})=a+\1_{\{\vec{x}(i) > j\}}(\vec{x}) \cdot b:\N^{d+2}\to\N$ \\ is obliviously-computable:}}]~\\
    Consider the output-oblivious CRN (with input species \\}
\opt{full}{\item[$c(a,b,\vec{x})=a+\1_{\{\vec{x}(i) > j\}}(\vec{x}) \cdot b:\N^{d+2}\to\N$ is obliviously-computable:]~\\
    Consider the output-oblivious CRN (with input species 
}    
    $A,B, X_1, \ldots, X_d$ and output species $Y$) with two reactions $A\to Y$ and $(j+1)X_i+B\to(j+1)X_i+Y$. The $A$ is all converted to $Y$, and $(j+1)$ copies of input species $X_i$ catalyze the conversion of $B$ to $Y$, which will only happen when $\1_{\{\vec{x}(i) > j\}}(\vec{x})=1$. Thus this stably computes $c(a,b,\vec{x})$ as desired. \qedhere
\end{description}
}



\opt{full,final}{\proofLemGeneralConstruction}

\section{Output-oblivious implies eventually min of quilt-affine functions}
\label{Forward Direction}
To complete the proof of Theorem \ref{thm:main result}, it remains to show the necessity of the eventually-min condition~\eqref{defn-obv-computable-min}:

\begin{theorem}
\label{MC implies min}
If $f:\N^d\to\N$ is obliviously-computable,
then there exist quilt-affine $g_1,\ldots,g_m:\N^d\to\Z$ and $\vec{n}\in\N^d$ such that for all $\vec{x}\geq\vec{n}$, $f(\vec{x})=\min_k(g_k(\vec{x}))$.
\end{theorem}


For the remainder of Section~\ref{Forward Direction},
we fix an obliviously-computable $f:\N^d\to\N$,
and Section~\ref{Forward Direction} is devoted to finding $g_1,\ldots,g_m$ and $\vec{n}$
satisfying Theorem~\ref{MC implies min}.

\subsection{Proof outline}
\label{sec:forward-direction-proof-outline}


\opt{full,sub}{
    \noindent {\bf Section \ref{subsec-sunshine-decomposition}.}
}
Since $f:\N^d\to\N$ is  obliviously-computable, $f$ is semilinear (recall Definition~\ref{semilinear function def}), and we first consider all threshold sets used to define the semilinear domains of the affine partial functions that define $f$. Each threshold set defines a hyperplane, and we use these hyperplanes to define \emph{regions} (see Fig.~\ref{fig:2d-regions} and Fig.~\ref{fig:3d-regions}). We consider regions as subsets of $\R^d_{\geq 0}$, 
so they are convex polyhedra with useful geometric properties.\footnote{What we consider is a restricted case of a \emph{hyperplane arrangement} \cite{StanleyHyperplaneArrangments}, with well-studied combinatorial properties.}

The regions partition
\footnote{Without loss of generality, we
assume that
the hyperplanes
do not intersect $\N^d$, so that the partition is well-defined
(see Fig.~\ref{fig:2d-regions}).}
 the points in the domain $\N^d$.
To prove Theorem \ref{MC implies min}, for each region $R_k$ we will identify a quilt-affine function $g_k$ (the \emph{extension} of $f$ from region $R$) such that $g(\vec{x})=f(\vec{x})$ for all integer $\vec{x}\in R$. To ensure $f=\min_k(g_k)$, we further require that these quilt-affine extensions \emph{eventually dominate} $f$ (each $g_k(\vec{x})\geq f(\vec{x})$ for sufficiently large $\vec{x}$). Also, because we only care about sufficiently large $\vec{x}$, we need only consider \emph{eventual} regions which are unbounded in all inputs (for example regions 3,4, and 5 in Fig.~\ref{fig:2d-regions}).

\opt{final}{
\begin{figure}[t]
    \centering
    \begin{subfigure}[t]{0.21\textwidth}
        \centering
        \includegraphics[width=\textwidth]{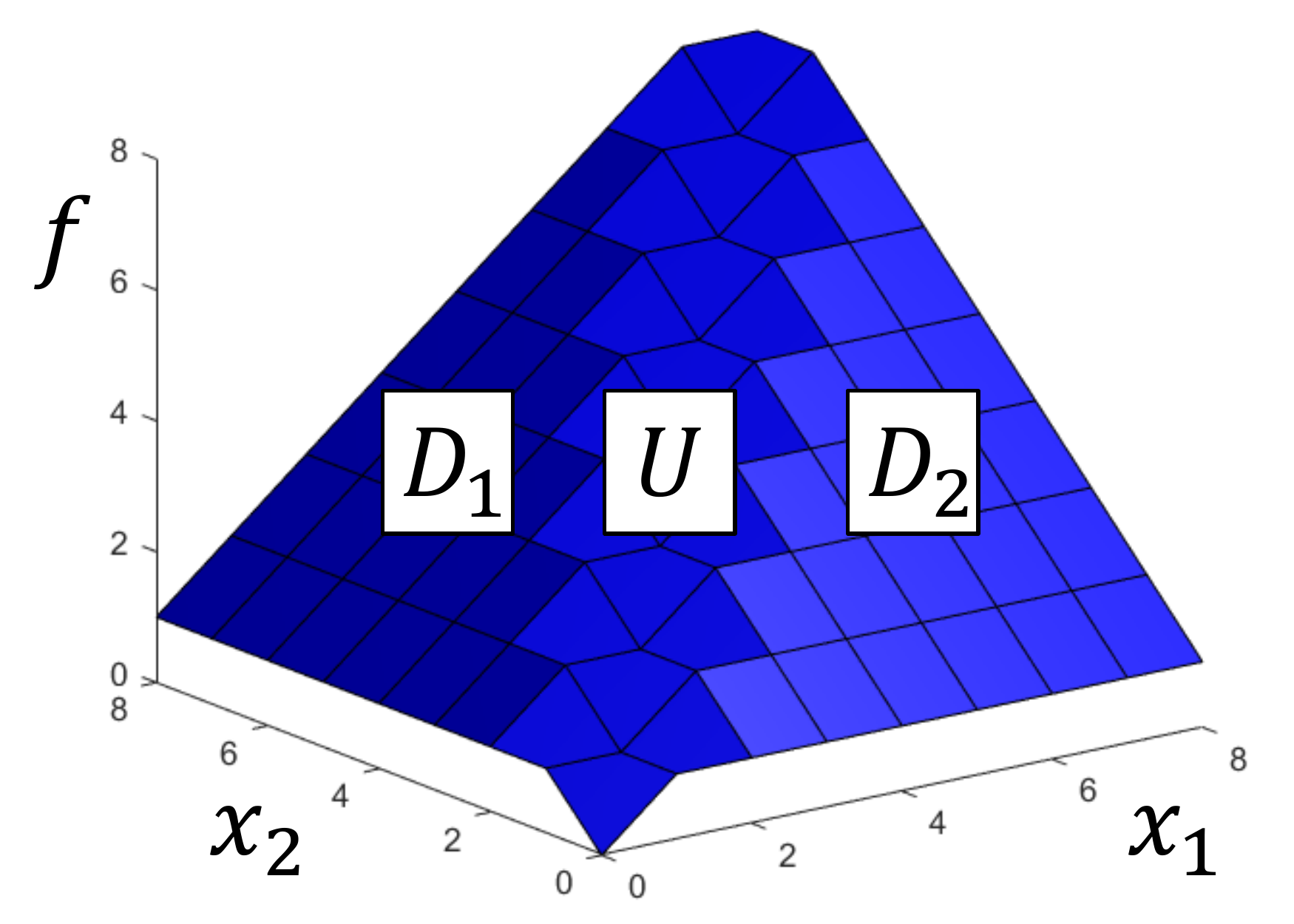}
        \caption{ \footnotesize
            Plot of $f$, whose domain has 3 regions: $D_1$, $D_2$, and $U$.
        }
        \label{fig:house-needs-roof}
    \end{subfigure}
    \qquad
    \begin{subfigure}[t]{0.21\textwidth}
        \centering
        \includegraphics[width=\textwidth]{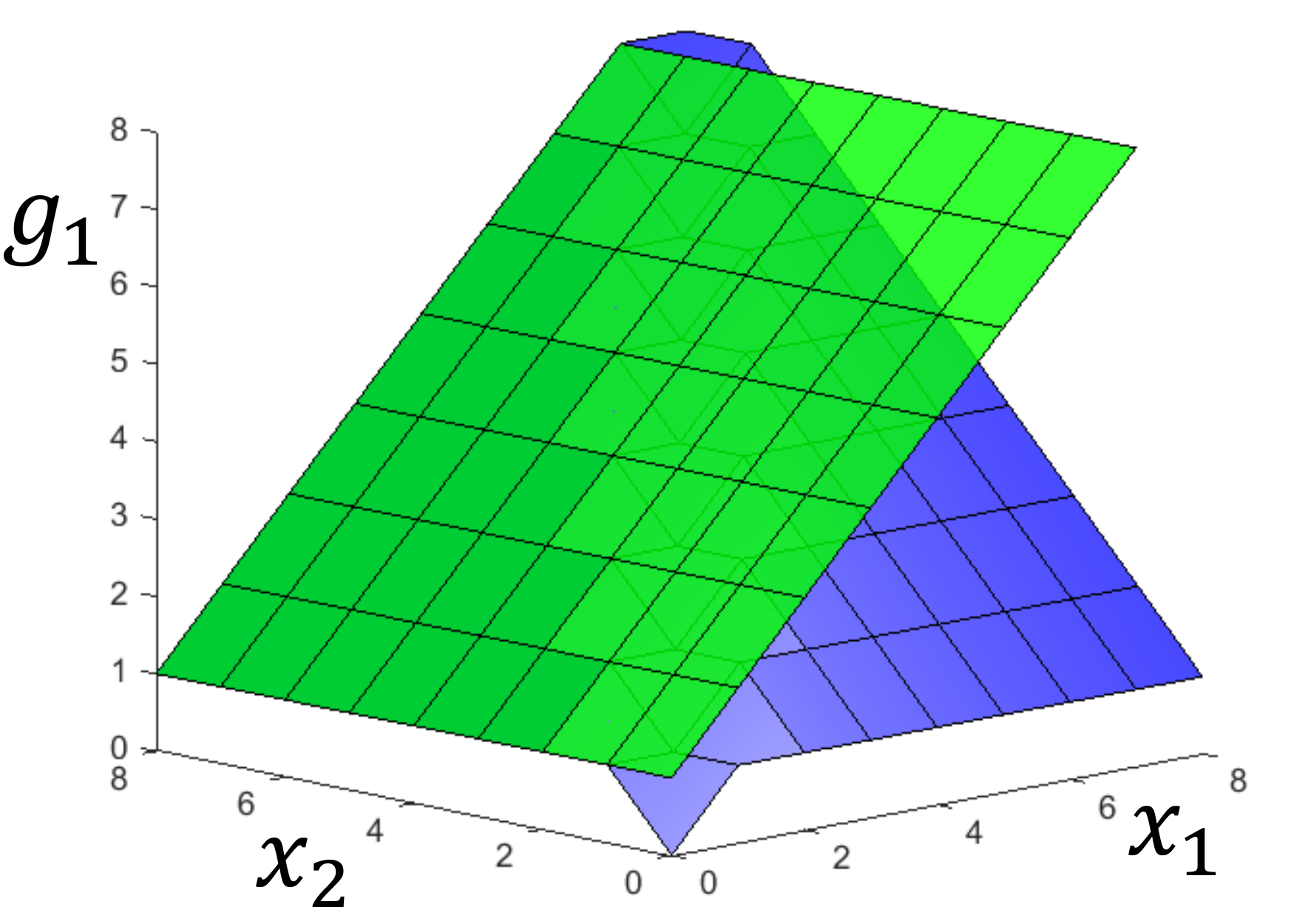}
        \caption{ \footnotesize
            $g_1$ (green) is the unique quilt-affine extension from region $D_1$.
        }
        \label{fig:wall-on-house}
    \end{subfigure}
    
    \vspace{0.5cm}
    \begin{subfigure}[t]{0.21\textwidth}
        \centering
        \includegraphics[width=\textwidth]{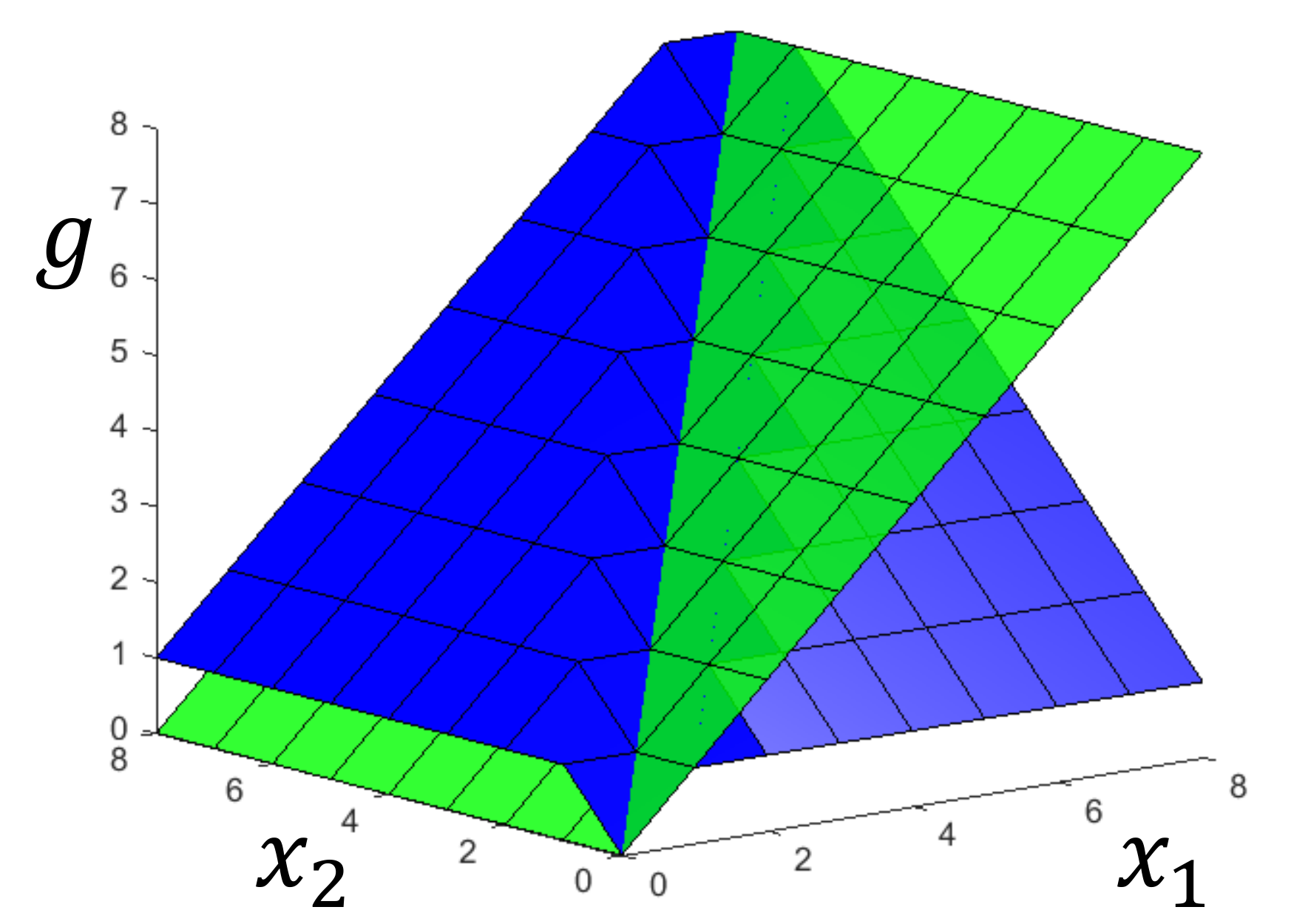}
        \caption{ \footnotesize
            $g$ (green) is a quilt-affine extension from $U$, but $g<f$ on $D_1$.
        }
        \label{fig:bad-roof-on-house}
    \end{subfigure}
        \qquad
    \begin{subfigure}[t]{0.21\textwidth}
        \centering
        \includegraphics[width=\textwidth]{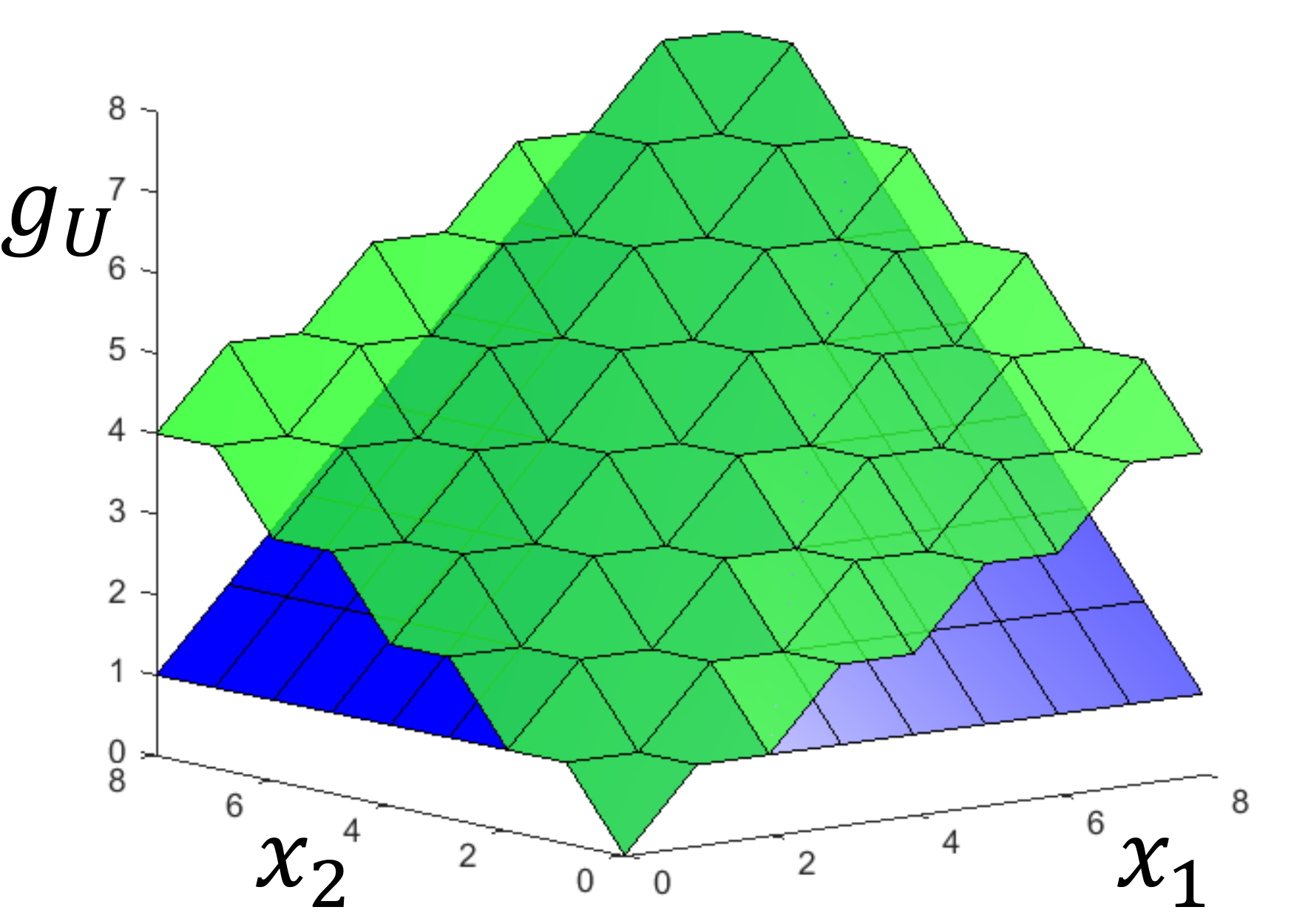}
        \caption{ \footnotesize
            $g_U$ (green) is a quilt-affine extension from $U$, and $g_U\geq f$.
        }
        \label{fig:roof-on-house}
    \end{subfigure}
    \caption{ \footnotesize
        Obliviously-computable $f$ can be expressed as a min of quilt-affine functions.
        \todoi{make this bigger in full version}
    }
\end{figure}
}

\opt{full}{
\begin{figure}[t]
    \centering
    \begin{subfigure}[t]{0.4\textwidth}
        \centering
        \includegraphics[width=\textwidth]{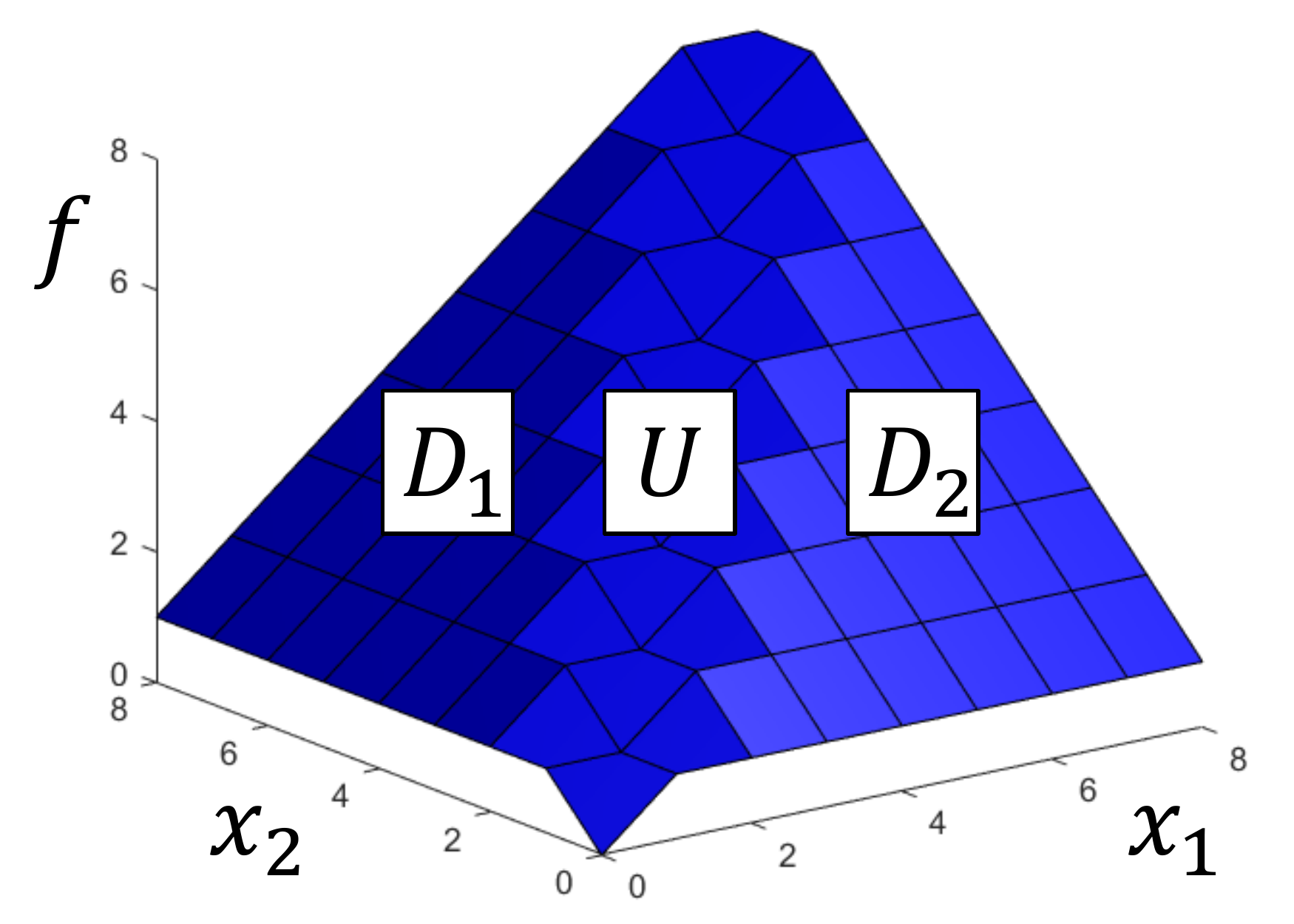}
        \caption{ \footnotesize
            Plot of $f$, whose domain has 3 regions: $D_1$, $D_2$, and $U$.
        }
        \label{fig:house-needs-roof}
    \end{subfigure}
    \qquad
    \begin{subfigure}[t]{0.4\textwidth}
        \centering
        \includegraphics[width=\textwidth]{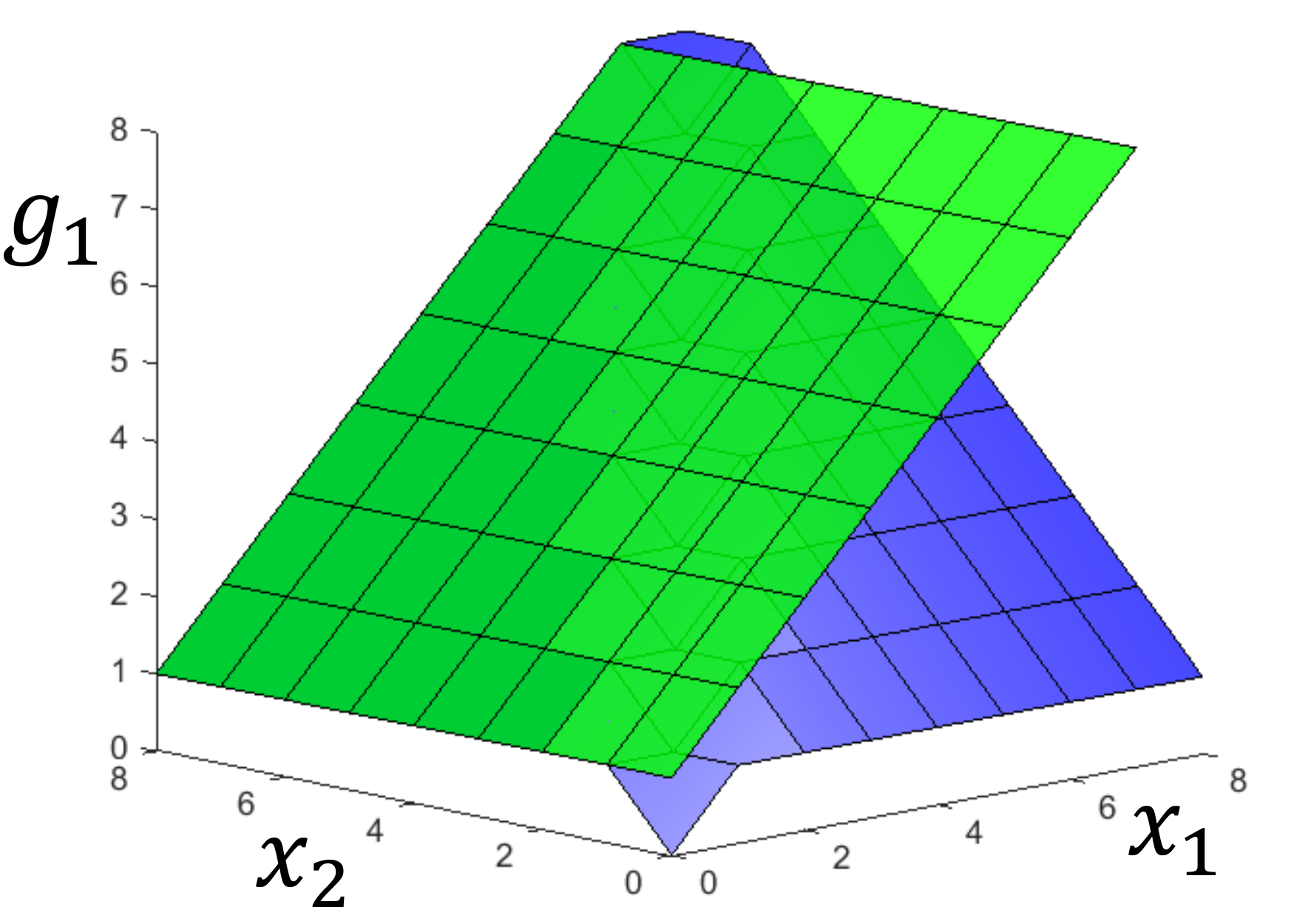}
        \caption{ \footnotesize
            $g_1$ (green) is the unique quilt-affine extension from region $D_1$.
        }
        \label{fig:wall-on-house}
    \end{subfigure}
    
    \vspace{0.5cm}
    \begin{subfigure}[t]{0.4\textwidth}
        \centering
        \includegraphics[width=\textwidth]{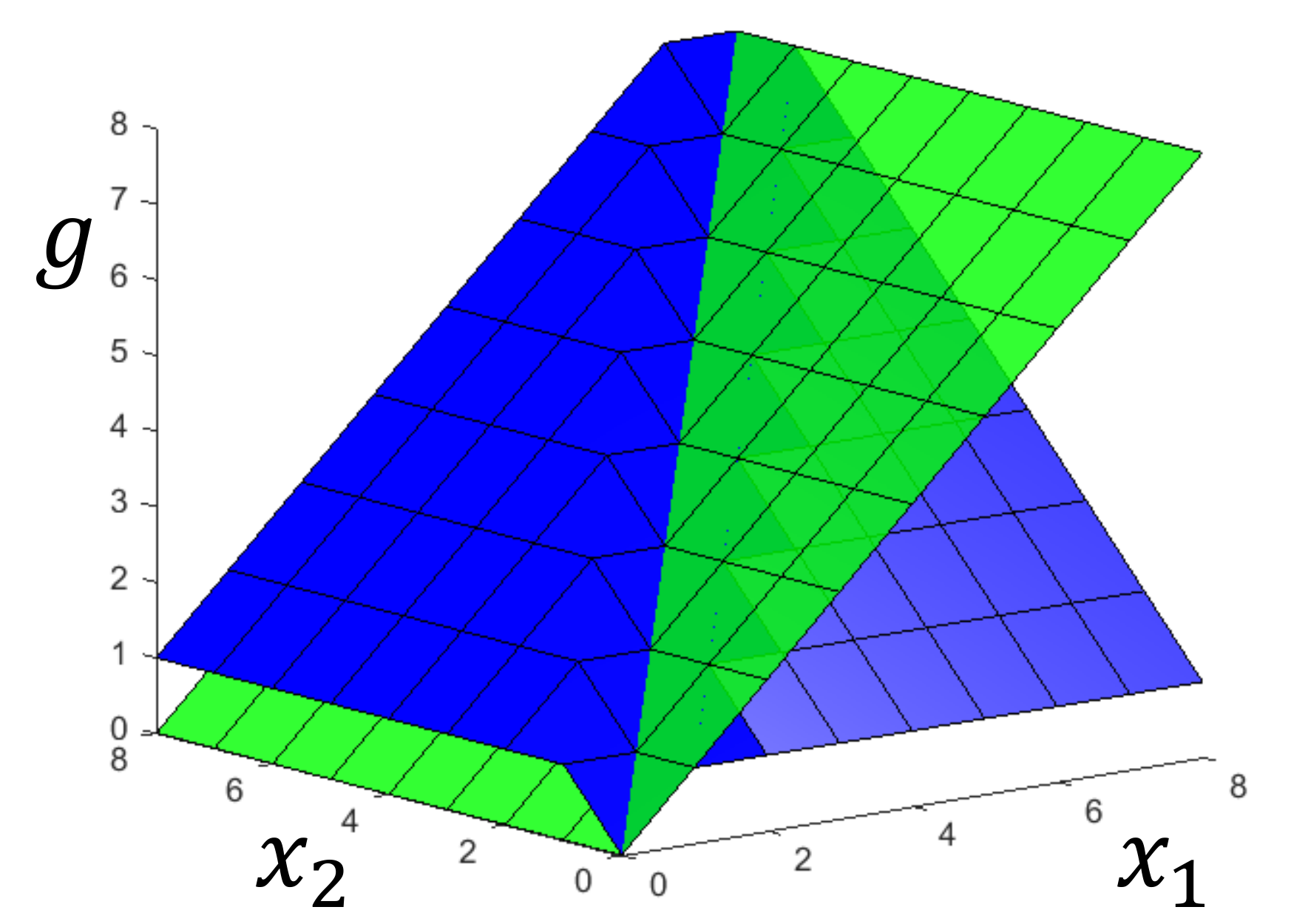}
        \caption{ \footnotesize
            $g$ (green) is a quilt-affine extension from $U$, but $g<f$ on $D_1$.
        }
        \label{fig:bad-roof-on-house}
    \end{subfigure}
        \qquad
    \begin{subfigure}[t]{0.4\textwidth}
        \centering
        \includegraphics[width=\textwidth]{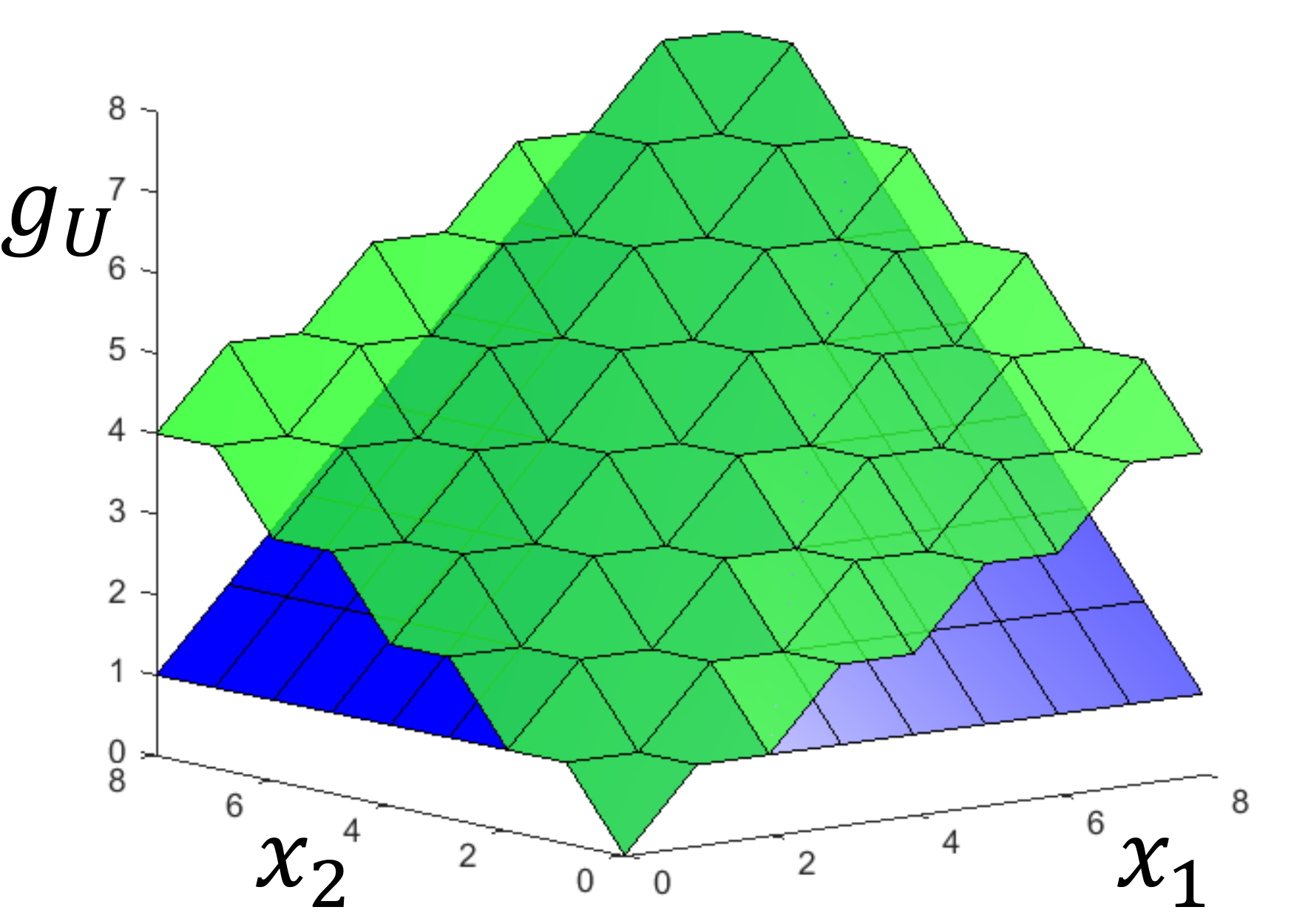}
        \caption{ \footnotesize
            $g_U$ (green) is a quilt-affine extension from $U$, and $g_U\geq f$.
        }
        \label{fig:roof-on-house}
    \end{subfigure}
    \caption{ \footnotesize
        Obliviously-computable $f$ can be expressed as a min of quilt-affine functions.
    }
\end{figure}
}

\begin{figure*}[t!]
    \centering
    \begin{subfigure}[t]{0.45\textwidth}
        \centering
        \includegraphics[width=0.7\textwidth]{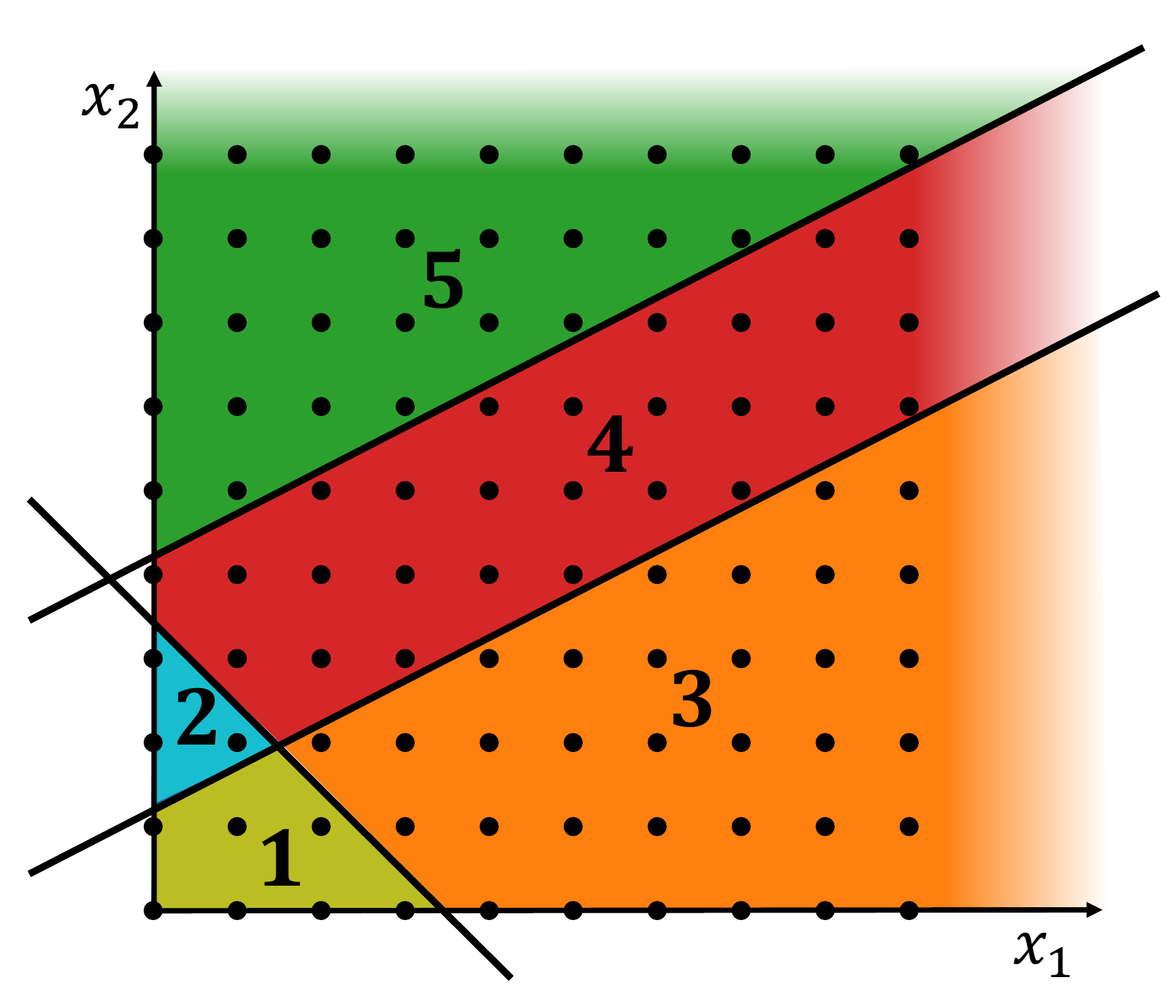}
        \caption{ \footnotesize
            Three threshold hyperplanes creating five regions. Regions 3 and 5 are determined, region 4 is under-determined but still eventual (unbounded in all input).
        }
        \label{fig:2d-regions}
    \end{subfigure}%
    \qquad
    \begin{subfigure}[t]{0.45\textwidth}
        \centering
        \includegraphics[width=0.7\textwidth]{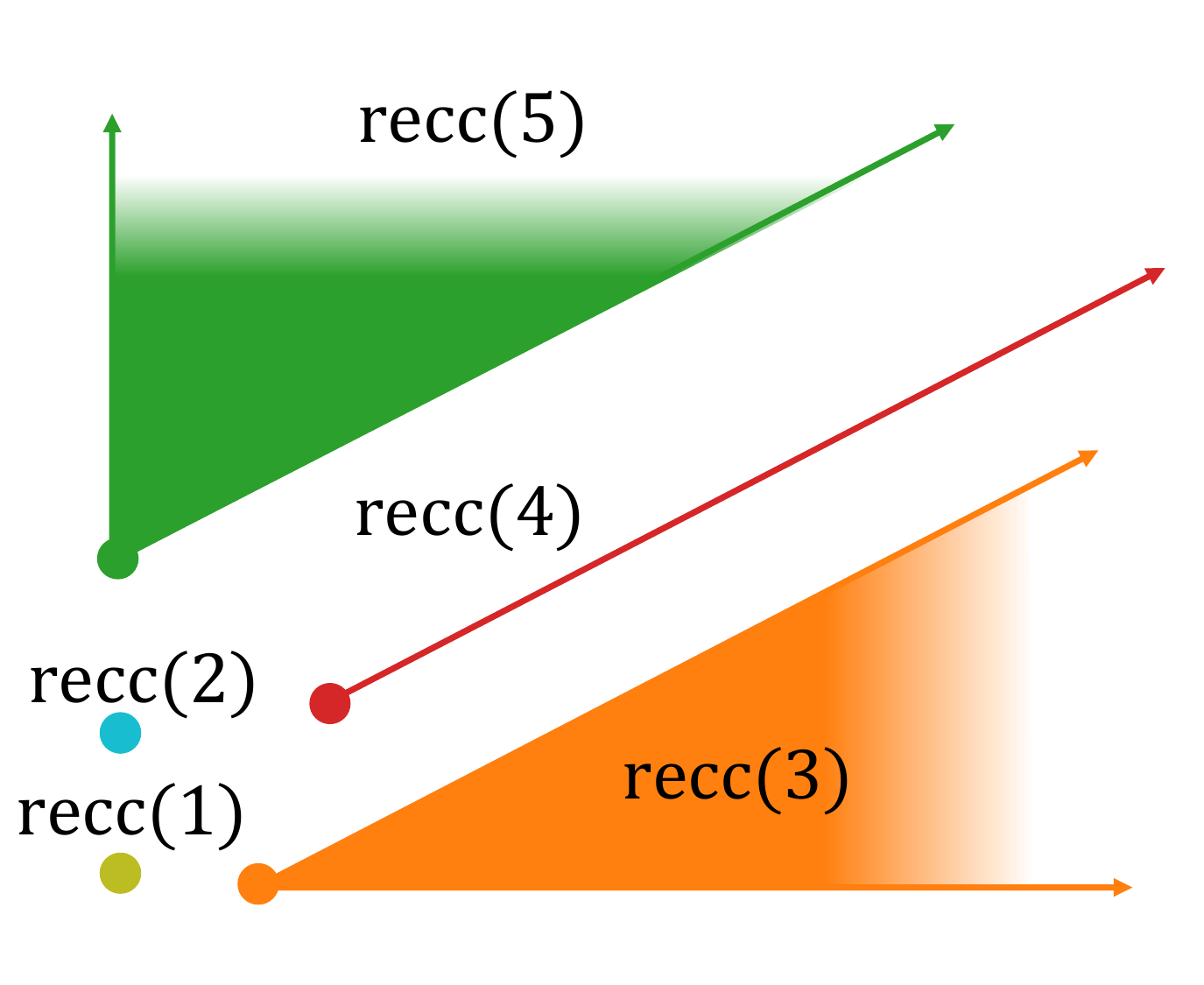}
        \caption{ \footnotesize
            The recession cones of all five regions. For finite regions, $\rec(1)=\rec(2)=\{\vec{0}\}$. Under-determined region 4 has a 1D recession cone, determined regions 3 and 5 have 2D recession cones.
        }
        \label{fig:2d-recession-cones}
    \end{subfigure}%
    \qquad
    \begin{subfigure}[t]{0.45\textwidth}
        \centering
        \includegraphics[width=0.9\textwidth]{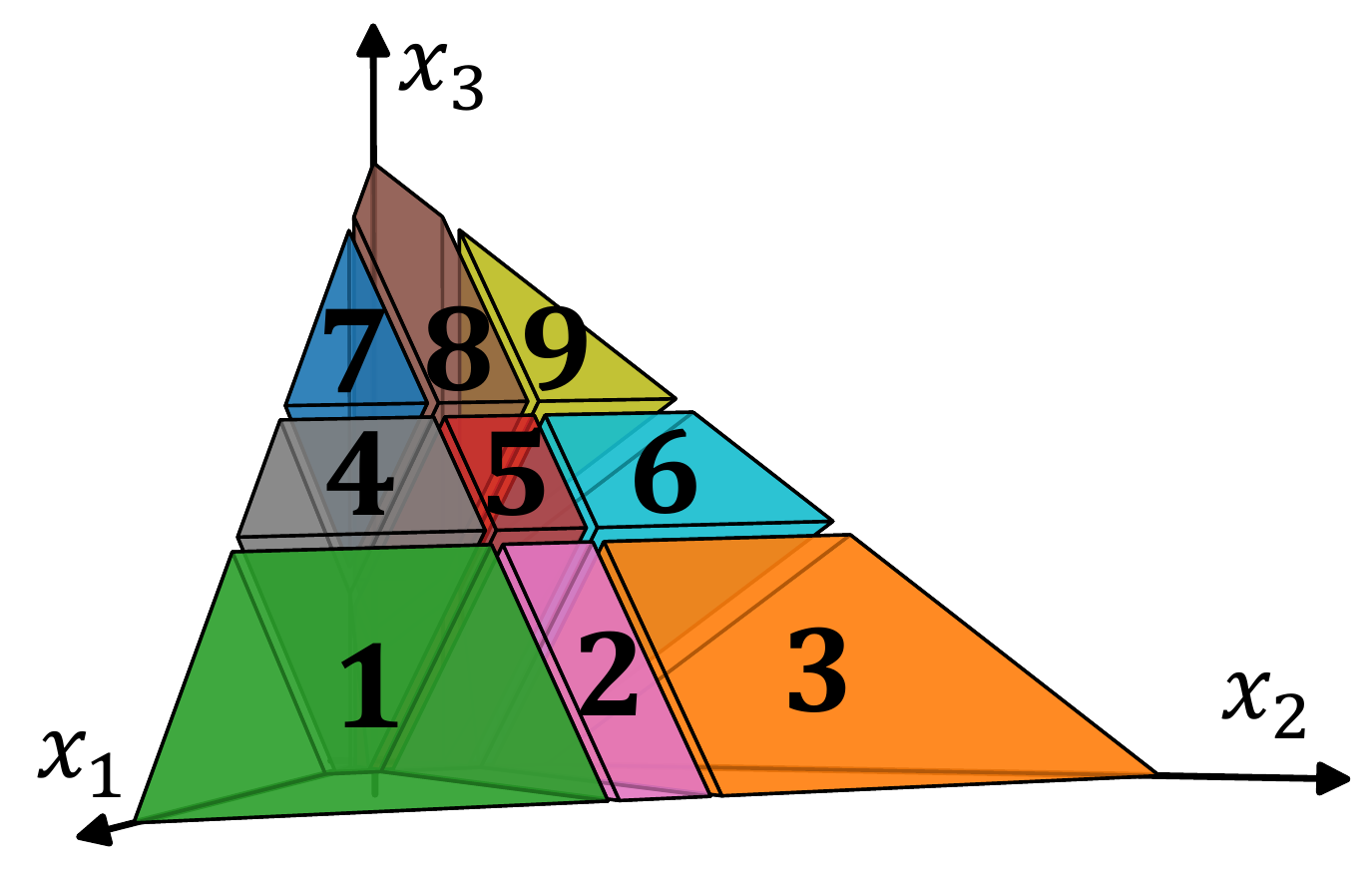}
        \caption{ \footnotesize
            Two pairs of parallel threshold hyperplanes creating nine eventual regions. Regions 1,3,7,9 are determined. Region 5 is under-determined with 1D recession cone. Regions 2,4,6,8 are under-determined with 2D recession cones.
        }
        \label{fig:3d-regions}
    \end{subfigure}
        \qquad
    \begin{subfigure}[t]{0.45\textwidth}
        \centering
        \includegraphics[width=0.9\textwidth]{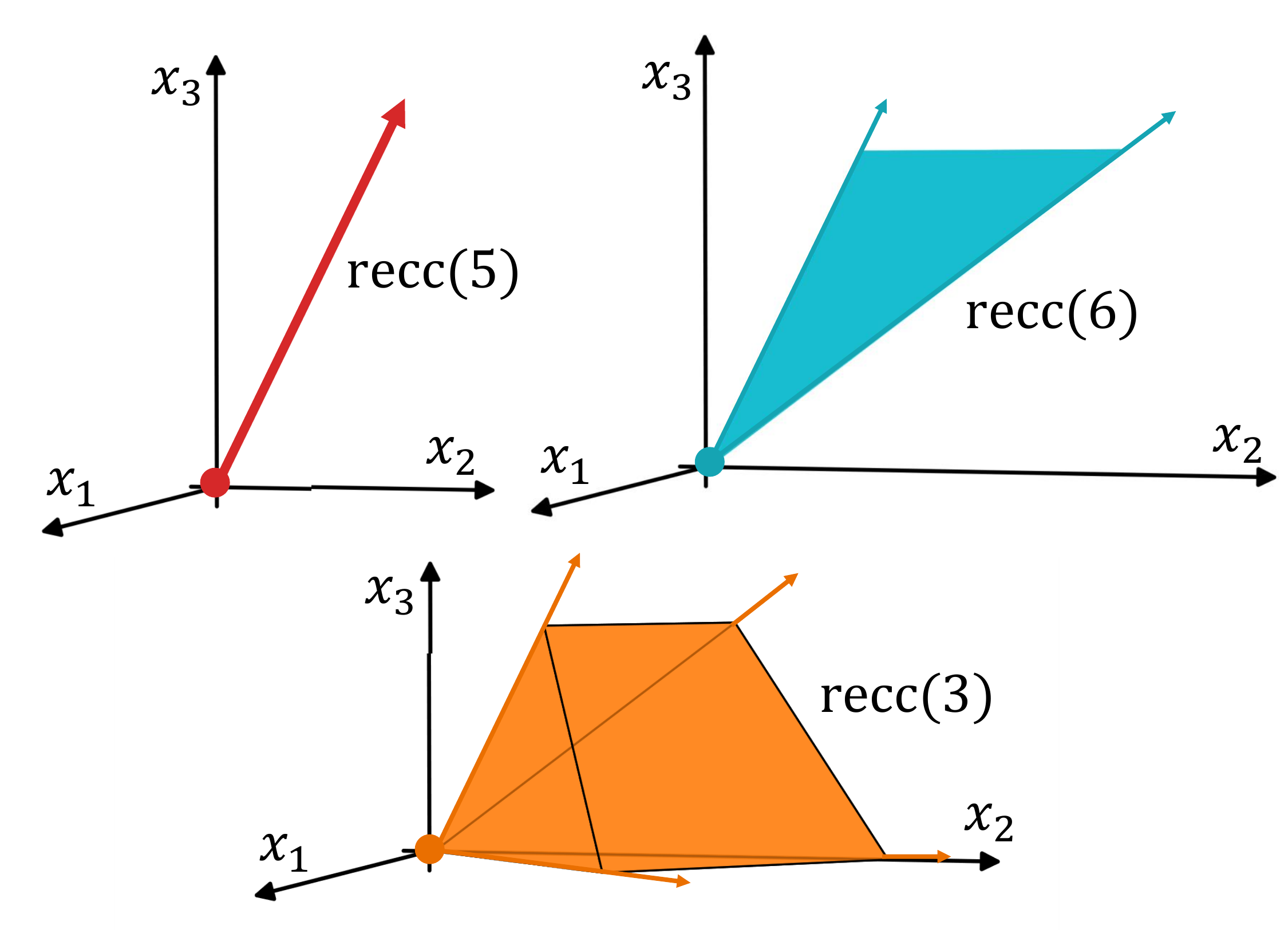}
        \caption{ \footnotesize
            $\rec(5)\subseteq\rec(6)\subseteq\rec(3)$ so region $3$ is a determined neighbor of under-determined region 5 and under-determined region 6. Also, region 6 is a neighbor of region 5.
        }
        \label{fig:3d-recession-cones}
    \end{subfigure}
    \caption{ \footnotesize
        Examples with domains $\N^2$ (top) and $\N^3$ (bottom), with threshold hyperplanes giving regions (left), which are classified by their recession cones (right).
    }
\end{figure*}

As a simple motivating example, consider the semilinear, nondecreasing function
\[
    f(x_1,x_2)=\begin{cases}
    x_1 + 1, & \text{if $x_1<x_2$ (region $D_1$)}\\
    x_2 + 1, & \text{if $x_1>x_2$ (region $D_2$)}\\
    x_1 & \text{if $x_1=x_2$ (region $U$)}
    \end{cases}
\]
As in Definition~\ref{semilinear function def}, $f$ is piecewise-affine, with semilinear domains that happen to be only defined by threshold sets. These thresholds then partition the domain into three regions: $D_1$, $D_2$, and $U$ (see Fig.~\ref{fig:house-needs-roof}). For region $D_1$, there is a unique quilt-affine extension $g_1(x_1,x_2)=x_1 + 1$ (note an affine function is the special case of a quilt-affine function with period $1$). Also, $g_1$ eventually dominates $f$ as desired, since $g_1(\vec{x})\geq f(\vec{x})$ for all $\vec{x}\in\N^2$ (see Fig.~\ref{fig:wall-on-house}).
By symmetry, we have the same for region $D_2$ and its extension $g_2(x_1,x_2)=x_2+1$.

These desirable properties follow from $D_1$ and $D_2$ being ``wide'' regions that we define to be \emph{determined} (formalized later). On the other hand, $U$ is a ``narrow'' region that is \emph{under-determined}. As a result, there is not a unique quilt-affine extension from $U$. For example, $g(x_1,x_2)=x_1$ is a quilt-affine extension, however, we do not have $g(\vec{x})\geq f(\vec{x})$ for all sufficiently large $\vec{x}$ (see Fig.~\ref{fig:bad-roof-on-house}).

In order to identify a quilt-affine extension from $U$ that does eventually dominate $f$, we will refer to the unique extensions $g_1$ and $g_2$ from regions $D_1$ and $D_2$, which are \emph{neighbors} of $U$ (formalized later). We can construct a quilt-affine function with a gradient $(\frac{1}{2},\frac{1}{2})$ that is the average of the gradients $(1,0)$ of $g_1$ and $(0,1)$ of $g_2$. In particular, we can let $g_U(x_1,x_2)=\ceil{\frac{x_1+x_2}{2}}$ (note this is a quilt-affine function with period 2, see Fig.~\ref{fig:roof-on-house}). 
We then have $f(\vec{x})=\min\qty[g_1(\vec{x}),g_2(\vec{x}),g_U(\vec{x})]$ for all $\vec{x}\geq\vec{n}=\vec{0}$ as guaranteed by Theorem~\ref{MC implies min}.

We now describe how we formalize the notion of a \emph{determined region}, \emph{under-determined region}, and \emph{neighbor}, for the general case of domain $\N^d$, where the regions are convex polyhedra in $\R^d$.

\opt{full,sub}{
    \noindent {\bf Section \ref{subsec-determined-regions}.}
}
To formally define determined regions, we identify the \emph{recession cone} $\rec(R) \subseteq \R^d$ of each region $R$: 
the set of vectors along infinite rays in $R$~\cite{RockafellerConvexAnalysis}
(see Fig.~\ref{fig:2d-recession-cones} and Fig.~\ref{fig:3d-recession-cones}).
A determined region $D$ is defined as having a $d$-dimensional recession cone (see regions 3 and 5 in Fig.~\ref{fig:2d-regions} and regions 1,3,7,9 in Fig.~\ref{fig:3d-regions}). For determined regions, we can prove
\opt{full}{ (see Lemmas \ref{determined region has unique extension} and \ref{determined extension eventually dominates f})}
 there is a unique quilt-affine extension, which eventually dominates $f$.


\opt{full,sub}{
    \noindent {\bf Section \ref{subsec-under-determined-regions}.}
}
Under-determined regions are then defined as having a recession cone with dimension $<d$ (see regions 1,2,4 in Fig.~\ref{fig:2d-regions} and regions 2,4,5,6,8 in Fig.~\ref{fig:3d-regions}). 
The above arguments do not work for under-determined regions. Instead, identify the \emph{neighbors} of an under-determined region $U$ as regions $R$ with $\rec(U)\subseteq\rec(R)$ (see Fig.~\ref{fig:2d-recession-cones} and Fig.~\ref{fig:3d-recession-cones}). 
We consider the neighbors of $U$ that are determined regions.
The possible behavior of $f$ on $U$ is constrained by the unique extensions from these regions, and we can define an extension from $U$ based on an averaging process.
\opt{final}{Formal definitions and a proof of Theorem~\ref{MC implies min} appear in \cite{severson2019composable}}.
\opt{full}{ (See Lemma \ref{strip extension by averaging}).}

\opt{full}{
\subsection{Domain Decomposition}
\label{subsec-sunshine-decomposition}

To identify the quilt-affine components $g_k$, the strategy will be to partition the domain $\N^d$ into regions where the restriction of $f$ to that region yields a quilt-affine function. 

By Lemma $\ref{Semilinear Are Stably Computable Lemma}$, $f$ is semilinear, so by Definition \ref{semilinear function def}, $f$ is the union of affine partial functions, whose disjoint domains are semilinear subsets of $\N^d$. This representation is not unique, so we now fix some arbitrary such representation of $f$. Recall by Definition \ref{semilinear set def}, these semilinear domains are finite Boolean combinations of threshold and mod sets, so consider the collection $\mathcal{T}$ of all threshold sets and the collection $\mathcal{M}$ of all mod sets that defined any of these semilinear domains.

Let $\mathcal{T}$ consist of $l$ threshold sets $\{x\in\N^d:\vec{t}_i\cdot\vec{x}\geq h_i\}$ for each $i=1,\ldots,l$, where $\vec{t}_i\in\Z^d$ and $h_i\in\Z$. These thresholds are equivalently written $2\vec{t}_i\cdot\vec{x}>2h_i-1$ (since $\vec{t}_i\cdot\vec{x}\geq h_i \iff \vec{t}_i\cdot\vec{x} > h_i - \frac{1}{2}$),
so we can assume without loss of generality that the boundary hyperplanes $H_i=\{\vec{x}\in\R^d:\vec{t}_i\cdot\vec{x}=h_i\}$ contain no integer points (see Fig.~\ref{fig:2d-regions}). These hyperplanes then partition the domain $\N^d$. For some $\vec{y}\in\N^d$, let $s_i=\mathrm{sign}(\vec{t}_i\cdot\vec{y}-h_i)=\pm1$ for each $i=1,\ldots,l$. Defining the threshold matrix $T\in\Z^{l\times k}$, offset vector $\vec{h}\in\Z^l$, and diagonal sign matrix $S\in\Z^{l\times l}$ as
$$T=\begin{pmatrix}
\vec{t}_1^\mathsf{T}\\
\vdots\\
\vec{t}_l^\mathsf{T}
\end{pmatrix}
\quad
\vec{h}=
\begin{pmatrix}
h_1\\
\vdots\\
h_l
\end{pmatrix}
\quad
S=
\begin{pmatrix}
s_1 & 0 & \ldots & 0 \\
0 & s_2 & \ldots & 0 \\
\vdots & \vdots & \ddots & \vdots\\
0 & 0 & \ldots & s_l
\end{pmatrix}
$$
then $S(T\vec{y}-\vec{h})\geq\vec{0}$. This concise form will let us define the region of points that are in precisely the same threshold sets as $\vec{y}$
(those that agree on the signs of the components of $T\vec{y}-\vec{h}$):

\begin{definition}
\label{def-region}
Let $S$ be a sign matrix: a diagonal matrix with diagonal entries $=\pm1$. Then the \emph{region (induced by $S$)} is defined as
\[
    R=\{\vec{x}\in\R^d_{\geq 0}:S(T\vec{x}-\vec{h})\geq\vec{0}\}.
\]
When referring to a region $R$, we use $S_R$ to denote the sign matrix that induces $R$.
\end{definition}

We consider nonnegative real vectors $\vec{x}\in\R^d$ rather than just $\vec{x}\in\N^d$, but we are only truly concerned with the integer points $R\cap\N^d$, and only consider regions where $R\cap\N^d\neq\emptyset$. Also, since each $\vec{y}\in\N^d$ induces a unique sign matrix as shown above, it follows that every $\vec{y}\in\N^d$ is contained in some unique region. The reason to consider $R\subset\R^d$ is that each region $R$ is a convex polyhedron, with convenient properties from convex geometry (see Fig.~\ref{fig:2d-regions} and Fig.~\ref{fig:3d-regions}).

Now consider the collection $\mathcal{M}$, consisting of $m$ mod sets $\{\vec{x}\in\N^d:\vec{a_i}\cdot\vec{x}\equiv b_i\mod c_i\}$ for each $i=1,\dots,m$, where $\vec{a_i}\in\Z^d,b_i\in\Z,c_i\in\N_+$. Then let the \emph{global period} $p$ be the least common multiple $\mathrm{lcm}_i(c_i)$, so all elements of a congruence class $\bara\in\Z^d/p\Z^d$ are contained in precisely the same mod sets. Thus for a region $R$, the set $R\cap\bara$ is contained in precisely the same threshold and mod sets, so the restriction $f|_{R\cap\bara}$ is an affine partial function.

We now summarize this decomposition as a characterization of a semilinear function.
Note that this applies to all semilinear functions, even those that are decreasing.

\begin{lemma}
\label{sunshine lemma}
Let $f:\N^d\to\N$ be a semilinear function. Then there exist a finite set of \emph{regions} $R_1,\ldots,R_n \subseteq \R^d$ and a \emph{global period} $p\in\N_+$ such that for each region $R_i$ and congruence class $\bara\in\Z^d/p\Z^d$, there exist $\nabla_{R_i,\bara}\in\Q^d$ and $b_{R_i,\bara}\in\Q$ such that the restriction of $f$
defined for all $\vec{x} \in R_i \cap \bara$ by
\[
f|_{R_i\cap\bara}(\vec{x})=\nabla_{R_i,\bara}\cdot\vec{x}+b_{R_i,\bara}
\]
is a (rational) affine partial function.
\end{lemma}

Notice the similarity in form between these affine partial functions in Lemma~\ref{sunshine lemma} and Definition~\ref{quilt-affine def} of quilt-affine functions. In fact, we will show the behavior of $f$ on some regions has a unique quilt-affine structure. This will require the region to be ``infinitely wide in all directions'', a notation we now make precise to define these determined regions.

}
\opt{full}{
\subsection{Determined Regions}
\label{subsec-determined-regions}
To formally define determined regions, we must first make the connection between regions and their \textit{recession cones} (more information on recession cones can be found in \cite{RockafellerConvexAnalysis}).

\begin{definition}
\label{def:recession-cone}
For a region $R$, the \emph{recession cone} of $R$ is
\[
    \rec(R)=\{\vec{y}\in\R^d:\vec{x}+\lambda\vec{y}\in R \text{ for all }\vec{x}\in R,\lambda\in\R_{\geq0}\}.
\]
\end{definition}

The recession cone corresponds to the directions ($\vec{y}$) one can proceed infinitely in a region $R$, and is a convex polyhedral cone (recall that a subset of $\R^d$ is a \emph{cone} if it is closed under positive scalar multiplication) (see Fig.~\ref{fig:2d-recession-cones} and Fig.~\ref{fig:3d-recession-cones}).

Recall by Definition \ref{def-region} a region $R=\{\vec{x}\in\R^d_{\geq 0}:S(T\vec{x}-\vec{h})\geq\vec{0}\}$. It is possible to equivalently define
\[
    \rec(R)=\{\vec{y}\in\R^d_{\geq 0}:S_RT\vec{y}\geq \vec{0}\}.
\]
In other words, the recession cone is defined by the homogenized version of the same inequalities that defined $R$. One can easily verify this is equivalent to Definition \ref{def:recession-cone}.

We then say a region $R$ is \emph{determined} if $\dim \rec(R)=d$. Otherwise, the region $R$ is \emph{under-determined}. We now make precise the idea that a determined region with full-dimensional recession cone is ``infinitely wide in all directions'':

\begin{lemma}
\label{determined region has large balls}
Let $D$ be a determined region. Then the recession cone $\rec(D)$ contains open balls of arbitrarily large radius.
\end{lemma}

\begin{proof}
Since $\rec(D)$ is a $d$-dimensional convex polyhedron, it has nonempty interior, so there exists some $\vec{x}\in \mathrm{int}(\rec(D))$ and an open ball $B_\epsilon(\vec{x})\subset\rec(D)$ of radius $\epsilon$ around $\vec{x}$ contained in the cone. 
Since recession cones are closed under positive scalar multiplication,
for any positive scalar $c$, the ball $B_{c\epsilon}(c\vec{x})\subset\rec(D)$.
\end{proof}

We next make precise the idea that the function has a unique quilt-affine structure on a determined region.

\begin{definition}
An \emph{extension} $g:\N^d\to\Z$ (of $f$) from a region $R$ is a quilt-affine function that agrees with $f$ on $R$: 
$f(\vec{x})=g(\vec{x})$ for all integer $\vec{x}\in R\cap\N^d$.
\end{definition}

We will now show there is a unique extension from each determined region.
The construction of the regions yields a periodic piecewise-affine structure $f$ restricted to a region (Lemma \ref{sunshine lemma}). In order for $f$ to be nondecreasing, these affine gradients must all agree, which will let us uniquely describe a quilt-affine extension $g$ using Definition \ref{quilt-affine def}.

\begin{lemma}
\label{determined region has unique extension}
There is a unique extension $g:\N^d\to\Z$ from any determined region $D$.
\end{lemma}


\begin{proof}
By Definition \ref{quilt-affine def}, 
it suffices to show that there is a $\vec{\nabla}_g\in\Q^d$ and $B:\Z^d/p\Z^d\to\Q$ such that, defining $g:\N^d \to \Z$ for all $\vec{x} \in \N^d$ as
\[
    g(\vec{x})=\vec{\nabla}_g\cdot\vec{x}+B(\barx\mod p),
\]
then for all $\vec{x} \in D \cap \N^d$, $g(\vec{x}) = f(\vec{x})$.\footnote{Note that while $p$ is not necessarily unique, any valid choice of $p$ will result in the same function $g$.}

By Lemma \ref{sunshine lemma}, for each $\bara\in\Z^d/p\Z^d$,
there are $\vec{\nabla}_{\bara}\in\Q^d$ and $b_{\bara}\in\Q$ such that
\[
    f|_{D\cap\bara}(\vec{x})=\vec{\nabla}_{\bara}\cdot\vec{x}+b_{\bara}.
\]
Since $D$ contains arbitrarily large open balls by Lemma \ref{determined region has large balls}, it contains points in all congruence classes, so all $p^d$ constants $b_{\bara}$ are well-defined. Then we define the periodic offset function $B(\barx\mod p)=b_{\barx\mod p}$.
This is almost what we require to prove the lemma, 
except that $\vec{\nabla}_{\bara}$ depends on $\bara$.
It remains to show that all vectors $\vec{\nabla}_{\bara}$ are equal, so we can define the gradient $\vec{\nabla}_g=\vec{\nabla}_{\bara}$ for any $\bara\in\Z^d/p\Z^d$. 

Assuming for the sake of contradiction that $\vec{\nabla}_{\bara}\neq\vec{\nabla}_{\overline{\vec{b}}}$ 
for some equivalence classes $\bara,\overline{\vec{b}}$, 
we will show that $f$ cannot be nondecreasing. 
Since the recession cone $\rec(D)$ is full-dimensional and $\vec{\nabla}_{\bara}-\vec{\nabla}_{\overline{\vec{b}}}\neq \vec{0}$, there must exist some $\vec{v}\in\rec(D)$ such that $\vec{\nabla}_{\bara}\cdot\vec{v}\neq\vec{\nabla}_{\overline{\vec{b}}}\cdot\vec{v}$. 
Furthermore, by density of the rationals, 
we can assume $\vec{v}\in\Q^d$, 
and by scaling by the denominator we can assume $\vec{v}\in\N^d$. Now without loss of generality assume $\vec{\nabla}_{\bara}\cdot\vec{v}>\vec{\nabla}_{\overline{\vec{b}}}\cdot\vec{v}$.

Now pick some $\vec{y}\in D\cap \bara$, and some $\vec{z}\in D\cap\overline{\vec{b}}$ with $\vec{y}<\vec{z}$, 
which must exist because again $D$ contains arbitrarily large open balls by Lemma \ref{determined region has large balls}. 
But since $\vec{\nabla}_{\bara}\cdot\vec{v}>\vec{\nabla}_{\overline{\vec{b}}}\cdot\vec{v}$, moving along $\vec{v}$, 
the function values in $\bara$ grow faster than in $\overline{\vec{b}}$
when moving along $\vec{v}$ from $\vec{y}$ and $\vec{z}$,
respectively. Note that $\vec{y}+cp\vec{v},\vec{z}+cp\vec{v}\in D$ by definition of $\vec{v}\in\rec(D)$, and $cp\vec{v}\in p\Z^d$, so $\vec{y}+cp\vec{v}\in D\cap\bara$ and $\vec{z}+cp\vec{v}\in D\cap\overline{\vec{b}}$.
Thus for some multiple $cp$ of the period $p$, we must have
$$f(\vec{y}+cp\vec{v})>f(\vec{z}+cp\vec{v})$$
but then $f$ is not nondecreasing, 
since $\vec{y}+cp\vec{v}<\vec{z}+cp\vec{v}$.

Thus there is a uniquely determined gradient $\vec{\nabla}_g$. While there was not necessarily a unique choice for the period $p$, any valid choice will define the same function $g$.

\end{proof}

We next make precise why these extensions $g_k$ can be used in the eventual-min that will define $f$ to prove Theorem~\ref{MC implies min}.

\begin{definition}
An extension $g:\N^d\to\Z$ \emph{eventually dominates} $f$ if there exists $\vec{n}\in\N^d$ such that $f(\vec{x})\leq g(\vec{x})$ for all $\vec{x}\geq\vec{n}$.
\end{definition}

We now show the uniquely defined extension $g$ from a determined region $D$ eventually dominates $f$. The idea is that if the extension $g$ did not eventually dominate $f$, then we can apply Lemma~\ref{dickson contradiction lemma} to show $f$ is not obliviously-computable. 

For example, the function
\[
\max(x_1,x_2)=
\begin{cases}
x_2, & \text{if } x_1\leq x_2\\
x_1, & \text{if } x_1 > x_2
\end{cases}
\]
is naturally identified by two determined regions, with unique extensions $g_1(\vec{x})=(0,1)\cdot\vec{x}$ and $g_2(\vec{x})=(1,0)\cdot\vec{x}$. These extensions do not eventually dominate $f$, and we already saw in Section~\ref{sec:max} how Lemma~\ref{dickson contradiction lemma} applies to max. Thus the following lemma generalizes this example, finding a ``contradiction sequence'' $(\vec{a}_1,\vec{a}_2,\ldots)$ to apply Lemma~\ref{dickson contradiction lemma} whenever a determined extension $g$ does not eventually dominate $f$:

\begin{lemma}
\label{determined extension eventually dominates f}
The unique extension $g$ from a determined region $D$ eventually dominates $f$.
\end{lemma}
\begin{proof}
Let $g$ be the extension from some determined region $D$. Assume toward contradiction that $g$ does not eventually dominate $f$, so for any point $\vec{n}\in\N^d$ there exists some ``bad point'' $\vec{b}\geq\vec{n}$ with $f(\vec{b})>g(\vec{b})$. We will use Lemma \ref{dickson contradiction lemma} to show this implies $f$ is not obliviously-computable. To satisfy the Lemma conditions, we construct an increasing sequence $(\vec{a}_1,\vec{a}_2,\ldots)\in\N^d$ such that for all $i<j$ there exists some $\vec{\Delta}_{ij}\in\N^d$ with
$$f(\vec{a}_i+\vec{\Delta}_{ij})-f(\vec{a}_i)>f(\vec{a}_j+\vec{\Delta}_{ij})-f(\vec{a}_j)$$

We will choose $(\vec{a}_1,\vec{a}_2,\ldots)\in D\cap\bara$ to be points in $D$ that are all in the same congruence class $\bara\mod p$, and a sequence of vectors $(\vec{v}_1,\vec{v}_2,\ldots)\in\N^d$ such that for all $i$, $\vec{a}_i+\vec{v}_i$ is a ``bad point'': $f(\vec{a}_i+\vec{v}_i)>g(\vec{a}_i+\vec{v}_i)$, while for all $j>i$, $\vec{a}_j+\vec{v}_i\in D$, so $f(\vec{a}_j+\vec{v}_i)=g(\vec{a}_j+\vec{v}_i)$. Then for $i<j$, let $\vec{\Delta}_{ij}=\vec{v}_i$, so 
\begin{align*}
    f(\vec{a}_i+\vec{\Delta}_{ij})-f(\vec{a}_i)&>g(\vec{a}_i+\vec{v}_i)-g(\vec{a}_i)\quad\text{since $\vec{a}_i+\vec{v}_i$ is a ``bad point''}\\
    &=g(\vec{a}_j+\vec{v}_i)-g(\vec{a}_j)\quad\text{since $g$ is quilt-affine and $\vec{a}_i\equiv\vec{a}_j\mod p$}\\
    &=f(\vec{a}_j+\vec{\Delta}_{ij})-f(\vec{a}_j)
\end{align*}
as desired, where $f=g$ for points in $\vec{a}_i,\vec{a}_j,\vec{a}_j+\vec{v}_i\in D$. 

We now construct the sequences $(\vec{a}_1,\vec{a}_2,\ldots)$ and $(\vec{v}_1,\vec{v}_2,\ldots)$ recursively, to ensure for all $i<j$ that $\vec{a}_i\in\bara$, $\vec{a}_i+\vec{v}_i$ is a ``bad point'', and $\vec{a}_j+\vec{v}_i\in D$. Let $\vec{a}_1\in D$ arbitrarily, and the fixed congruence class $\bara=\overline{\vec{a}_1}$. For each $\vec{a}_i$, there must be a ``bad point'' above $\vec{a}_i$, so we recursively define each $\vec{v}_i$ based on $\vec{a}_i$ such that $\vec{a}_i+\vec{v}_i$ is this ``bad point'': $f(\vec{a}_i+\vec{v}_i)>g(\vec{a}_i+\vec{v}_i)$.

Now we recursively define each $\vec{a}_j$ based on $\vec{a}_{j-1}$ and $\vec{v}_1,\ldots,\vec{v}_{j-1}$ to ensure the sequence $(\vec{a}_1,\vec{a}_2,\ldots)$ is increasing, congruent, and $\vec{a}_j+\vec{v}_i\in D$ for all $i<j$. Since the recession cone $\rec(D)$ contains open balls of arbitrary radius by Lemma \ref{determined region has large balls}, we can find a point $\vec{a}_j$ in the recession cone above $\vec{a}_{j-1}$ in the same congruence class $\bara$ such that the open ball $B_r(\vec{a}_j)\subset D$ for a large radius $r\geq\max_i|\vec{v}_i|$. This gives the desired condition $\vec{a}_j+\vec{v}_i\in D$ for all $i<j$.

By the above analysis, this increasing sequence $(\vec{a}_1,\vec{a}_2,\ldots)$ with $\vec{\Delta}_{ij}=\vec{v}_i$ satisfies the conditions of Lemma \ref{dickson contradiction lemma}, giving the contradiction that $f$ is not obliviously-computable.
\end{proof}

The results of Lemmas \ref{determined region has unique extension} and \ref{determined extension eventually dominates f} bring us close to proving Theorem \ref{MC implies min}. From each determined region $D_1,\ldots,D_q$ we have a quilt-affine extension $g_1,\ldots,g_q:\N^d\to\Z$ that all eventually dominate $f$, so for some large enough $\vec{n}\in\N^d$ we have $f(\vec{x})\leq\min_k g_k(\vec{x})$ for all $\vec{x}\geq\vec{n}$. Furthermore, $f(\vec{x})=g_k(\vec{x})$ if $\vec{x}\in D_k$ for some determined region $D_k$. However, it is possible $f(\vec{x})<\min_k g_k(\vec{x})$ for any $\vec{x}$ that belong to an under-determined region. Since the bound $\vec{n}$ can be arbitrarily large, we need only consider \emph{eventual} under-determined regions that are unbounded in all inputs:

\begin{definition}
A region $R$ is \emph{eventual} if for any $\vec{n}\in\N^d$, there exists some $\vec{x}\in\N^d\cap R$ such that $\vec{x}\geq\vec{n}$.
\end{definition}

To finish the proof of Theorem \ref{MC implies min}, it remains to show how to construct a quilt-affine extension $g_U$ from each eventual under-determined region $U$, where $g_U$ eventually dominates $f$.
}
\opt{full}{
\subsection{Under-determined regions}
\label{subsec-under-determined-regions}

Let $U$ be an under-determined eventual region. The eventual condition implies that there is some $\vec{v}\in\rec(U)$ strictly positive on all coordinates. 
Let $W=\mathrm{span}(\rec(U))$, which we call the \emph{determined subspace} of $U$, with $1\leq\dim W<d$.
For example, for the ``pizza slice'' shaped region 6 in Fig.~\ref{fig:3d-regions}, $W$ is a 2D subspace, and region 6 is ``infinitely wide'' within the directions of $W$.
This term reflects the fact
that although the extensions from $U$ are not unique,
their values are uniquely defined \emph{within $W$}.
(For a determined region, its determined subspace is all of $\R^d$.)


\begin{definition}
Region $R$ is a \emph{neighbor} of an under-determined region $U$ if $\rec(U)\subseteq \rec(R)$.
\end{definition}

We will construct the extension from $U$ by referencing the determined regions which are neighbors of $U$. Any under-determined eventual region will in fact have at least two determined neighbors (proved as Corollary \ref{cor-determined-neighbor-exists} to the later Lemma \ref{lem-neighbor-in-direction-z}). Geometrically, we can think of the under-determined recession cones as faces of each of the recession cones of the determined neighbors (see Fig.~\ref{fig:3d-recession-cones}).

Unlike in the proof of Lemma \ref{determined region has unique extension}, the affine partial functions defining $f$ within $U$ do not need to have equal gradients. However, these gradients will be equal projected onto the subspace $W$. We now show a stronger statement, that this common gradient (projected onto $W$) agrees with the gradient of the extension from any determined neighbor $D$. Intuitively, if these gradients disagreed within $W$, then moving along directions in $\rec(U)$, the differences between $f$ on $U$ and on $D$ become arbitrarily large, contradicting that $f$ must be nondecreasing.







\begin{lemma}
\label{projection of determined gradient is equal}
Let $U$ be an under-determined eventual region. Let $D$ be a determined neighbor of $U$, with unique extension $g(\vec{x})=\vec{\nabla}_g\cdot\vec{x}+B(\barx\mod p)$ given by Lemma~\ref{determined region has unique extension}. For any $\vec{u}\in U\cap\N^d$, consider the affine partial function $f|_{U\cap(\baru\mod p)}(\vec{x})=\nabla_{\baru}\cdot\vec{x}+b_{\baru}$ given from Lemma \ref{sunshine lemma}. Then
\[
\mathrm{proj}_W(\vec{\nabla}_{\baru})=\mathrm{proj}_W(\vec{\nabla}_g).
\]
\end{lemma}

\begin{proof}
This proof uses similar techniques to the proof of Lemma \ref{determined region has unique extension}: with two nonequal gradients, moving far enough along a recession cone direction to contradict the fact that $f$ is nondecreasing.

Assume toward contradiction that for some $\vec{u}\in U\cap\N^d$, $\mathrm{proj}_W(\vec{\nabla}_{\baru})\neq\mathrm{proj}_W(\vec{\nabla}_g)$. Then since $W=\mathrm{span}(\rec(U))$, we have $\vec{\nabla}_{\baru}\cdot\vec{y}\neq\vec{\nabla}_{g}\cdot\vec{y}$ for some $\vec{y}\in\rec(U)$. Again, we can assume $\vec{y}\in\N^d$ by density of the rationals then scaling to clear denominators. Without loss of generality further assume $\vec{\nabla}_{\baru}\cdot\vec{y}>\vec{\nabla}_{g}\cdot\vec{y}$.

Now pick some $\vec{d}\in D\cap\N^d$ such that $\vec{d}\geq\vec{u}$, so $f(\vec{d})\geq f(\vec{u})$ because $f$ is nondecreasing. Then since $\vec{y}\in\rec(U)\subset\rec(D)$, for any $c\in\N$, $\vec{u}+cp\vec{y}\in U\cap\baru$ and $\vec{d}+cp\vec{y}\in D\cap\overline{\vec{d}}$. Since $f|_{U\cap\baru}(\vec{x})=\vec{\nabla}_{\baru}\cdot\vec{x}+b_{\baru}$ and $f|_{D\cap\overline{\vec{d}}}(\vec{x})=\vec{\nabla}_{g}\cdot\vec{x}+B(\overline{\vec{d}})$, where $\vec{\nabla}_{\baru}\cdot\vec{y}>\vec{\nabla}_{g}\cdot\vec{y}$, for some large enough $c\in\N$, we have $f(\vec{u}+cp\vec{y})>f(\vec{d}+cp\vec{y})$. But this contradicts the fact that $f$ is nondecreasing, since $\vec{u}\leq\vec{d}$.
\end{proof}

Lemma \ref{projection of determined gradient is equal} constrains the behavior of $f$ moving within the subspace $W$. The region $U$, however, could have a finite ``width'' in other directions. This motivates us to separate $U$ into ``strips'', partitioning its integer points to classes lying on translated versions of $W$:

\begin{definition}
Let $U$ be an under-determined region with $W=\mathrm{span}(\rec(U))$. The equivalence relation $\equiv_W$, where $\vec{x}\equiv_W\vec{y}$ if $\vec{x}-\vec{y}\in W$, partitions $U\cap\N^d$ into sets called \emph{strips}. Thus a strip $I=\{\vec{x}\in U\cap\N^d:\vec{x}\equiv_W\vec{u}\}$ for some $\vec{u}\in U\cap\N^d$.
\end{definition}

For a strip $I$, we will consider the smallest affine set containing $I$: the \emph{affine hull} $\mathrm{aff}(I)$ \cite{de2013algebraic}. Note that for every $\vec{u}\in I$, $\mathrm{aff}(I)=\vec{u}+W=\{\vec{u}+\vec{w}:\vec{w}\in W\}$, and $\mathrm{aff}(I)$ is a rational affine subspace. We next show a useful lemma about the distance from rational affine subspaces to surrounding integer points. This will be used to show there are only a finite number of strips, and will also be key to a trick in the later proof of Lemma \ref{strip extension by averaging}.

\begin{lemma}
\label{closest lattice point to affine space}
Let $A\subset\R^d$ be a rational affine subspace, containing some point $\vec{x}\in\N^d$. Then there exists a constant $c>0$ such that for any period $p^*\in\N_+$, for all $\vec{y}\in(\barx\mod p^*)$ with $\vec{y}\notin A$, the distance $\mathrm{dist}(\vec{y},A)\geq cp^*$.
\end{lemma}

\begin{proof}
Let $A$ be a rational affine subspace containing $\vec{x}\in\N^d$. Let $\vec{y}\in\barx\mod p^*$ with $\vec{y}\notin A$, so we can write $\vec{y}=\vec{x}+p^*\vec{v}$ for some $\vec{v}\in\Z^d$.

First we consider the case that $A$ is a hyperplane, so we can write $A=\{\vec{z}\in\R^d:\vec{a}\cdot\vec{z}=b\}$ for some $\vec{a}\in\Z^d$ and $b\in\Z$. Then using a standard formula \cite{CheneyLinearAlgebra} for the distance from a point $\vec{y}\in\R^d$ to $A$:
\[
\mathrm{dist}(\vec{y},A)=\frac{|\vec{a}\cdot\vec{y}-b|}{\norm{\vec{a}}}=\frac{|\vec{a}\cdot(\vec{x}+p^*\vec{v})-b|}{\norm{\vec{a}}}=\frac{p^*|\vec{a}\cdot\vec{v}|}{\norm{\vec{a}}}\geq\frac{p^*}{\norm{\vec{a}}}
\]
where the inequality follows from $|\vec{a}\cdot\vec{v}|\geq 1$: $\vec{a}\cdot\vec{v}\in\Z$ and $\vec{a}\cdot\vec{v}\neq 0$ since $\vec{y}\notin A$. The desired result then holds taking $c=1/\norm{\vec{a}}$, which depends only on $A$ (and not on $p^*$).

Finally if the rational affine space $A$ is not a hyperplane, it is the intersection of a finite number of hyperplanes, so $A$ is contained in some rational hyperplane $H$ and then by the above result $\mathrm{dist}(\vec{y},A)\geq\mathrm{dist}(\vec{y},H)\geq cp^*$.
\end{proof}

Now we can show there are only a finite number of strips in each under-determined region.

\begin{lemma}
\label{finite number of strips}
The equivalence relation $\equiv_W$ partitions $U\cap\N^d$ into a finite number of strips.
\end{lemma}

\begin{proof}
Consider the set of unique strips $\{I_1,I_2,\ldots\}$ each with some representative $\vec{u}_j\in I_j$. For each strip $I_j$, consider the affine hull $\mathrm{aff}(I_j)$, which is a rational affine space, which are all parallel. For any $I_j\neq I_k$, since both $\mathrm{aff}(I_j)$ and $\mathrm{aff}(I_k)$ contain integer points, using Lemma \ref{closest lattice point to affine space} with $p=1$ implies that $\mathrm{dist}(\mathrm{aff}(I_j),\mathrm{aff}(I_k))\geq c$ for some constant $c>0$. This lower bound $c$ is the same for all $j,k$ because the $\mathrm{aff}(I_j)$ are all parallel. The affine hulls of the strips being bounded away from each other will imply there can only be finitely many strips.

Since $U$ is a convex polyhedron, we can write it as $U=H+\rec(U)=\{\vec{h}+\vec{y}:\vec{h}\in H,\quad \vec{y}\in\rec(U)\}$, 
the sum of a bounded polytope $H$ and the recession cone \cite{de2013algebraic}. Thus each representative $\vec{u}_j=\vec{h}_j+\vec{y}_j$ for some $\vec{h}_j\in H$ and $\vec{y}_j\in\rec(U)$, so $\vec{h}_j=\vec{u}_j-\vec{y}_j\in \mathrm{aff}(I_j)$ (since $\vec{y}_j\in W$). 
Then $\mathrm{dist}(\vec{h}_j,\vec{h},k)\geq c$ for all $j\neq k$. Since all $\vec{h}_j$ are contained in the bounded polytope $H$, there must be finitely many and thus finitely many strips.
\end{proof}

Since there are a finite number of under-determined regions, by Lemma \ref{finite number of strips} there are a finite total number of strips, so it suffices to consider each strip separately and show there exists an extension from each strip that eventually dominates $f$.

We can build off Lemma \ref{projection of determined gradient is equal} for a strip $I$ to define an extension $g_I$ from $I$ that eventually dominates $f$. Intuitively, we will take the average gradient of the extensions from all determined neighbors. We crucially assume that the gradients of these extensions are not all the same. This will imply that their average grows faster than the minimum (and thus grows faster than $f$) moving away from $\mathrm{aff}(I)$. It is then immediate that $g_I$ will eventually dominate $f$ sufficiently far from $\mathrm{aff}(I)$, but requires a subtle trick to make this hold near $\mathrm{aff}(I)$.

We choose $g_I$ to have a larger period $p^*$ that will guarantee (via Lemma \ref{closest lattice point to affine space}) that points congruent to points in $I$ are sufficiently far from $\mathrm{aff}(I)$ (where $g_I\geq f$). The offsets in these congruence classes mod $p^*$ are uniquely defined so $g_I=f$ on $I$. The remaining offsets are then maximized subject to the constraint that $g_I$ is nondecreasing.

Finally, we must show $g_I$ eventually dominates $f$ on $\mathrm{aff}(I)$. This is only nontrivial for $d\geq 3$. For example if $I$ is a strip in the ``pizza slice'' shaped region 6 in Fig.~\ref{fig:3d-regions}, we must argue that an extension from $I$ will eventually dominate within the whole spanning plane. This is a generalization of Lemma \ref{determined extension eventually dominates f}, since $I$ is essentially determined within $\mathrm{aff}(I)$ (hence us calling $W$ the determined subspace of $U$).


\begin{lemma}
\label{strip extension by averaging}
Let $I$ be a strip of an under-determined eventual region $U$. Let $D_1,\ldots,D_m$ be the determined neighbors of $U$, with extensions $g_1,\ldots,g_m$. Assume for all $\vec{z}\in W^\perp$, the gradients of the extensions along $\vec{z}$ are not all equal: $\vec{\nabla}_{g_i}\cdot\vec{z}\neq\vec{\nabla}_{g_j}\cdot\vec{z}$ for some $i,j$. Then there exists an extension $g_I$ from the strip $I$ that eventually dominates $f$.
\end{lemma}

\newcommand{\nabavg}{\vec{\nabla}_{\mathrm{avg}}}

\begin{proof}
Let $D_1,\ldots,D_m$ be determined neighbors of $U$, with a quilt-affine extension from each $D_i$ $g_i(\vec{x})=\vec{\nabla}_{g_i}\cdot\vec{x}+B_{g_i}(\barx\mod p)$. Let $\nabavg=\frac{1}{m}\sum_{i=1}^m\vec{\nabla}_{g_i}\in\Q^d$. 
We will define 
$p^* \in \N_+$ and $B^*:\Z^d/p^*\Z^d \to \Q$
and let the extension be 
\[
g_I(\vec{x})=\nabavg\cdot\vec{x}+B^*(\barx\mod p^*),
\]
which is quilt-affine with a potentially larger period $p^*=kp$ for some $k\in\N_+$, so $(\barx\mod p^*)\subseteq(\barx\mod p)$.
To ensure that $g_I$ still has integer outputs, we pick $p^*$ such that $p^*\nabavg\in\N^d$. We will show later that $p^*$ can be chosen large enough to make $g_I$ eventually dominate $f$. Let $\vec{r}\in I$ be a fixed reference vector we will use to define $B^*$.

First we show that for each $\vec{u}\in I$, $B^*(\baru\mod p^*)$ is uniquely defined so that $g_I(\vec{x})=f(\vec{x})$ for all $\vec{x}\in I\cap(\baru\mod p^*)$ (in other words there are fixed values of $B^*$ for inputs in $I$ that will make $g_I$ an extension of $f$ from $I$). 
By Lemma \ref{sunshine lemma} we have affine partial function $f|_{U\cap(\baru\mod p)}(\vec{x})=\nabla_{\baru}\cdot\vec{x}+b_{\baru}$ and by Lemma \ref{projection of determined gradient is equal} we have $\mathrm{proj}_W(\vec{\nabla}_{\baru})=\mathrm{proj}_W(\vec{\nabla}_{g_i})$ for all gradients of determined neighbor extensions $g_i$. Thus we also have $\mathrm{proj}_W(\vec{\nabla}_{\baru})=\mathrm{proj}_W(\nabavg)$, so $\vec{\nabla}_{\baru}\cdot\vec{w}=\nabavg\cdot\vec{w}$ for all $\vec{w}\in W$. We then define $B^*(\baru\mod p^*)=\vec{\nabla}_{\baru}\cdot\vec{r}-\nabavg\cdot\vec{r}+b_{\baru}$, which depends only on the congruence class $\baru\mod p$ (but doesn't depend on $p^*$). We can now verify that for any $\vec{x}\in I\cap(\baru\mod p^*)$, where $\vec{x}-\vec{r}\in W$ by definition of the strip $I$, we have
\begin{align*}
    g_I(\vec{x})&=\nabavg\cdot\vec{x}+B^*(\barx\mod p^*)\\
    &=\nabavg\cdot(\vec{x}-\vec{r})+\nabavg\cdot\vec{r}+\vec{\nabla}_{\baru}\cdot\vec{r}-\nabavg\cdot\vec{r}+b_{\baru}\\
    &=\vec{\nabla}_{\baru}\cdot(\vec{x}-\vec{r})+\vec{\nabla}_{\baru}\cdot\vec{r}+b_{\baru}
    \ \ \ \ \text{since $\vec{x} - \vec{r} \in W$}
    \\
    &=\nabla_{\baru}\cdot\vec{x}+b_{\baru}=f|_{U\cap(\baru\mod p)}(\vec{x})=f(\vec{x})
\end{align*}

$g_I$ is currently a partial function, only defined on the set $I^*=(I+p^*\Z^d)\cap\N^d$ of points congruent $\mod p^*$ to some $\vec{u}\in I$. For all other congruence classes $\bara\mod p^*\in\Z^d/p^*\Z^d$ such that $\bara\cap I=\emptyset$, we will define $B^*(\bara)$ to be as large as possible while still having $g_I$ be nondecreasing. For $g_I$ to be nondecreasing, $g_I(\vec{x})\leq\min_{\vec{y}\geq\vec{x}}g_I(\vec{y})$ for all $\vec{x}\in\N^d$. We maximize $B^*(\bara)$ such that for all $\vec{x}\in\bara$,
\[
g_I(\vec{x})=\min_{\vec{y}\in I^*,\vec{y}\geq\vec{x}}g_I(\vec{y})
\]
Observe that since the finite differences above each $\vec{x}$ are periodic (as observed formally to prove Lemma \ref{CRC for quilt-affine function}), this required offset $B^*(\vec{a})$ depends only on the congruence class $\bara$ of $\vec{x}$.

Now in order to show that $g_I$ eventually dominates $f$, we claim it suffices to show that $g_I$ eventually dominates $f$ on $I^*$: for some $\vec{n}\in\N^d$, $g_I(\vec{x})\geq f(\vec{x})$ for all $\vec{x}\in I^*$ with $\vec{x}\geq\vec{n}$. If this holds, then for any $\vec{x}\notin I^*$ with $\vec{x}\geq\vec{n}$, we have $g_I(\vec{x})=g_I(\vec{y})$ for some $\vec{y}\in I^*,\vec{y}\geq\vec{x}$, so
\[
g_I(\vec{x})=g_I(\vec{y})\geq f(\vec{y})\geq f(\vec{x})
\]
showing that $g_I$ eventually dominates $f$ as long as $g_I$ eventually dominates $f$ on $I^*$.

We will next show $g_I$ eventually dominates $f$ for $\vec{x}$ sufficiently far from $\mathrm{aff}(I)$ by comparing $g_I$ to the extension $g_j(\vec{x})=\vec{\nabla}_{g_j}\cdot\vec{x}+B_{g_j}(\barx\mod p)$ from some determined neighbor $D_j$. 
Let $\vec{x}\in I^*$ and let $g_j$ be the extension of any determined neighbor $D_j$. 
Writing $\vec{x}=\vec{r}+\vec{w}+\vec{z}$ for the fixed reference $\vec{r}\in I$, $\vec{w}=\mathrm{proj}_W(\vec{x}-\vec{r})\in W$ and $\vec{z}=\mathrm{proj}_{W^{\perp}}(\vec{x}-\vec{r})\in W^{\perp}$, we have
\begin{align*}
    g_I(\vec{x})-g_j(\vec{x})&=\nabavg\cdot(\vec{r}+\vec{w}+\vec{z})+B^*(\barx\mod p^*)-\vec{\nabla}_{g_j}\cdot(\vec{r}+\vec{w}+\vec{z})-B_{g_j}(\barx\mod p)\\
    &=\nabavg\cdot\vec{z}-\vec{\nabla}_{g_j}\cdot\vec{z}+[\nabavg\cdot\vec{r}-\vec{\nabla}_{g_j}\cdot\vec{r}+B^*(\barx\mod p^*)-B_{g_j}(\barx\mod p)]
\end{align*}
Notice that the term $[\ldots]$ depends only on $j$ and $\barx\mod p$ (since $B^*$ was uniquely defined on $I^*$ based only on $\barx\mod p$). Thus minimizing over all finitely many $j$ and $\barx\mod p$ gives some (possibly negative) lower bound $-q\in\Q$ (which crucially does not depend on the choice $p^*$) such that
\[
 g_I(\vec{x})-g_j(\vec{x})\geq \nabavg\cdot\vec{z}-\vec{\nabla}_{g_j}\cdot\vec{z}-q
\]
for all $\vec{x}\in I^*$ and extensions $g_j$.

Now we use the crucial assumption that for any $\vec{z}\in W^{\perp}$,  $\vec{\nabla}_{g_i}\cdot\vec{z}\neq\vec{\nabla}_{g_j}\cdot\vec{z}$ for some $i,j$. Considering first unit vectors $\vec{v}\in W^{\perp}$ with $\norm{\vec{v}}=1$, then there is some $j$ minimizing $\vec{\nabla}_{g_j}\cdot\vec{v}$ with $\nabavg\cdot\vec{v}-\vec{\nabla}_{g_j}\cdot\vec{v}>0$. 
We claim that there exists $\epsilon>0$ such that $\nabavg\cdot\vec{v}-\vec{\nabla}_{g_j}\cdot\vec{v}\geq\epsilon$ for all such $\vec{v}$ and their corresponding $j$.
If not, then there is a sequence $(\vec{v}_i)$ of unit vectors with $\nabavg\cdot\vec{v}_i-\vec{\nabla}_{g_j}\cdot\vec{v}_i\to0$. Since the unit ball is compact, there must be a subsequence of $(\vec{v}_i)$ converging to some $\vec{v}$, which implies $\nabavg\cdot\vec{v}-\vec{\nabla}_{g_j}\cdot\vec{v}=0$.
This completes the claim that such $\epsilon>0$ exists.
Then as long as $\norm{\vec{z}}\geq q/\epsilon$, we have for some $j$ (which minimizes $\vec{\nabla}_{g_j}\cdot\vec{z}$)
\[
 g_I(\vec{x})-g_j(\vec{x})\geq\norm{\vec{z}}\qty(\nabavg\cdot\frac{\vec{z}}{\norm{\vec{z}}}-\vec{\nabla}_{g_j}\cdot\frac{\vec{z}}{\norm{\vec{z}}})-q\geq\norm{\vec{z}}\epsilon-q\geq 0
\]

Since $\vec{x}=\vec{r}+\vec{w}+\vec{z}$ for some $\vec{r}+\vec{w}\in\mathrm{aff}(I)$ and $\vec{z}\in W^{\perp}$, we have $\norm{\vec{z}}=\mathrm{dist}(\vec{x},\mathrm{aff}(I))$. Thus we have shown for $\vec{x}$ sufficiently far from $\mathrm{aff}(I)$, $g_I(\vec{x})\geq g_j(\vec{x})$ for some quilt-affine $g_j$ which itself eventually dominates $f$ (by Lemma \ref{determined extension eventually dominates f}). Crucially this bound $\norm{\vec{z}}\geq q/\epsilon$ did not depend on $p^*$, so we will now use Lemma \ref{closest lattice point to affine space} to choose $p^*$ large enough that $\mathrm{dist}(\vec{x},\mathrm{aff}(I))\geq q/\epsilon$ for all $\vec{x}\in I^*$ with $\vec{x}\notin\mathrm{aff}(I)$. Since such $\vec{x}\in(\baru\mod p^*)$ for some $\vec{u}\in I\subset\mathrm{aff}(I)$ and $\mathrm{aff}(I)$ is a rational affine subspace, by Lemma \ref{closest lattice point to affine space} there is some bound $c>0$ (depending only on $\mathrm{aff}(I)$) such that $\mathrm{dist}(\vec{x},\mathrm{aff}(I))\geq cp^*$. Thus we choose a large enough multiple $p^*=kp$ such that $cp^*\geq q/\epsilon$.

We have shown that $g_I(\vec{x})$ eventually dominates $f(\vec{x})$ for all $\vec{x}\in I^*$ with $\vec{x}\notin\mathrm{aff}(I)$. It finally remains to show that $g_I$ eventually dominates $f$ on $\mathrm{aff}(I)$. This is true for the same reasons as Lemma \ref{determined extension eventually dominates f} (because $g_I$ is a quilt-affine extension of $f$ from $I$) following the same proof strategy. In the proof of Lemma \ref{determined extension eventually dominates f} we assumed toward contradiction a sequence of ``bad points'', and used the fact that a determined region $D$ contained arbitrarily large open balls to construct a contradiction sequence for Lemma \ref{dickson contradiction lemma}. Now we are only showing $g_I$ eventually dominates $f$ on $\mathrm{aff}(I)$, so all ``bad points'' would be in $\mathrm{aff}(I)$. We can construct a contradiction sequence $(\vec{a}_i)\in I$ by the same argument. Since $W=\mathrm{span}(\rec(U))$, by the same argument as Lemma \ref{determined region has large balls}, $\rec(U)$ contains arbitrarily large open balls within the subspace $W$.  This is sufficient to ensure $(\vec{a}_i)\in I$, since the sequence of vectors $(\vec{v}_i)$ with $\vec{a}_i+\vec{v}_i$ being ``bad points'' are in $W$ because all ``bad points'' are in $\mathrm{aff}(I)$.

Thus $g_I$ eventually dominates $f$ everywhere, as desired.
\end{proof}

Lemma \ref{strip extension by averaging} crucially assumed that along any $\vec{z}\in W^\perp$, the gradients of all extensions from determined neighbors are not equal. If this does not hold, $f$ could fail to be obliviously-computable. For example, consider the function 
\begin{equation}
\label{example:depressed-strip}
f(x_1,x_2)=
\begin{cases}
x_1+x_2+1, & \text{if } x_1\neq x_2\\
x_1+x_2, & \text{if } x_1 = x_2
\end{cases}
\end{equation}
which is a single affine function, depressed by $1$ along the diagonal $x_1=x_2$. $f$ is semilinear and nondecreasing. The two determined regions where $x_1>x_2$ and $x_1<x_2$ have the same quilt-affine extension ($(1,1)\cdot\vec{x}+1$), which eventually dominates $f$. On the underdetermined region, which here consists of just the single strip where $x_1=x_2$, $f$ is strictly smaller. There does not actually exist given a quilt-affine extension from this strip that eventually dominates $f$. (One can show directly $f$ is not obliviously-computable by Lemma \ref{dickson contradiction lemma}, with $\vec{a}_i=(i,0)$ and $\vec{\Delta}_{ij}=(0,j)$).

The remaining case thus serves to disallow general versions of this counterexample. Lemma \ref{strip extension by averaging} assumed for all $\vec{z}\in W^\perp$, $\vec{\nabla}_{g_i}\cdot\vec{z}\neq\vec{\nabla}_{g_j}\cdot\vec{z}$ for some $i,j$. We will now consider the negation: that for some $\vec{z}\in W^\perp$ we have $\vec{\nabla}_{g_i}=\vec{\nabla}_{g_j}$ for all $i,j$. To proceed in this case, we will need to be able to identify the \emph{neighbor of $U$ in the direction of} $\vec{z}$.

For example, consider the under-determined eventual region 5 in Fig.~\ref{fig:3d-regions}. The determined subspace $W$ is 1D, while the orthogonal complement $W^\perp$ is a 2D. For each $\vec{z}\in W^\perp$, the neighbor in the direction of $\vec{z}$ will correspond to one of the 8 other pictured regions.




We now identify which threshold hyperplanes can distinguish a region from its neighbors. Recall the threshold hyperplanes $H_i=\{\vec{x}\in\R^d:\vec{t}_i\cdot\vec{x}=h_i\}$
for $i=1,\ldots,l$.
We now show that some of these hyperplanes must be parallel to all vectors in $\rec(U)$:

\begin{lemma}
\label{lem-neighbor-thresholds-exist}
Let $U$ be an under-determined eventual region. Then there exists some threshold hyperplane $H_i=\{\vec{x}\in\R^d:\vec{t}_i\cdot\vec{x}=h_i\}$ such that $\vec{t}_i\cdot\vec{y}=0$ for all $\vec{y}\in\rec(U)$.
\end{lemma}

\begin{proof}
Assume toward contradiction that for all $\vec{t}_i$, there exists some $\vec{y}_i\in\rec(U)$ such that $\vec{t}_i\cdot\vec{y}_i\neq 0$, so $s_i(\vec{t}_i\cdot \vec{y}_i)>0$, where $s_i$ is the $i$th sign that defined the region $U$, and $s_j(\vec{t}_j\cdot\vec{y}_i)\geq 0$ for all $j=1,\ldots,l$ since $\vec{y}_i\in\rec(U)$. Then let $\vec{y}=\sum_{i=1}^l\vec{y}_i\in\rec(U)$ since $\rec(U)$ is closed under addition, so $s_i(\vec{t}_i\cdot\vec{y})>0$ for all $i=1,\ldots,l$.

Recall that $\rec(U)=\{\vec{x}\in\R^d_{\geq 0}:(\forall i)s_i(\vec{t}_i\cdot\vec{x})\geq 0\}$ is a closed convex polyhedron that is the intersection of closed half-spaces. Then $\mathrm{int}(\rec(U))=\{\vec{x}\in\R^d_{>0}:(\forall i)s_i(\vec{t}_i\cdot\vec{x})>0\}$ is the intersection of the respective open half-spaces, and we have $\vec{y}\in\mathrm{int}(\rec(U))$. Thus because $\mathrm{int}(\rec(U))$ is nonempty, $\rec(U)$ must be full-dimensional, contradicting that $U$ is an under-determined region.
\end{proof}

We call such hyperplanes \emph{neighbor-separating hyperplanes} for reasons that will be made clear shortly. If $\vec{t}_i\cdot\vec{y}=0$ for all $\vec{y}\in\rec(U)$, we also have $\vec{t}_i\cdot\vec{y}=0$ for all $\vec{y}\in W$, so by definition $\vec{t}_i\in W^\perp$. Then let 
$L_U = \{ i \in \{1,\ldots,l\} \mid \vec{t}_i\in W^\perp \}$
be the subset of labels of all such neighbor-separating hyperplanes for $U$. 

For example, in Fig.~\ref{fig:3d-regions}, for under-determined region 5, all four hyperplanes are neighbor-separating hyperplanes. For under-determined region 6, only the pair of horizontally oriented hyperplanes are neighbor-separating.

Recalling Definition \ref{def-region}, let $S_U=\mathrm{diag}(s_1,\ldots,s_l)$ be the sign matrix that defined $U$. For $\vec{z}\in W^\perp$,
define $R_\vec{z}$, the \emph{neighbor of $U$ in the direction of} $\vec{z}$, by a related sign matrix $S_\vec{z}=\mathrm{diag}(s_1',\ldots,s_l')$, where if $i\in L_U$ and $\mathrm{sign}(\vec{t}_i\cdot\vec{z})=-s_i$, 
then let $s_i'=-s_i$, but otherwise $s_i'=s_i$ for all other $i=1,\ldots,l$. Intuitively, for all neighbor-separating hyperplanes, $R_\vec{z}$ is on the same side as the direction $\vec{z}$, but is otherwise identical to $U$.

The following lemma justifies the use of the word ``neighbor'' in the previous definition,
and further shows that such a neighbor is ``more determined'' 
(recession cone has higher dimension) than $U$.

\begin{lemma}
\label{lem-neighbor-in-direction-z}
Let $U$ be an under-determined eventual region with $W = \mathrm{span}(\rec(U))$,
let $\vec{z} \in W^\perp$,
and let the region $R_\vec{z}$ be the neighbor of $U$ in the direction of $\vec{z}$. Then $R_\vec{z}\cap\N^d$ is nonempty, $R_\vec{z}$ is a neighbor of $U$ and furthermore $\dim \rec(U)<\dim \rec(R_\vec{z})$.
\end{lemma}

\begin{proof}
We first show that $R_\vec{z}$ is a neighbor of $U$, i.e., that
\[
\rec(U)=\{\vec{x}\in\R^d_{\geq 0}:S_UT\vec{x}\geq 0\}\subset\rec(R_\vec{z})=\{\vec{x}\in\R^d_{\geq 0}:S_\vec{z}T\vec{x}\geq 0\}.
\]
Let $\vec{x}\in\rec(U)$. Then for all $i\in L_U$, $s_i\cdot\vec{x}=0$, so $s_i'\cdot\vec{x}\geq 0$, and otherwise $s_i'=s_i$, so $S_\vec{z}T\vec{x}\geq 0$.
This implies $\vec{x} \in \rec(R_\vec{z})$,
i.e., $\rec(U) \subset \rec(R_\vec{z})$,
proving the claim that $R_\vec{z}$ is a neighbor of $U$.

Next we argue that $\dim \rec(U)<\dim \rec(R_\vec{z})$.
Now similar to the proof of Lemma \ref{lem-neighbor-thresholds-exist}, for each $i\notin L_U$ there exists some $\vec{y}_i\in\rec(U)$ such that $\vec{t}_i\cdot\vec{y}_i\neq 0$, so $s_i(\vec{t}_i\cdot\vec{y}_i)>0$. Then taking $\vec{y}=\sum_{i\notin L_U}\vec{y}_i$ we have $\vec{y}\in\rec(U)$, and also $s_i(\vec{t}_i\cdot\vec{y})>0$ for all $i\notin L_U$.\footnote{
    In fact $\vec{y}$ can be shown to be in the \emph{relative interior} of $\rec(U)$, where the relative interior is the interior within the affine hull.\cite{de2013algebraic}
}

We will show for some $\epsilon>0$, $\vec{y}+\epsilon\vec{z}\in\rec(R_\vec{z})$, which will imply $\vec{z}\in\mathrm{span}(\rec(R_\vec{z}))$ so $\dim \rec(R_\vec{z})\geq\dim \rec(U)+1$ since $\vec{z}\in W^\perp$. 
Intuitively, to show this, we perturb the vector $\vec{y}$ by a slight amount in the direction of $\vec{z}$ to be on the correct side of all neighbor-separating hyperplanes, while remaining on the same side of all other hyperplanes. 
Formally,
for all $i\in L_U$, we have $s'_i(\vec{t}_i\cdot\vec{z})\geq 0$ by construction of $s'_i$. This might not hold for $i\notin L_U$, but in that case $s'_i(\vec{t}_i\cdot\vec{y})>0$. Thus we can pick some small enough $\epsilon>0$ such that $s'_i(\vec{t}_i\cdot(\vec{y}+\epsilon\vec{z}))\geq0$ for all $i\notin L_U$. For $i\in L_U$, we have $\vec{t}_i\cdot\vec{y}=0$ (by definition of $L_U$ since $\vec{y}\in\rec(U)$), so we also have $s'_i(\vec{t}_i\cdot(\vec{y}+\epsilon\vec{z}))\geq0$ and thus $\vec{y}+\epsilon\vec{z}\in\rec(R_\vec{z})$. 
This concludes the claim that $\dim \rec(U)<\dim \rec(R_\vec{z})$.

Finally, we argue that $R_\vec{z}\cap\N^d$ is nonempty, so the region $R_\vec{z}$ is meaningfully defined.
We can further assume that the vector $\vec{y}^+=\vec{y}+\epsilon\vec{z}\in\N^d$ (again by density assuming that the pieces are rational and then scaling up to clear denominators). 
Now consider a point $\vec{u}\in U\cap\N^d$, so $S_U(T\vec{u}-\vec{h})\geq0$, and consider moving along $\vec{y}^+$. For all $i$ such that $s_i\neq s'_i$, we have $s_i'(\vec{t}_i\cdot\vec{y}^+)>0$. Thus for all sufficiently large constants $c$, $S_\vec{z}(T(\vec{u}+c\vec{y}^+)-\vec{h})\geq0$ so $\vec{u}+c\vec{y}^+\in R_\vec{z}$. Intuitively, the path from $U$ along the vector $\vec{y}^+$ will eventually remain in the region $R_\vec{z}$. This shows that the neighbor of $U$ in the direction of $\vec{z}$ is well-defined.
\end{proof}


For under-determined eventual region $U$, by repeatedly applying Lemma \ref{lem-neighbor-in-direction-z} using directions $\pm\vec{z}$ for any $\vec{z}\in W^{\perp}$, we can show that determined neighbors must actually exist:

\begin{corollary}
\label{cor-determined-neighbor-exists}
An under-determined eventual region $U$ has at least $2$ determined neighbors.
\end{corollary}

We are now ready to consider the remaining case left after Lemma \ref{strip extension by averaging}, when all determined neighbor gradients agree along some $\vec{z}\in W^\perp$.

\begin{lemma}
\label{strip extension as top neighbor}
Let $I$ be a strip of an under-determined eventual region $U$. Let $D_1,\ldots,D_m$ be the determined neighbors of $U$, with extensions $g_1,\ldots,g_m$. Assume there exists $\vec{z}\in W^\perp$ such that the gradients of the extensions along $\vec{z}$ are all equal: $\vec{\nabla}_{g_i}\cdot\vec{z}=\vec{\nabla}_{g_j}\cdot\vec{z}$ for all $i,j$. Let $R_\vec{z}$ be the neighbor of $U$ in the direction $\vec{z}$, with extension extension $g_\vec{z}$. Then $g_\vec{z}$ is also an extension from $I$: $g_\vec{z}(\vec{x})=f(\vec{x})$ for all $\vec{x}\in I$, so taking $g_\vec{z}$ gives an extension from $I$ which eventually dominates $f$.
\end{lemma}

\begin{proof}
Let $I$ be a strip of under-determined eventual region $U$, with extensions $g_1,\ldots,g_m$ from determined neighbors $D_1,\ldots,D_m$ respectively. Let $\vec{z}\in W^\perp$ such that $\vec{\nabla}_{g_i}\cdot\vec{z}=\vec{\nabla}_{g_j}\cdot\vec{z}$ for all $i,j$. We will consider the regions $R_{\vec{z}}$ and $R_{-\vec{z}}$ which are the neighbors of $U$ in the directions $\vec{z}$ and $-\vec{z}$. Recall by Lemma \ref{lem-neighbor-in-direction-z} that $\dim \rec(U)<\dim \rec(R_{\pm\vec{z}})$. The proof will proceed by induction on the codimension: $d-\dim \rec(U)$. For $U$ with codimension $1$, $R_{\pm\vec{z}}$ must be determined regions with unique extensions. In general, $R_{\pm\vec{z}}$ could be under-determined, but will have lower codimension, so by the inductive hypothesis we assume $R_{\pm\vec{z}}$ have extensions which eventually dominate $f$ (considering this Lemma alongside Lemma \ref{strip extension by averaging}). Thus there exist quilt-affine extensions $g_\vec{z}$ and $g_{-\vec{z}}$ (from $R_\vec{z}$ and $R_{-\vec{z}}$) which eventually dominate $f$. Note these may have a larger period $p^*$ as used in the proof of Lemma \ref{strip extension by averaging}. We assume $g_\vec{z}$ and $g_{-\vec{z}}$ have common period $p^*$ by taking the least common multiple if necessary. Thus we can write $g_\vec{z}(\vec{x})=\vec{\nabla}_{g_\vec{z}}\cdot\vec{x}+B_{g_\vec{z}}(\barx\mod p^*)$ and $g_{-\vec{z}}(\vec{x})=\vec{\nabla}_{g_{-\vec{z}}}\cdot\vec{x}+B_{g_{-\vec{z}}}(\barx\mod p^*)$.

Now by assumption $\vec{\nabla}_{g_i}\cdot\vec{z}=\vec{\nabla}_{g_j}\cdot\vec{z}$ for any determined neighbors $D_i,D_j$. Also by Lemma \ref{projection of determined gradient is equal}, $\mathrm{proj}_W(\vec{\nabla}_{g_i})=\mathrm{proj}_W(\vec{\nabla}_{g_j})$. Thus all determined gradients agree along $\mathrm{span}(W,\vec{z})$. The regions $R_{\pm\vec{z}}$ are either determined, or their determined neighbors are among $D_1,\ldots,D_m$ (by transitivity of the neighbor relation). Regardless, we can say $\mathrm{proj}_{\mathrm{span}(W,\vec{z})}(\vec{\nabla}_{g_\vec{z}})=\mathrm{proj}_{\mathrm{span}(W,\vec{z})}(\vec{\nabla}_{g_{-\vec{z}}})$. For this proof, we will consider the affine space $A=\mathrm{aff}(I)+\mathrm{span}(\vec{z})=\{\vec{i}+\vec{w}+c\vec{z}:\vec{i}\in I,\vec{w}\in W,c\in\R\}$ containing all points reachable from $I$ by vectors in $\mathrm{span}(W,\vec{z})$. We now claim that $g_\vec{z}(\vec{x})=g_{-\vec{z}}(\vec{x})$ for all $\vec{x}\in A\cap\N^d$. The gradients along directions in $A$ ($\mathrm{span}(W,\vec{z})$) were already shown to be equal. Thus if $g_\vec{z}(\vec{x})\neq g_{-\vec{z}}(\vec{x})$ (without loss of generality $g_\vec{z}(\vec{x})< g_{-\vec{z}}(\vec{x})$), then $g_\vec{z}(\vec{y})< g_{-\vec{z}}(\vec{y})$ for all congruent $\vec{y}\in\barx\mod p^*$. However, $g_{-\vec{z}}$ is an extension of $f$ from $R_{-\vec{z}}$, so we have $g_\vec{z}(\vec{y})< f(\vec{y})$ for all $\vec{y}\in R_{-\vec{z}}\cap(\barx\mod p^*)$. This contradicts the fact that $g_\vec{z}$ eventually dominates $f$, and completes the claim that $g_\vec{z}(\vec{x})=g_{-\vec{z}}(\vec{x})$ on $A\cap\N^d$.

Thus we must have $g_\vec{z}(\vec{x})=f(\vec{x})$ for all $\vec{x}\in A\cap R_\vec{z} \cap \N^d$ (since $g_\vec{z}$ is an extension from $R_\vec{z}$) and $\vec{x}\in A\cap R_{-\vec{z}} \cap \N^d$ (since $g_\vec{z}=g_{-\vec{z}}$, the extension from $R_{-\vec{z}}$). 
We now show also that $g_\vec{z}(\vec{x})=f(\vec{x})$ for all $\vec{x}\in I$. Assume toward contradiction that $g_\vec{z}(\vec{u})\neq f(\vec{u})$ for some $\vec{u}\in I$. By Lemma \ref{sunshine lemma} we have affine partial function $f|_{U\cap(\baru\mod p)}(\vec{x})=\nabla_{\baru}\cdot\vec{x}+b_{\baru}$ and by Lemma \ref{projection of determined gradient is equal} we have $\mathrm{proj}_W(\vec{\nabla}_{\baru})=\mathrm{proj}_W(\vec{\nabla}_{g_\vec{z}})$. Then for any $\vec{x}\in I\cap(\baru\mod p^*)$, $\vec{x}=\vec{u}+\vec{w}$ for some $\vec{w}\in W$ by definition of $I$, so
\begin{eqnarray}
    g_\vec{z}(\vec{x})-f(\vec{x})
&=&
    \vec{\nabla}_{g_\vec{z}}\cdot(\vec{u}+\vec{w})+B_{g_\vec{z}}(\baru\mod p^*)-\nabla_{\baru}\cdot(\vec{u}+\vec{w})-b_{\baru} \nonumber
\\&=&
    \underbrace{ (\vec{\nabla}_{g_\vec{z}} \cdot \vec{w} - \nabla_{\baru} \cdot \vec{w})}_{=\, 0 \text{ since } \vec{w} \in W }
    +
    ( \vec{\nabla}_{g_\vec{z}} \cdot \vec{u} + B_{g_\vec{z}}(\baru\mod p^*) )
    - 
    ( \nabla_{\baru} \cdot \vec{u} + b_{\baru} ) \nonumber
\\&=&
    g_\vec{z}(\vec{u})-f(\vec{u}) \label{eq:diff-g-f-equal}
\end{eqnarray}
In other words, if $g_\vec{z}(\vec{u})\neq f(\vec{u})$, they are also unequal for all $\vec{x}$ within the strip $I$ on the entire congruence class $\baru$.

If we had $g_\vec{z}(\vec{u})<f(\vec{u})$, then $g_\vec{z}(\vec{x})<f(\vec{x})$ for all $\vec{x}\in I\cap(\baru\mod p^*)$ by~\eqref{eq:diff-g-f-equal}, which contradicts that $g_\vec{z}$ eventually dominates $f$. 

The other case is that $g_\vec{z}(\vec{u})>f(\vec{u})$, so again by~\eqref{eq:diff-g-f-equal}, we have $g_\vec{z}(\vec{x})>f(\vec{x})$ for all $\vec{x}\in I\cap(\baru\mod p^*)$. (This is the behavior of our example~\eqref{example:depressed-strip}, which will be shown to not be obliviously-computable by the following general argument). Here, similar to the proof of Lemma \ref{determined extension eventually dominates f}, we will apply Lemma \ref{dickson contradiction lemma} by creating a contradiction sequence $(\vec{a}_1,\vec{a}_2,\ldots)\in \N^d$ such that for all $i<j$ there exists some $\vec{\Delta}_{ij}\in\N^d$ with
\[
f(\vec{a}_i+\vec{\Delta}_{ij})-f(\vec{a}_i)>f(\vec{a}_j+\vec{\Delta}_{ij})-f(\vec{a}_j).
\]
To do this, we will find a sequence $(\vec{a}_1,\vec{a}_2,\ldots)\in A\cap R_\vec{z} \cap (\baru\mod p^*)$, so $g_\vec{z}(\vec{a}_i)=f(\vec{a}_i)$ for all $i$. 
We will then find another sequence $(\vec{v}_1,\vec{v}_2,\ldots)\in\N^d$
such that $\vec{a}_i+\vec{v}_i\in I\cap(\baru\mod p^*)$ for all $i$, 
implying $g_\vec{z}(\vec{a}_i+\vec{v}_i)>f(\vec{a}_i+\vec{v}_i)$.
Also, we need that for all $i<j$, $\vec{a}_i+\vec{v}_j\in A\cap R_{-\vec{z}} \cap (\baru\mod p^*)$, so $g_\vec{z}(\vec{a}_i+\vec{v}_j)=f(\vec{a}_i+\vec{v}_i)$. 
If these are both true, 
then choosing $\vec{\Delta}_{ij}=\vec{v}_j$ for $i<j$ gives
\begin{eqnarray*}
&&\mathclap{f(\vec{a}_i+\vec{\Delta}_{ij})-f(\vec{a}_i)}
\\&=&
    g_\vec{z}(\vec{a}_i+\vec{v}_j)-g_\vec{z}(\vec{a}_i)
\\&=&
    g_\vec{z}(\vec{a}_j+\vec{v}_j)-g_\vec{z}(\vec{a}_j)
    \qquad \text{since $g_\vec{z}$ is quilt-affine and $\vec{a}_i \equiv \vec{a}_j \mod p^*$}
\\&>& 
    f(\vec{a}_j+\vec{v}_j)-f(\vec{a}_j)=f(\vec{a}_j+\vec{\Delta}_{ij})-f(\vec{a}_j).
\end{eqnarray*}
Lemma \ref{dickson contradiction lemma} then tells us that $f$ is not obliviously-computable, a contradiction.
It remains to show that such sequences $\vec{a}_i \in R_\vec{z}$ and $\vec{v}_i \in \N^d$ can be found satisfying $\vec{a}_i+\vec{v}_i\in I\cap(\baru\mod p^*)$ for all $i$, 
and for all $i<j$, $\vec{a}_i+\vec{v}_j\in A\cap R_{-\vec{z}} \cap (\baru\mod p^*)$.

Now from the proof of Lemma \ref{lem-neighbor-in-direction-z}, we take the same $\vec{y}\in\rec(U)$ and perturbed $\vec{y}^+=\vec{y}+\epsilon\vec{z}\in\rec(R_\vec{z})$ that were defined in that proof.
Recall we showed for $\vec{x}\in U$, for all large enough $c$, $\vec{x}+c\vec{y}^+\in R_\vec{z}$. Likewise, we also have $\vec{y}^-=\vec{y}-\epsilon\vec{z}\in\rec(R_{-\vec{z}})$ (taking $\epsilon$ small enough to work for both $R_\vec{z}$ and $R_{-\vec{z}}$) and we can assume (by density of rationals and scaling up denominators) that $\vec{y}^+,\vec{y}^-\in\N^d$. 

Pick $c\in\N$ large enough that $\vec{u}+cp^*\vec{y}^+\in R_\vec{z}$ and $\vec{u}+cp^*\vec{y}^-\in R_{-\vec{z}}$. Then for all $i\in\N_+$, let $\vec{a}_i=\vec{u}+icp^*\vec{y}^+$ and $\vec{v}_i=icp^*\vec{y}^-$. Since $\vec{y}^+,\vec{y}^-\in\mathrm{span}(W,\vec{z})$, we have $\vec{a}_i\in A$ for all $i$, and the multiple of $p^*$ ensures all points are in $(\baru\mod p^*)$ as desired. Finally, we can check that
\[
\vec{a}_i+\vec{v}_i=\vec{u}+icp^*\vec{y}^+ + icp^*\vec{y}^- = \vec{u}+2icp^*\vec{y}\in I
\]
since $\vec{y}^+ + \vec{y}^- = 2\vec{y} \in \rec(U)$. Also, for $i<j$ we have
\[
\vec{a}_i+\vec{v}_j = \vec{u}+icp^*\vec{y}^+ + jcp^*\vec{y}^- = \vec{u}+(j-i)cp^*\vec{y}^- + icp^*\vec{y} 
\]
Note that $j-i \geq 1$,  
$\vec{u}+(j-i)cp^*\vec{y}^- \in R_{-\vec{z}}$, 
and $icp^*\vec{y}\in\rec(R_{-\vec{z}})$.
Thus
\[
    \vec{a}_i+\vec{v}_j 
= \vec{u}+(j-i)cp^*\vec{y}^- + icp^*\vec{y} \in R_{-\vec{z}}
\]
as required.

Thus Lemma \ref{dickson contradiction lemma} gives a contradiction that $f$ cannot be obliviously-computable. We reached this contradiction by assuming that $g_\vec{z}(\vec{u})\neq f(\vec{u})$ for some $\vec{u}\in I$. Thus we conclude that $g_\vec{z}(\vec{x})=f(\vec{x})$ for all $\vec{x}\in I$, so $g_\vec{z}$ is an extension from $I$ that eventually dominates $f$.
\end{proof}

For any strip $I$ of an under-determined eventual region $U$, one of the cases from Lemmas \ref{strip extension by averaging} or \ref{strip extension as top neighbor} applies to show there exists an extension from $I$ that eventually dominates $f$. There are only finitely many such strips (Lemma \ref{finite number of strips}), so alongside the unique extensions from the determined regions (Lemmas \ref{determined region has unique extension} and \ref{determined extension eventually dominates f}), we have identified a finite collection $g_1,\ldots,g_m$ of quilt-affine functions to complete the proof of Theorem \ref{MC implies min}.
}

\opt{full,final}{
\section{Comparison to continuous case}

In \cite{ContinuousMC}, the authors classified the power of output-oblivious \emph{continuous} CRNs to stably compute real-valued functions $f:\R^d_{\geq0}\to\R_{\geq0}$. 
We can generalize to also consider such functions by introducing the following natural scaling:

\begin{definition}
For a function $f:\N^d\to\N$, the \emph{$\infty$-scaling} $\hat{f}:\R_{\geq0}^d\to\R_{\geq0}$ is given by 
\opt{sub,full}{
    \[\hat{f}(\vec{z})=\lim_{c\to\infty}\frac{f(\floor{c\vec{z}})}{c}.\]
}
\opt{final}{
    $\hat{f}(\vec{x})=\lim\limits_{c\to\infty}\frac{f(\floor{c\vec{x}})}{c}.$
}
\end{definition}

Note this limit may not exist for arbitrary $f:\N^d\to\N$, but it will exist for all obliviously-computable $f$.

The next theorem 
shows that in this scaling limit, our output-oblivious function class exactly corresponds to the real-valued function class from \cite{ContinuousMC} (see Fig.~\ref{fig:2D-scaling_limit}).
\opt{final}{The proof appears in~\cite{severson2019composable}.}

\begin{theorem}
\label{thm-continuous-comparison}
If $f:\N^d\to\N$ is obliviously-computable, then the $\infty$-scaling $\hat{f}:\R_{\geq0}^d\to\R_{\geq0}$ is obliviously-computable by a continuous CRN.
Furthermore,
every function obliviously-computable by a continuous CRN
is the $\infty$-scaling of some function obliviously-computable by a discrete CRN.
\end{theorem}


    \begin{proof}
    To prove the first statement, let $f:\N^d\to\N$ be obliviously-computable. We will show the $\infty$-scaling $\hat{f}:\R_{\geq0}^d\to\R_{\geq0}$ satisfies the main classification of \cite{ContinuousMC}: that $\hat{f}$ is superadditive, positive-continuous, and piecewise rational-linear.
    
    First we prove that for any quilt-affine $g:\N^d\to\Z$, the $\infty$-scaling $\hat{g}$ is nonnegative and rational-linear. From Definition \ref{quilt-affine def}, we can express $g(\vec{x})=\vec{\nabla}_g\cdot\vec{x}+B(\barx\mod p)$ for $\vec{\nabla}_g\in\Q^d_{\geq0}$ and $B:\Z^d/p\Z^d\to\Q$.
    Then for any $\vec{z}\in\R_{\geq0}^d$,
    \[
    \hat{g}(\vec{z})=\lim\limits_{c\to\infty}\frac{\vec{\nabla}_g\cdot\floor{c\vec{z}}+B(\overline{\floor{c\vec{z}}}\mod p)}{c}
    =\vec{\nabla}_g\cdot\vec{z},
    \]
    since $B$ is bounded. Because $\vec{\nabla}_g\in\Q^d_{\geq0}$, $\hat{g}$ is nonnegative and rational-linear.
    
    Now by the eventually-min condition \eqref{defn-obv-computable-min} of Theorem \ref{thm:main result}, there exists quilt-affine $g_1,\ldots,g_m:\N^d\to\Z$ and $\vec{n}\in\N^d$ such that $f(\vec{x})=\min_k(g_k(\vec{x}))$ for all $\vec{x}\geq\vec{n}$. Then for any $\vec{z}\in\R^d_{>0}$, $\floor{c\vec{z}}\geq\vec{n}$ for large enough $c$, so
    
    \begin{equation}
    \label{eq-hatf-is-min-of-hatg}
        \hat{f}(\vec{z})
        = \lim_{c\to\infty}\frac{f(\floor{c\vec{z}})}{c}
        = \lim_{c\to\infty}\frac{\min_k(g_k(\floor{c\vec{z}}))}{c}
        = \min_k\qty(\lim_{c\to\infty}\frac{g_k(\floor{c\vec{z}})}{c})
        = \min_k(\hat{g}_k(\vec{z})),
    \end{equation}
    where we pass the limit through the $\min$ function because min is continuous. Since $\hat{g}_k(\vec{z})=\vec{\nabla}_{g_k}\cdot\vec{z}$ are all rational-linear, on the domain $\R_{>0}$, $\hat{f}(\vec{z})= \min_k(\hat{g}_k(\vec{z}))$ is continuous and piecewise rational-linear.
    
    We will now generalize this argument to show on the full domain $\R^d_{\geq0}$, 
    $\hat{f}$ is piecewise rational-linear and \emph{positive-continuous}: 
    for each subset $S\subseteq\{1,\ldots,d\}$, $\hat{f}$ is continuous on domain $D_S=\{\vec{z}\in\R^d_{\geq0}:\vec{z}(i)=0\iff i\in S\}$. 
    Fix any such $S$. By repeatedly applying the recursive condition \eqref{defn-obv-computable-3} of Theorem~\ref{thm:main result}, 
    the fixed-input restriction $f_{[(\forall i\in S)\ \vec{x}(i) \to 0]}$ fixing input coordinates in $S$ to $0$, is obliviously-computable. Then by the eventually-min condition \eqref{defn-obv-computable-min}, there exists quilt-affine $g^S_1,\ldots,g^S_m$ and $\vec{n}\in\N^d$ such that $f_{[(\forall i\in S)\ \vec{x}(i) \to 0]}(\vec{x})=\min_k(g^S_k(\vec{x}))$ for all $\vec{x}\geq\vec{n}$, but it sufficient for $\vec{x}(i)\geq \vec{n}(i)$ for all $i\notin S$. Now let $\vec{z}\in D_S$, so $\vec{z}(i)=0\iff i\in S$. Then for large enough $c$, 
    $\floor{c\vec{z}}(i)\geq\vec{n}(i)$ for all $i\notin S$, so
    \[
    f(\floor{c\vec{z}})=f_{[(\forall i\in S)\ \vec{x}(i) \to 0]}(\floor{c\vec{z}})=\min_k\qty(g^S_k(\floor{c\vec{z}})).
    \]
    Now repeating equation~\eqref{eq-hatf-is-min-of-hatg}, we have 
    $\hat{f}(\vec{z})=\min_k(\hat{g}^S_k(\vec{z}))$, so $\hat{f}$ is continuous and piecewise rational-linear on $D_S$. This holds for all $S\subseteq\{1,\ldots,d\}$, so $\hat{f}$ is positive-continous and piecewise rational linear.
    
    It remains to show that $\hat{f}$ must be superadditive: $\hat{f}(\vec{a})+\hat{f}(\vec{b})\leq \hat{f}(\vec{a}+\vec{b})$ for all $\vec{a},\vec{b}\in\R^d_{\geq0}$. Let $\vec{a},\vec{b}\in\R^d_{\geq0}$ with $\vec{a}+\vec{b}\in D_S$ for a domain $D_S$ defined as above. Then
    \[
    \hat{f}(\vec{a}+\vec{b})=\min_k(\hat{g}^S_k(\vec{a}+\vec{b}))=\hat{g}^S_i(\vec{a}+\vec{b})=\hat{g}^S_i(\vec{a})+\hat{g}^S_i(\vec{b})
    \]
    for some minimizing rational-linear $\hat{g}^S_i$. It remains to show $\hat{g}^S_i(\vec{a})\geq\hat{f}(\vec{a})$ (and by symmetry $\hat{g}^S_i(\vec{b})\geq\hat{f}(\vec{b})$). This is immediate if $\vec{a}\in D_S$ since $\hat{f}(\vec{a})=\min_k(\hat{g}^S_k)$. Otherwise if $\vec{a}\notin D_S$, assume toward contradiction that $\hat{f}(\vec{a})>\hat{g}^S_k(\vec{a})$. Then for some small enough $\epsilon>0$, we also have $\hat{f}(\vec{a})>\hat{g}^S_k(\vec{a}+\epsilon\vec{b})$. Observing that $\vec{a}+\epsilon\vec{b}\in D_S$, then $\hat{g}^S_k(\vec{a}+\epsilon\vec{b})\geq\hat{f}(\vec{a}+\epsilon\vec{b})$. But then $\hat{f}(\vec{a})>\hat{f}(\vec{a}+\epsilon\vec{b})$, a contradiction since $\hat{f}$ must be nondecreasing as the $\infty$-scaling of the nondecreasing function $f$.
    
    Thus $\hat{f}$ is semilinear, positive-continuous, and piecewise rational-linear as desired. 
    
    Next, to prove the second statement, let $\hat{f}:\R_{\geq0}^d\to\R_{\geq0}$ be any semilinear, positive-continuous, and piecewise rational-linear function. 
    We will show that there exists some obliviously-computable
    $f:\N^d\to\N$ such that its $\infty$-scaling is $\hat{f}$. 
    
    On each domain $D_S=\{\vec{z}\in\R^d_{\geq0}:\vec{z}(i)=0\iff i\in S\}$ for $S\subseteq\{1,\ldots,d\}$, $\hat{f}|_{D_S}$ is superadditive, continuous, and piecewise rational-linear. By Lemma~8 in \cite{ContinuousMC}, $\hat{f}|_{D_S}$ can be written as the minimum of a finite number of rational linear functions $\hat{g}^S_k(\vec{z})=\vec{\nabla}_{g_k}\cdot\vec{z}$. For each $\hat{g}^S_k$, we will identify a quilt-affine $g^S_k$ with gradient $\vec{\nabla}_{g_k}$. 
    In particular, we can 
    define $g^S_k:\N^d \to \N$ for all $\vec{x} \in \N^d$ by 
    $g^S_k(\vec{x})=\floor{\vec{\nabla}_{g_k}\cdot\vec{x}}$, which will be quilt-affine.
    
    Now for all $S$ and integer $\vec{x}\in D_S\cap\N^d$, define $f(\vec{x})=\min_k(g^S_k(\vec{x}))$. From the above proof it follows that $\hat{f}$ is the $\infty$-scaling of $f$. It is also straightforward to verify that $f$ is obliviously-computable by satisfying Theorem~\ref{thm:main result}. $f$ is nondecreasing, satisfying condition~(\ref{defn-obv-computable-1}), because $\hat{f}$ was semilinear and thus nondecreasing. $f$ satisfies eventually-min condition~(\ref{defn-obv-computable-min}) since for all $\vec{x}\geq(1,\ldots,1)$, $\vec{x}\in D_\emptyset=\R^d_{>0}$, so $f(\vec{x})=\min_k(g^{\emptyset}_k(\vec{x}))$. For all other $S\neq\emptyset$, the fixed-input restriction $f_{[(\forall i\in S)\ \vec{x}(i) \to 0]}(\vec{x})=\min_k(g_k^S(\vec{x}))$. It follows that $f$ satisfies recursive condition~(\ref{defn-obv-computable-3}), because any fixed-input restriction will be eventually-min of quilt-affine functions.
    
    Thus any function $\hat{f}$ obliviously-computable by a continuous CRN is the $\infty$-scaling limit of some $f$ obliviously-computable by a discrete CRN.
    \end{proof}




}
\opt{final}{\section{Leaderless one-dimensional case}
\label{sec:leaderless}
In this section we show a characterization of 1D functions $f:\N\to\N$ that are obliviously-computable \emph{without} a leader.
The general case for leaderless oblivious computation in higher dimensions remains open.


Note that the following observation applies to any number of dimensions. We say $f:\N^d\to\N$ is \emph{superadditive} if 
$f(\vec{x})+f(\vec{y}) \leq f(\vec{x}+\vec{y})$ 
for all $\vec{x},\vec{y}\in\N^d$.

\begin{observation}
\label{MC0 is superadditive}
Every $f$ obliviously-computable by a leaderless CRN is superadditive.
\end{observation}

\begin{proof}
    Let $\mathcal{C}$ be a leaderless CRN stably computing $f$.
    We prove the observation by contrapositive.
    Suppose $f$ is not superadditive.
    Then there are $\vec{x},\vec{z} \in \N^d$ such that 
    $f(\vec{x}) + f(\vec{z}) > f(\vec{x} + \vec{z})$.
    Recall $\vec{I}_\vec{w}$ is the initial configuration of $\mathcal{C}$ representing input $\vec{w}$.
    Let $\alpha_\vec{x}$ be a sequence of reactions applied to $\vec{I}_\vec{x}$ to produce $f(\vec{x})$ copies of $Y$,
    and let $\alpha_\vec{z}$ be a sequence of reactions applied to $\vec{I}_\vec{z}$ to produce $f(\vec{z})$ copies of $Y$.
    
    Since $\mathcal{C}$ is leaderless, 
    $\vec{I}_{\vec{x}+\vec{z}} = \vec{I}_\vec{x} + \vec{I}_\vec{z}$.
    Thus we can apply $\alpha_\vec{x}$ to $\vec{I}_{\vec{x}+\vec{z}}$,
    followed by $\alpha_\vec{z}$,
    producing $f(\vec{x}) + f(\vec{z})$ copies of $Y$.
    Since this is greater than $f(\vec{x}+\vec{z})$,
    to stably compute $f$,
    $\mathcal{C}$ must have a reaction consuming $Y$, so it is not output-oblivious.
    Since $\mathcal{C}$ was arbitrary, $f$ cannot be obliviously-computable. 
\end{proof}

This added condition of superadditivity gives us the 1D leaderless characterization.
\opt{final}{The proof appears in~\cite{severson2019composable}.}

\begin{theorem}
\label{1d MC0 theorem}
    For any $f:\N\to\N$,
    $f$ is obliviously-computable by a leaderless CRN $\iff f$ is semilinear and superadditive.
\end{theorem}

\opt{full}{

\begin{proof}
    $\implies:$ By Lemma~\ref{Semilinear Are Stably Computable Lemma} and Observation~\ref{MC0 is superadditive}.
    
    $\impliedby:$ If $f$ is superadditive, then $f$ is also nondecreasing (since $f(x+1)\geq f(x)+f(1)\geq f(x)$). 
    Then as in the Proof of Theorem \ref{1d MC theorem}, $f$ is eventually quilt-affine, so there exist $n\in\N$, period $p\in\N_+$, and finite differences $\delta_{\overline{0}},\ldots,\delta_{\overline{p-1}}\in\N$, such that for all $x\geq n$,  $f(x+1)-f(x)=\delta_{(\overline{x}\mod p)}$. Also, without loss of generality assume $p$ divides $n$, so $\overline{n}\mod p=\overline{0}$.
    
    The new CRN construction is motivated by trying to simply remove the leader species $L$ from the construction used in Theorem \ref{1d MC theorem}. Recall that set of reactions was 
        \begin{align*}
            L&\to f(0)Y+L_0
            \\
            L_i+X&\to[f(i+1)-f(i)]Y+L_{i+1}\qquad\text{for all }i=0,\ldots,n-2
            \\
            L_{n-1}+X&\to[f(n)-f(n-1)]Y+P_{\overline{n}}
            \\
            P_{\overline{a}}+X&\to\delta_{\overline{a}} Y+P_{\overline{a+1}}
            \qquad\text{for all }\overline{a}=\overline{0},\ldots,\overline{p-1}. 
        \end{align*}
    Since $f$ is superadditive, we must have $f(0)=0$. We then remove the species $L$ and $L_0$, and the two reaction that contain them, and add the first reaction
    \[
    X \to f(1)Y+L_{1}
    \]
    If this reaction only occurred once, this would still correctly compute $f$. Otherwise, however, there will be multiple ``auxiliary leader species'' from $\{L_1,\ldots,L_{n-1},P_{\overline{0}},\ldots,P_{\overline{p-1}}\}$ in the system. To correctly compute $f$, we must introduce pairwise reactions between these species that reduce the count of auxiliary leaders and add a corrective difference. 
    
    For all $i,j\in\{1,\ldots,n-1\}$, the reaction between $L_i$ and $L_j$ is
    \[
    L_i+L_j\to D_{i,j}Y+
    \begin{cases}
    L_{i+j} \text{ if }i+j<n\\
    P_{\overline{i+j}} \text{ if }i+j\geq n
    \end{cases}
    \]
    where $D_{i,j}:=f(i+j)-f(i)-f(j)\geq 0$ by superadditivity, and is the difference between how much output $Y$ was released in the reactions that produced $L_i$ and $L_j$ and how much should have been produced from the input that led to $L_i$ and $L_j$.
    
    We have similar reactions between $L_i$ and $P_j$ for all $i\in\{1,\ldots,n-1\}$ and $\overline{a}\in\{\overline{0},\ldots,\overline{p-1}\}$
    \[
    L_i+P_{\overline{a}}\to D_{i,\overline{a}}Y+P_{\overline{i+a}}
    \]
    where $D_{i,\overline{a}}:=f(i+n+a)-f(i)-f(n+a)\geq 0$ by superadditivity. The reaction sequences that produced $L_i$ consumed $i$ copies of input $X$, and those that produced $P_{\overline{a}}$ consumed $n+a+kp$ for some $k\in\N$, so we have undercounted by $f(i+n+a+kp)-f(i)-f(n+a+kp)=D_{i,\overline{a}}$ since the periodic differences cancel.
    
    Finally, the reactions between $P_{\overline{a}}$ and $P_{\overline{b}}$ for all $\overline{a},\overline{b}\in\{\overline{0},\ldots,\overline{p-1}\}$ are
    \[
    P_{\overline{a}}+P_{\overline{b}}\to D_{\overline{a},\overline{b}}Y+P_{\overline{a+b}}
    \]
    where $D_{\overline{a},\overline{b}}:=f(n+a+n+b)-f(n+a)-f(n+b)\geq 0$ by superadditivity, and this gives the corrective difference in output by a similar argument.
    
    Note that the rest of the reactions used in Theorem \ref{1d MC theorem} are not strictly necessary, since if all input $X$ undergoes the first reaction $X\to f(1)Y+L_1$, the corrective difference reactions will then reduce the count of auxiliary leader species down to $1$, while outputting the correct differences to produce precisely $f(x)$ output. 
\end{proof}
}

}
\section{Conclusion}
\label{sec:conclusion}

An obvious question is the computational power of output-oblivious CRNs without an initial leader.
A leaderlessly-obliviously-computable function must be superadditive, which is a strictly stronger condition than being nondecreasing. The continuous result \cite{ContinuousMC} had the same restriction of superadditivity, so our ``scaling limit'' reduction to their function class (Theorem~\ref{thm-continuous-comparison}) shows our main function class is already ``almost superadditive.'' We also showed in the 1D case, $f:\N\to\N$ is leaderlessly-obliviously-computable if and only if $f$ is semilinear and superadditive (Theorem \ref{1d MC0 theorem}).

Does adding the additional constraint of superaddivity to our full result (Theorem \ref{thm:main result}) classify leaderlessly-obliviously-computable $f:\N^d\to\N$? 
If this were true, a proof would require modifying our construction (Section \ref{Construction}) to eliminate the leader $L$. 
We successfully modified the 1D construction (Theorem \ref{1d MC theorem}) to remove the leader in proving Theorem \ref{1d MC0 theorem}, but it has been difficult to extend the same ideas to our much more complicated general construction.

An initial leader can also help make computation faster~\cite{kosowski2018brief, angluin2006fast, belleville2017hardness}. Many recent results in population protocols have shown time upper and lower bounds for computational tasks such as leader election and function/predicate computation~\cite{kosowski2018brief, gasieniec2018fast, LeaderElectionDIST, belleville2017hardness, alistarh2018space, alistarh2017timespace}.
These techniques, however, are not at all designed to handle the constraint of output-obliviousness.
It would be interesting to study how this constraint affects the time required for computation.

\paragraph{Acknowledgements.}
We thank Anne Condon, Cameron Chalk, Niels Kornerup, Wyatt Reeves, and  David Soloveichik for discussing their related work with us and contributing early ideas.

\opt{sub}{\newpage}
\bibliography{references,refs}
\bibliographystyle{plain}

\end{document}